\documentclass[aos,preprint]{imsart}

\RequirePackage{amsthm,amsmath,amsfonts,amssymb}
\RequirePackage[numbers]{natbib}
\RequirePackage[colorlinks,citecolor=blue,urlcolor=blue]{hyperref}
\RequirePackage{graphicx}

\startlocaldefs
\theoremstyle{plain}

\newtheorem{theorem}{Theorem}
\newtheorem{lemma}{Lemma}
  
\newtheorem{proposition}{Proposition}
\theoremstyle{remark}
\newtheorem{remark}{Remark}
\newtheorem{definition}{Definition}
\newtheorem{example}{Example}[section]

\newtheorem{assumption}{Assumption}
\usepackage[ruled]{algorithm2e} 
\graphicspath{{./Figure/}}

\endlocaldefs

\begin{document}

\begin{frontmatter}
\title{A generalized e-value feature detection method \\with FDR control at multiple resolutions}
\runtitle{A generalized e-value feature detection method}

\begin{aug}
\author[A,C]{\fnms{Chengyao}~\snm{Yu}\ead[label=e1]{12532239@mail.sustech.edu.cn}},
\author[A]{\fnms{Ruixing}~\snm{Ming}\ead[label=e2]{rxming@zjgsu.edu.cn}}
\author[A]{\fnms{Min}~\snm{Xiao}\ead[label=e3]{xiaomin@zjgsu.edu.cn}}
\author[B]{\fnms{Zhanfeng}~\snm{Wang}\ead[label=e4]{zfw@ustc.edu.cn}}
\and
\author[C]{\fnms{Bingyi}~\snm{Jing}\ead[label=e5]{jingby@sustech.edu.cn}}
\address[A]{School of Statistics and Mathematics, Zhejiang Gongshang University, \\ \printead[presep={\ }]{e1,e2,e3}}

\address[B]{School of Management, University of Science and Technology of China\printead[presep={,\ }]{e4}}

\address[C]{Department of Statistics and Data Science, Southern University of Science and Technology\printead[presep={,\ }]{e5}}
\end{aug}

\begin{abstract}
Multiple resolutions arise across a range of explanatory features due to domain-specific structures, leading to the formation of feature groups. It follows that the simultaneous detection of significant features and groups aimed at a specific response with false discovery rate (FDR) control stands as a crucial issue, such as the spatial genome-wide association studies. Nevertheless, existing detection methods with multilayer FDR control generally rely on valid p-values or knockoff statistics, which can be not flexible, powerful and stable in several settings. 
To fix this issue effectively, this article develops a novel method of Stabilized Flexible E-Filter Procedure (SFEFP), by constructing unified generalized e-values, leveraging a generalized e-filter, and adopting a stabilization treatment with power enhancement.
This method flexibly incorporates diverse base detection procedures at different resolutions to provide consistent, powerful, and stable results, while controlling FDR at multiple resolutions simultaneously.
Statistical properties of multilayer filtering procedure encompassing one-bit property, multilayer FDR control, and stability guarantee are established.
We also develop several examples for SFEFP such as the eDS-filter. Simulation studies and the analysis of HIV mutation data demonstrate the efficacy of SFEFP.
\end{abstract}



\end{frontmatter}

\section{Introduction}\label{section1}
Identifying explanatory features with significant effects on a specific response variable from a large set of candidates is a critical issue in various scientific domains, where features often possess multi-resolution structures arising from domain-specific characteristics, allowing their partition into several meaningful groups through different schemes (multiple distinct resolutions)~\cite{P-filter,Ramdas2019A,katsevich2019multilayer,gablenz2023catch}. For domain scientists, findings across different resolutions are often of substantial interest, where the detection of important features naturally facilitates discoveries at multiple resolutions, since the group encompassing it may also carry scientific significance. 
For example, with genome-wide association studies (GWAS), the primary objective is to identify both the SNPs that influence the level of a medically relevant phenotype and the genes that harbor these SNPs. Ensuring the reproducibility of findings across multiple resolutions requires rigorous control of false discoveries at each resolution. The false discovery rate (FDR)~\cite{Benjamini1995} offers an appropriate statistical assessment of reproducibility in large-scale multiple testing issues. Nevertheless, achieving simultaneous FDR control across multiple resolutions is a non-trivial task. A selected group frequently includes multiple selected features, and the counts of falsely rejected groups and features tend to be comparable—resulting in inflated FDR control at the group level. Two-stage controlled selection: the selection of relevant groups and then the location of features was developed in the literature ~\cite{benjamini2014selective,peterson2016many,yekutieli2008hierarchical}. However, ensuring coherence between discoveries at the group and individual levels can be difficult, particularly when the layers do not conform to a hierarchy.

Several approaches have been proposed to address the coherence between discoveries at different resolutions and provide multilayer FDR control, such as the p-filter~\cite{P-filter,Ramdas2019A} based on the p-value and the multilayer knockoff filter (MKF)~\cite{katsevich2019multilayer} integrating the knockoff framework~\cite{Barber2015,dai2016knockoff,Candes2018}. 
But, constructing valid p-values for the p-filter in high-dimensional settings can be challenging, and handling dependence by reshaping p-values can result in reduced power. For the MKF method, its multilayer FDR control makes the procedure conservative, due to the decoupling of dependencies among cross-layer knockoff statistics.
When features are highly correlated, the knockoff-based procedures may suffer severe power loss~\cite{spector2022powerful, dai2023scale}. In addition,
joint distribution of features required by the model-X knockoff is generally intractable to estimate. Although there are several alternatives to tackle high correlations and high-dimensional settings~\cite{Dai2022,dai2023scale,Xing2023,wang2025adaptive}, these methods focus on providing valid FDR control for single resolution. Extensions to multiple resolutions remain unexplored. Hence, it is essential to develop a unified feature detection framework that enables users to select the state-of-the-art methods for each resolution and achieves multilayer FDR control, not necessarily relying on a specific method based on p-values or knockoffs.

Methods based on e-values are  promising alternatives to p-value due to their simplicity in construction and integration ~\cite{ grunwald2020safe, howard2020time, shafer2021testing, shafer2019game, Wang2022, ren2023derandomizing, gablenz2023catch}. 
For example, Wang and Ramdas~\cite{Wang2022} introduced an analog of the BH procedure~\cite{Benjamini1995} for e-values called e-BH procedure, and suggested that an analog of the p-filter for e-values could also be developed. Using the e-BH procedure and developing one-bit knockoff e-values, Ren et al.~\cite{ren2023derandomizing} proposed the derandomizing knockoff procedure. Gablenz et al.~\cite{gablenz2023catch} further develop the e-filter and leverage the one-bit knockoff e-values to give an extension of the knockoff procedure, called e-MKF. Although the e-MKF enhances the power of MKF~\cite{katsevich2019multilayer} and guaranties multilayer FDR control, we reveal that the e-filter with one-bit e-values (e-MKF) may suffer zero power by demonstrating its one-bit property. Furthermore, to flexibly select different methods at different resolutions, it necessitates the development of generalized e-filter and unified constructions of generalized e-values (e.g. the asymptotic e-value) to incorporate situations beyond the knockoff procedure.

To fill in the gaps existing in the above methods, we propose a stabilized flexible e-filter procedure (Stabilized FEFP, or simply, SFEFP), which enables the use of different feature detection techniques at various layers to output a stable and powerful selection set with multilayer FDR control. Guided by this, this paper provides several contributions as follows. 
First, we develop a generalized e-filter via explicitly defining the generalized e-values, and further propose a unified construction of generalized e-values, based on which FEFP is developed.
In contrast to \cite{P-filter,Ramdas2019A,gablenz2023catch}, such a generalization and construction provides the flexibility to directly leverage different (state-of-the-art) detection methods such as \cite{Barber2015,dai2016knockoff,Candes2018,Dai2022,dai2023scale,Xing2023,wang2025adaptive} for each layer. The proposed construction of generalized e-values is more generic and powerful compared with existing works~\cite{ren2022derandomized, banerjee2023harnessing, li2024note, ignatiadis2024compound}; for a detailed comparison, see Remark~\ref{conpound_evalue} in Section 3.
Second, we develop the SFEFP by merging generalized e-values, a much more powerful and stable procedure compared with FEFP.  Different from \cite{gablenz2023catch}, SFEFP tackles the zero-power dilemma of the e-filter with input of one-bit, significantly improving the detection power. Also different from derandomization for single resolution~\cite{ren2023derandomizing} where stability comes with the cost of power, we provide a stabilized selection set with multilayer FDR control and enhanced power via constructing non-binary generalized e-values that better reconcile discrepancies across layers; a more detailed comparison see Section C of the Supplementary Material. Third, we investigate the statistical theories of the multilayer filtering procedures and several applications of SFEFP. Specifically, we explore the one-bit property of FEFP, revealing the output characteristics of the generalized e-filter under one-bit input. We further provide multilayer FDR control assurances for FEFP and SFEFP, and establish the stability of SFEFP under finite samples.
For SFEFP applications, we develop a novel method called the eDS-filter by extending the DS and MDS method~\cite{Dai2022} to multiple resolutions, providing a powerful alternative to MKF~\cite{katsevich2019multilayer} and e-MKF~\cite{gablenz2023catch} in highly correlated settings. The multilayer FDR control of the eDS-filter is also studied. Applications of SFEFP to
\cite{Xing2023,wang2025adaptive} are provided in Section F of the Supplementary Material. Lastly, we also extend the methods and theories to broader settings including scenarios with prior knowledge by providing a unified version of generalized e-filter, see Section E of the Supplementary Material.

The remainder of this paper is organized as follows. Section~\ref{setup} rigorously frames the problem of feature detection with multilayer FDR control and introduces necessary preliminaries. Section~\ref{section3_method} introduces the proposed SFEFP and relevant theoretical results.
In Section~\ref{section4_application}, we propose eDS-filter and eDS+gKF-filter as examples of SFEFP.
Other applications and possible extensions of SFEFP are also discussed.
Section~\ref{section5_simulation} presents the performance of our proposed methods on simulated data.
In Section~\ref{section5}, the eDS-filter is applied to analyze HIV mutation data.
Section~\ref{section6} discusses the applicable scenarios and the future work of SFEFP. 
The proofs of all propositions, theorems, and lemmas in the subsequent sections of the paper are provided in the supplementary material.

\section{Preliminaries}\label{setup}
\subsection{Problem setup}\label{section2.1}
Denote the response variable by $Y$ and the set of features by $(X_{1}, \ldots, X_{N})$. 
For the response variable and features, we denote the $n$ i.i.d samples by $(\boldsymbol{y}, \boldsymbol{X})$, where $\boldsymbol{y} \in \mathbb{R}^{n}$ is a vector of the responses and $\boldsymbol{X} \in \mathbb{R}^{n \times N}$ refers to a known design matrix. 
The feature detection issue is formalized as a multiple hypothesis testing problem \cite{Candes2018} that
\[
H_{j}: Y \perp \!\!\! \perp X_{j}| X_{-j},~j\in [N],
\]
where $X_{-j}=\{X_{1}, \dots, X_{N}\} \setminus \{X_{j}\}$ and $[N] = \{1,\dots,N\}$.
Intuitively, the hypothesis $H_{j}$ implies that $X_{j}$ does not provide additional information about $Y$ given the other features. 
If $H_{j}$ is not true, the feature $X_{j}$ is then considered to be relevant to $Y$. 
The set of relevant features is denoted as $\mathcal{H}_{1} =\{j: H_{j}~\text{is not true}\}$, and the set of irrelevant features by $\mathcal{H}_{0}=[N] \setminus \mathcal{H}_{1}$.

We further consider a scenario where the set of features is interpreted at $M$ different resolutions of inference. 
Specifically, for each $m\in [M]$, the features are divided into $G^{(m)}$ groups at the $m$-th layer, denoted by $\{\mathcal{A}_{g}^{(m)}\}_{g \in [G^{(m)}]}$. 
Subsequently, let the features in the $g$th group at layer $m$ be $X_{\mathcal{A}_{g}^{(m)}} = \{X_{j}: j\in \mathcal{A}_{g}^{(m)}\}$. 
The partition $\{\mathcal{A}_{g}^{(m)}\}_{g \in [G^{(m)}]}$ should be predetermined based on our specific objectives.
Let the function $h(m, j)$ represent the group which $X_{j}$ belongs to at the $m$-th layer. The group detection at layer $m$ is represented as
\[
H_{g}^{(m)}: Y \perp \!\!\! \perp X_{\mathcal{A}_{g}^{(m)}}| X_{-\mathcal{A}_{g}^{(m)}},~g\in [G^{(m)}],
\]
where $X_{-\mathcal{A}_{g}^{(m)}} = \{X_{j}: j\in [N]\setminus \mathcal{A}_{g}^{(m)}\}$.
We assume that~\cite{katsevich2019multilayer}
\[
H_{g}^{(m)}=\bigwedge_{j\in \mathcal{A}_{g}^{(m)}}H_{j}, \quad  g\in [G^{(m)}].
\] 
Within this assumption, the index set of irrelevant groups at layer $m$, $\mathcal{H}_{0}^{(m)}$, is given as
\[
\mathcal{H}_{0}^{(m)} = \{g \in [G^{(m)}]: \mathcal{A}_{g}^{(m)}  \subset \mathcal{H}_{0} \}.
\]
Thus, given the selected feature set $\mathcal{S} \subset [N]$, the set of selected groups in the layer $m$, $\mathcal{S}^{(m)}$, is defined as 
\[
\mathcal{S}^{(m)} = \{g \in [G^{(m)}]: \mathcal{A}_{g}^{(m)} \cap \mathcal{S} \neq \emptyset \}.
\]
The false discovery rate at the $m$-th layer is defined as
\[
\text{FDR}^{(m)}=\mathbb{E}[\text{FDP}^{(m)}], \quad \text{FDP}^{(m)} = \frac{\left| \mathcal{S}^{(m)} \cap \mathcal{H}_{0}^{(m)}   \right|}{\left| \mathcal{S}^{(m)} \right|  \vee 1},
\]
where $\left| \cdot \right|$ measures the size of a set. 
Our goal is to determine the largest set $\mathcal{S}$ such that $\text{FDR}^{(m)}$ remains below the predefined level $\alpha^{(m)}$ for all $m$ simultaneously.

\subsection{Recap: Multiple testing with e-values}\label{e-intro}
The main component of our proposed method is the e-value \cite{ grunwald2020safe, howard2020time, shafer2021testing, shafer2019game}, a measure of the strength of the evidence against the null hypothesis.
Formally, an e-variable $E$ is a nonnegative random variable with an expected value of at most 1 under the null hypothesis, say, $\mathbb{E}_{\text{null}}[E]\leq 1$, and an e-value $e$ is a one-time realization of an e-variable. For simplicity, the two concepts are not distinguished in this paper.
Compared with the p-value, the e-value is simpler to construct and exhibits greater robustness since it depends solely on the expectation.
For any significance level $\alpha\in(0,1)$, the null hypothesis is rejected if $e\geq 1/\alpha$.
 Type I error control is ensured by the Chebyshev inequality along with the definition of e-variable, say, $\mathbb{P}_{\text{null}}(E \geq \frac{1}{\alpha}) \leq \alpha \mathbb{E}_{\text{null}}(E) \leq \alpha$. Ramdas and Manole~\cite{ramdas2023randomized} studied a more precise method to control type I error, resulting in increased testing power when using e-values.

For multiple tests with the e-values, Wang and Ramdas~\cite{Wang2022} propose the e-BH procedure.
Specifically, assuming that each hypothesis $H_j$ is associated with an e-value $e_j$, the e-BH procedure first ranks e-values in a descending order $e_{(1)}\geq\dots\geq e_{(N)}$ and then rejects the hypotheses with their corresponding e-values exceeding $N/(\alpha\widehat{k})$, where $\widehat{k}=\max\{k\in[N]:e_{(k)}\geq N/(\alpha k)\}$. 
It has been demonstrated that the condition
$\sum_{j\in \mathcal{H}_0}\mathbb{E}[e_j]\leq N$ is sufficient to guarantee the FDR control~\cite{Wang2022} and e-values satisfying this condition are called relaxed e-values.
There are several literatures providing specific constructions of such relaxed e-values''~\cite{ren2022derandomized, li2024note, ignatiadis2024compound,banerjee2023harnessing}.
For multiple resolutions, Gablenz et al.~\cite{gablenz2023catch} developed the e-filter and leveraged the one-bit knockoff e-values~\cite{ren2022derandomized} to develop the e-MKF.

\section{Methodology}\label{section3_method}
This section presents the FEFP method and SFEFP method. The development of FEFP relies on the construction of the generalized e-values and generalized e-filter.

\subsection{FEFP: Flexible E-Filter Procedure}\label{section2.2}
This subsection introduces the FEFP in the following steps.
First, we summarize a framework for recent controlled feature or group detection procedures. Second, we provide the construction of generalized e-values for evaluating feature or group importance as indicated by the detection procedure within the proposed framework. Third, we describe how the generalized e-values are leveraged to produce a selected set. Lastly, we propose the multilayer FDR control theory for FEFP.
A high-level view of FEFP is presented in Algorithm~\ref{alg1}.

\subsubsection*{Framework for controlled detection procedures} We first give the definition of the framework for controlled feature or group detection procedures.
\begin{algorithm}[t]\label{alg1}
	\caption{FEFP: a flexible e-filter procedure for feature detection}
	\KwIn {Data $(\mathbf{X},\mathbf{y})$; target FDR level vector ($\alpha^{(1)}, \ldots, \alpha^{(M)})$; original FDR level vector $(\alpha_{0}^{(1)}, \ldots, \alpha_{0}^{(M)})$; partition $m$ given by $\left\{ \mathcal{A}_{1}^{(m)}, \ldots, \mathcal{A}_{G^{(m)}}^{(m)} \right\}$, $m \in [M].$}
	
	\For {$m=1, \ldots, M$}{
		Separately process any fixed detection procedure $\mathcal{G}^{(m)} \in \mathcal{K}_{\text{finite}} \bigcup \mathcal{K}_{\text{asy}}$ with level $\alpha_{0}^{(m)}$.
		
		Compute the generalized e-value $e_{g}^{(m)}= \frac{G^{(m)}} {\widehat{V}_{\mathcal{G}^{(m)}(\alpha_{0}^{(m)})} \bigvee \alpha_{0}^{(m)}} \cdot \mathbb{I}\left \{g \in \mathcal{G}^{(m)}(\alpha_{0}^{(m)}) \right \}$, $g\in[G^{(m)}]$.}
        Compute candidate set
            $\mathcal{S}_{\text{init}}^{(m)}=\left\{g\in G^{(m)}: \mathcal{A}_{g}^{(m)}\bigcap \left[\bigcap_{l=1}^{M}\left\{j:e_{h(l,j)}^{(l)}>0\right\}\right]\neq \emptyset\right\}$ for $m\in[M]$.
            
    \If{$\bigcap_{m=1}^M\left\{ \widehat{V}^{(m)}_{\mathcal{G}^{(m)}(\alpha_0^{(m)})}\leq \alpha^{(m)}|\mathcal{S}^{(m)}_{\text{init}}| \right\}$ occurs}
    {The selected set is $\mathcal{S}= \mathcal{S}\left(G^{(1)}/\widehat{V}_{\mathcal{G}^{(1)}(\alpha_{0}^{(1)})},\dots,G^{(M)}/\widehat{V}_{\mathcal{G}^{(M)}(\alpha_{0}^{(M)})}\right)$.}
    \Else{
    No features are selected, i.e., $\mathcal{S}=\emptyset$.
    }
	\KwOut{the selected set $\mathcal{S}$. }
\end{algorithm}

\begin{definition}\label{definition1}
	Consider $N$ features (groups). For any $\alpha \in (0,1)$, a feature (group) detection procedure $\mathcal{G}: \mathbb{R}^{n\times (N+1)} \to 2^{[N]}$ determines the rejection threshold $t_{\alpha}$ by
	\begin{equation}\label{FDR_framework}
	t_{\alpha}= \sideset{}{}{\arg\max}_{t}^{}R(t), \quad \text{subject to}~\widehat{\text{FDP}} (t) = \frac{\widehat{V}(t)\vee \alpha} {R(t)\vee 1} \leq \alpha, 
	\end{equation}
	where $\widehat{V}(t)$ and $R(t)$ represent the estimated number of false rejections (estimate of $V(t)$) and the number of rejections under the threshold $t$, respectively. If the procedure $\mathcal{G}$ controls the FDR under finite samples by ensuring $\mathbb{E}[V(t_{\alpha}) \slash \max\{\widehat{V}(t_{\alpha}),\alpha\}] \leq 1$, then $\mathcal{G} \in \mathcal{K}_{\text{finite}}$. If it ensures $\mathbb{E}[V(t_{\alpha}) \slash \max\{ \widehat{V}(t_{\alpha}), \alpha\}] \leq 1$ asymptotically as $(n, N) \to \infty$ at a proper rate, then $\mathcal{G}\in \mathcal{K}_{\text{asy}}$.
\end{definition}
The flexibility of FEFP allows for selecting the feature (group) detection procedure $\mathcal{G}^{(m)}$ for each layer $m\in[M]$ from the sets $\mathcal{K}_{\text{finite}}$ or $\mathcal{K}_{\text{asy}}$. Consequently, researchers can select the most suitable detection procedure for each layer based on prior knowledge.

\begin{proposition}\label{proposition0} 
	Consider an alternative formulation to~(\ref{FDR_framework}) that
	\begin{equation}\label{new_control}
		t_{\alpha}= \sideset{}{}{\arg\max}_{t}^{}R(t), \quad \text{subject to}~\widehat{\mathrm{FDP}} (t) = \frac{\widehat{V}(t)} {R(t)\vee 1} \leq \alpha.
	\end{equation}
	The rejection set of the detection procedure $\mathcal{G} \in \mathcal{K}_{\text{asy}} \bigcup \mathcal{K}_{\text{asy}}$ is unchanged if we replace~(\ref{FDR_framework}) by~(\ref{new_control}).
\end{proposition}

Proposition~\ref{proposition0} characterizes an equivalent form of Definition~\ref{definition1}. With this insight, the following examples illustrate that Definition~\ref{definition1} includes the majority of commonly used feature or group detection procedures with FDR control.

\begin{example}[P-value based method]
	Consider the p-values $\{p_{j}\}$ for $\{H_j\}$ are available. The BH procedure~\cite{Benjamini1995} determines the number of rejections by $k_{\alpha} = \max_{k} \{k:  \# \{j: p_{j} \leq \alpha k \slash N   \} \geq k \}$, which is equivalent to determining the threshold $t_{\alpha}  = \max_{t} \{t\in \{\alpha \slash N, \ldots, \alpha N \slash N \}: \widehat{\text{FDP}} (t)=Nt \slash \#\{i : p_{j} \leq t \} \leq \alpha \}$. 
	It can be proved that $\mathbb{E}[\#\{j\in\mathcal{H}_{0}: p_{j} \leq t_{\alpha}\} \slash (Nt_{\alpha})] \leq 1$ if the p-values are independent, which implies that the BH procedure meets Definition~\ref{definition1}. In addition, it can also be verified that other p-values based methods such as \cite{benjamini2001control, sarkar2002some, storey2004strong} satisfy Definition~\ref{definition1}.
\end{example}

\begin{example}[Knockoff filters]
	Knockoff-based procedures including \cite{Barber2015,Candes2018, 2018Robust, li2022searching,ren2023knockoffs} satisfy Definition~\ref{definition1}. These methods estimate $V(t)$ by the symmetry property of null test statistics.
	Specifically, consider the example for low-dimensional Gaussian linear models. Barber and Candès~\cite{Barber2015} constructed a test statistic $W_{j}$ for $H_{j}$, which exhibits a large positive value when $H_{j}$ is false and is symmetric about zero when $H_{j}$ is true. The number of false discoveries $V(t)= \#\{j \in \mathcal{H}_{0}: W_{j} \geq t \}$ is overestimated by $\#\{j : W_{j} \leq -t \}$, owing to the symmetry property $\#\{j \in  \mathcal{H}_{0}: W_{j} \geq t\} \overset{d}{=} \#\{j \in  \mathcal{H}_{0}: W_{j} \leq -t\}$. 
    By the optional stopping time theorem, it holds that $\mathbb{E}[V(t_{\alpha}) \slash \max\{ \widehat{V}(t_{\alpha}), \alpha\}] \leq 1$.
\end{example}

\begin{example}[Other extensions]
Many powerful alternatives to knockoff methods satisfy Definition~\ref{definition1}.
For example, we demonstrate that the DS procedure~\cite{Dai2022}, a powerful method under high correlated settings with asymptotical FDR control, belongs to $\mathcal{K}_{\text{asy}}$.
In the Section F.1 of the Supplementary Material, we also prove that the Gaussian Mirror (GM) method~\cite{Xing2023} belongs to $\mathcal{K}_{\text{asy}}$.
More recently, Wang et al.~\citep{wang2025adaptive} proposed symmetry-based adaptive selection (SAS) framework, which directly uses two-dimensional statistics to determine a rejection hyperplane.
The SAS framework also satisfies Definition~\ref{definition1}; see Section F.2 of the Supplementary Material.
It can be also demonstrated that other procedures such as the symmetrized data aggregation (SDA) filter~\cite{du2023false} and the data-adaptive threshold selection procedure~\cite{guo2022threshold} also conform Definition~\ref{definition1}.
\end{example}

For clarity in the remainder of the paper, we denote $\widehat{V}_{\mathcal{G}(\alpha)}$ and $\mathcal{G}(\alpha)$ as $\widehat{V}(t_{\alpha})$ and $\mathcal{S}(t_{\alpha})$ in the context of $\mathcal{G}$, respectively.
The notation $\mathcal{G}$ should be distinguished from $\mathcal{G}(\alpha)$, where the former denotes a detection procedure, and the latter refers to a selected set.

\subsubsection*{Construction of the generalized e-values}
We start by explicitly defining the generalized e-values as follows.

\begin{definition}\label{definition2}
	Consider any fixed $N\in \mathbb{Z}_{+}$ and hypotheses $H_{1}, \dots, H_{N}$ associated with non-negative test statistics $e_{1}, \dots, e_{N}$.  
	The sample size is denoted as $n$ and the index set of null hypotheses is $\mathcal{H}_{0}$.
	If $\sum_{j\in\mathcal{H}_{0}}\mathbb{E}[e_{j}]\leq N$ holds, then the $e_{1}, \dots, e_{N}$ are termed relaxed e-values. Consider the asymptotic case where $(n, N )\to \infty$ at a proper rate. For $j\in \mathcal{H}_{0}$, if $\limsup_{(n, N)\to \infty}\mathbb{E}[e_{j}] \leq 1$, then $e_{j}$ is called an asymptotic e-value. Otherwise, if $\limsup_{(n, N)\to \infty}\sum_{j\in\mathcal{H}_{0}}\mathbb{E}[e_{j}]/N\leq 1 $ holds, then $e_{1}, \dots, e_{N}$ are called asymptotic relaxed e-values. Asymptotic e-values, relaxed e-values, and asymptotic relaxed e-values are collectively referred to as the generalized e-values.
\end{definition}

Based on the detection procedure $\mathcal{G}^{(m)} \in \mathcal{K}_{\text{finite}} \bigcup \mathcal{K}_{\text{asy}}$, we provide the construction of the generalized e-values for hypothesis $H_{g}^{(m)}$ which is as outlined below.
\begin{enumerate}
\item[Step 1.] For each $m\in[M]$, execute the detection procedure $\mathcal{G}^{(m)}$ with the original FDR level $\alpha^{(m)}_{0}$.
Denote the rejection set and the estimated number of false discoveries by $\mathcal{G}^{(m)}(\alpha_{0}^{(m)})$ and $\widehat{V}_{\mathcal{G}^{(m)}(\alpha_{0}^{(m)})}$, respectively.
\item[Step 2.] Transform the result of $\mathcal{G}^{(m)}$ into a list of the generalized e-values $\{e_{g}^{(m)}\}_{g\in [G^{(m)}]}$, where
\begin{equation}\label{generalized-e-values}
	e_{g}^{(m)}=  G^{(m)} \cdot \frac{\mathbb{I}\left \{g \in \mathcal{G}^{(m)}(\alpha_{0}^{(m)})\right \} } {\widehat{V}_{\mathcal{G}^{(m)}(\alpha_{0}^{(m)})} \bigvee \alpha_{0}^{(m)}}.
\end{equation}
\end{enumerate}

The validity of generalized e-values in Equation~\eqref{generalized-e-values} is supported by the equivalence between the generalized e-BH procedure and the detection procedure $\mathcal{G}^{(m)}$, where the generalized e-BH procedure refers to the e-BH procedure with generalized e-values.

\begin{theorem}[A generalized version of \cite{ren2022derandomized}]\label{theorem1}
	For any detection procedure $\mathcal{G}^{(m)} \in \mathcal{K}_{\text{finite}} \bigcup \mathcal{K}_{\text{asy}} $, the set $\mathcal{G}^{(m)}(\alpha_{0}^{(m)})$ equals to that of the generalized e-BH procedure at the same level $\alpha^{(m)}_{0}$, provided that the input is $\{e_g^{(m)}\}_{g\in[G^{(m)}]}$ defined in Equation~(\ref{generalized-e-values}). 
\end{theorem}

\begin{remark}[Comparison with existing constructions]\label{conpound_evalue}
The recent compound e-values~\cite{banerjee2023harnessing,ignatiadis2024compound} given by
\begin{equation*}\label{compound_evalue}
	\tilde{e}_{g}^{(m)}=  G^{(m)} \cdot \frac{\mathbb{I}\left \{g \in \mathcal{G}^{(m)}(\alpha_{0}^{(m)})\right \} } {\alpha_{0}^{(m)}\left(\left|\mathcal{G}^{(m)}(\alpha_{0}^{(m)})\right|\bigvee 1\right)}.
\end{equation*}
seem to provide an alternative construction of \eqref{generalized-e-values}, since the set $\{\tilde{e}_g^{(m)}\}$ consistently represents a set of (asymptotic) relaxed e-values when procedure $\mathcal{G}^{(m)}$ controls FDR, where $\sum_{g\in \mathcal{H}_{0}^{(m)}}\mathbb{E}[\tilde{e}_{g}^{(m)}] = G^{(m)}\mathbb{E}\left[V_{\mathcal{G}^{(m)}} \slash\left|\mathcal{G}^{(m)}(\alpha_{0}^{(m)})\right|\right] \slash \alpha^{(m)}_{0} \leq G^{(m)}$.\footnote{The expectation $\mathbb{E}\left[V_{\mathcal{G}^{(m)}} \slash\left|\mathcal{G}^{(m)}(\alpha_{0}^{(m)})\right|\right]$ is exactly the achieved FDR of the detection procedure $\mathcal{G}^{(m)}$ with $\alpha_{0}^{(m)}$.}
However, such generalized e-values are inherently less powerful compared to those defined in Equation~\eqref{generalized-e-values} since it always holds that $e_{g}^{(m)'} \leq e_{g}^{(m)}$. The numerical results in the Section D of the Supplementary Materials demonstrate that the generalized e-values given by~\eqref{generalized-e-values} achieve significantly higher power.
Compared with the knockoff e-values~\cite{ren2022derandomized}, our construction is more generic, allowing the construction of generalized e-values based on any detection procedures that satisfy Definition~\ref{definition1}. Also differs from \cite{li2024note} where the validity of e-values relies on the mutual independence of the null p-values, etc., our construction is more generic and direct, which can be obtained from a broad class of FDR-controlled procedures beyond the knockoff-filter and the p-value based approaches. For example, the DS e-values developed in Section~\ref{section4_application}, the GM e-values and the two-dimensional statistics based generalized e-values (SAS e-values) developed in Section F of the Supplementary Material are beyond the scope of \citep{li2024note}.
\end{remark}

\subsubsection*{Leveraging the generalized e-filter}
We illustrate how to identify the selected features by using the derived generalized e-values and the generalized e-filter.
To ensure consistency across all layers, we define the candidate selection set as
\begin{equation}\label{Set}
\mathcal{S}(t^{(1)},\dots,t^{(M)}) = \{j: \text{for all}~m \in [M], e_{h(m,j)}^{(m)} \geq t^{(m)} \}
\end{equation}
for any thresholds $(t^{(1)},\dots,t^{(M)})\in [1, \infty]^{M}$,
that is, $H_{j}$ is rejected only if all groups containing $X_{j}$ are rejected. 
The set of admissible thresholds, denoted as $\mathcal{T}(\alpha^{(1)}, \dots, \alpha^{(M)})$, is defined by
\[
\mathcal{T}(\alpha^{(1)}, \dots, \alpha^{(M)})=\left\{(t^{(1)}, \dots, t^{(M)}) \in [1, \infty]^{M}: \widehat{\text{FDP}}^{(m)}(t^{(1)}, \dots, t^{(M)}) \leq \alpha^{(m)} \text{for all } m \right\},
\]
where the estimated false discovery rate is defined as
\[
\widehat{\text{FDP}}^{(m)}(t^{(1)}, \dots, t^{(M)})=\frac{G^{(m)}\slash t^{(m)}}{\left|\mathcal{S}^{(m)}(t^{(1)}, \dots, t^{(M)}) \right|}.
\]
The justification for $\widehat{\text{FDP}}^{(m)}$ is based on the approximation
\[
V(t^{(m)}) \approx \sum_{g \in \mathcal{H}^{(m)}_{0}}\mathbb{P}(e_{g}^{(m)}\geq t^{(m)}) \leq \left|\mathcal{H}^{(m)}_{0}\right| \slash t^{(m)} \leq G^{(m)}\slash t^{(m)}.
\]
For each layer $m\in[M]$, the final threshold $\widehat{t}^{(m)}$ is determined by
\begin{equation}\label{threshold}
	\widehat{t}^{(m)} = \min \left\{ t^{(m)}: (t^{(1)},\ldots,t^{(M)}) \in \mathcal{T} (\alpha^{(1)}, \ldots, \alpha^{(M)}) \right\}.
\end{equation}
According to \cite{gablenz2023catch}, we have $(\widehat{t}^{(1)},\dots,\widehat{t}^{(M)})\in \mathcal{T}(\alpha^{(1)},\dots,\alpha^{(M)})$.

Thus far, we establish the hypotheses for multiple testing in the context of feature detection and give a detailed explanation of FEFP.
Note that the generalized e-filter (summarized in Algorithm~\ref{alg2}) can be used to guaranty multilayer FDR control for other hypotheses, provided that the corresponding generalized e-values are available.
The following proposition further establishes the correctness of Algorithm~\ref{alg2}.

\begin{proposition}\label{proposition2} 
The output threshold vector $(\widehat{t}^{(1)}, \dots, \widehat{t}^{(M)})$ generated by Algorithm~\ref{alg2} corresponds precisely to the one given in Equation~(\ref{threshold}).
\end{proposition}

Now we turn to establish the theoretical property of FEFP.
\begin{algorithm}[t]\label{alg2}
	\caption{The generalized e-filter}
	\KwIn {Generalized e-value $ e_{g}^{(m)}~\text{with respect to hypothesis } H_{g}^{(m)},~g\in[G^{(m)}],~m \in [M]$; target FDR level vector $(\alpha^{(1)}, \ldots, \alpha^{(M)})$; partition $m$ given by $\left\{ \mathcal{A}_{1}^{(m)}, \ldots, \mathcal{A}_{G^{(m)}}^{(m)} \right\}$, $m \in [M]$. }
	\textbf{Initialize:} $t^{(m)}= \frac{1}{\alpha^{(m)}}, m \in[M]$.
	
	\Repeat{$t^{(1)}, \ldots, t^{(M)} $ \textrm{are all unchanged}}{
		\For{$m=1, \ldots, M$}{
			Compute candidate rejection set $\mathcal{S}(t^{(1)},\dots,t^{(M)})$ by Equation~(\ref{Set}).
			
			Update the threshold
			\[
			t^{(m)} \leftarrow \min \left\{t \in [t^{(m)},+ \infty   ] : \frac{G^{(m)} \frac{1}{t}}{ 1  \bigvee 
				\left| \mathcal{S}^{(m)}(t^{(1)},\ldots,t^{(m-1)},t, t^{(m+1)},\ldots,t^{(M)})\right|}  \le \alpha^{(m)} \right\}.
			\]
	}}
	
	\KwOut{the final threshold $\widehat{t}^{(m)} = t^{(m)}$ for $m\in [M]$ and the rejected set $\mathcal{S}(\widehat{t}^{(1)}, \ldots, \widehat{t}^{(M)} )$. }
\end{algorithm}

\subsubsection*{Multilayer FDR control and one-bit property}
We demonstrate that the FEFP method enables the use of distinct feature or group detection procedures, $\mathcal{G}^{(m)}$, across different layers, while maintaining multilayer FDR control.

\begin{lemma}\label{lemma1_1}
	For any $N\in \mathbb{Z}_{+}$, let the threshold vector $(\widehat{t}^{(1)}, \dots, \widehat{t}^{(M)})$ be determined by Equation~(\ref{threshold}) and $(e_1^{(m)},\dots,e_{G^{(m)}}^{(m)})$ be the generalized e-values satisfying Definition~\ref{definition2} for $m\in[M]$.
    Let $\mathcal{S}$ be the selected set of the generalized e-filter. Then it holds that 
    \[
    \textnormal{FDR}^{(m)} \leq \frac{\alpha^{(m)}}{G^{(m)}}\sum_{g=1}^{G^{(m)}} \mathbb{E}[e_{g}^{(m)}\mathbb{I}\{g\in\mathcal{H}_0^{(m)}\}] \quad\text{simultaneously for }m = 1, \dots, M.
    \]
    Specifically, if $\{e_{g}^{(m)}\}$ represents a set of e-values, then 
    $\textnormal{FDR}^{(m)} \leq \pi_{0}^{(m)}\alpha^{(m)}$, where $\pi_{0}^{(m)} =| \mathcal{H}_{0}^{(m)}|/G^{(m)}$ is the null proportion. If $\{e_{g}^{(m)}\}$ represents a set of relaxed e-values, then $\textnormal{FDR}^{(m)} \leq \alpha^{(m)}$. Consider $(n, G^{(m)} )\to \infty$ at a proper rate. If $\{e_{g}^{(m)}\}$ represents a set of asymptotic e-values, then $\limsup_{(n, G^{(m)}) \to \infty}\textnormal{FDR}^{(m)} \leq \overline{\pi}_{0}^{(m)} \alpha^{(m)}$, where $\overline{\pi}_{0}^{(m)} =\limsup_{G^{(m)} \to \infty} | \mathcal{H}_{0}^{(m)}|/G^{(m)}$. If $\{e_{g}^{(m)}\}$ represents a set of asymptotic relaxed e-values, then $\limsup_{(n, G^{(m)}) \to \infty}\textnormal{FDR}^{(m)} \leq \alpha^{(m)}$.
\end{lemma}
In particular, different layers may correspond to distinct types of generalized e-values.

\begin{theorem}\label{theorem2}
	Given $(\alpha^{(1)}_{0},\dots,\alpha^{(M)}_{0}) \in (0,1)^{M}$ and $(\alpha^{(1)},\dots,\alpha^{(M)}) \in (0,1)^{M}$, for each $m \in [M]$, if $\mathcal{G}^{(m)} \in \mathcal{K}_{\text{finite}}$, then FEFP ensures that $\text{FDR}^{(m)} \leq \alpha^{(m)}$ in finite sample settings. Furthermore, if $\mathcal{G}^{(m)} \in \mathcal{K}_{\text{asy}}$ and the group size $|\mathcal{A}_{g}^{(m)}|$ is uniformly bounded, then FEFP guaranties that $\text{FDR}^{(m)} \leq \alpha^{(m)}$ as $(n, N)\to \infty$ at a proper rate.
\end{theorem}
Theorem~\ref{theorem2} demonstrates that FEFP provides multilayer FDR control.
The following theorem further presents the one-bit property of FEFP.

\begin{theorem}[One-bit Property]\label{onebitproperty}
Write 
\[
\mathcal{S}_{\text{init}}^{(m)}=\left\{g\in G^{(m)}: \mathcal{A}_{g}^{(m)}\bigcap \left[\bigcap_{l=1}^{M}\left\{j:e_{h(l,j)}^{(l)}>0\right\}\right]\neq \emptyset\right\}.
\]
The FEFP selects the set $\mathcal{S}\left(G^{(1)}/\widehat{V}_{\mathcal{G}^{(1)}(\alpha_{0}^{(1)})},\dots,G^{(M)}/\widehat{V}_{\mathcal{G}^{(M)}(\alpha_{0}^{(M)})}\right)$ if and only if the event $\bigcap_{m=1}^M\left\{ \widehat{V}^{(m)}_{\mathcal{G}^{(m)}(\alpha_0^{(m)})}\leq \alpha^{(m)}|\mathcal{S}^{(m)}_{\text{init}}| \right\}$ occurs. Otherwise, if there exists $m$ such that $\widehat{V}^{(m)}_{\mathcal{G}^{(m)}(\alpha_0^{(m)})}>\alpha^{(m)}|\mathcal{S}_{\text{init}}^{(m)}|$, then FEFP makes no rejections.
A sufficient condition for the event $\bigcap_{m=1}^M\left\{ \widehat{V}^{(m)}_{\mathcal{G}^{(m)}(\alpha_0^{(m)})}\leq \alpha^{(m)}|\mathcal{S}^{(m)}_{\text{init}}| \right\}$ to occur is that $|\mathcal{G}^{(m)}(\alpha_0^{(m)})|/|\mathcal{S}^{(m)}_{\text{init}}|\leq \alpha^{(m)}/\alpha_0^{(m)}$ for all $m\in[M]$.
\end{theorem}

Based on this property, we summarize FEFP in Algorithm~\ref{alg1}.
By Theorem~\ref{onebitproperty}, 
FEFP may yield zero power. This instability is caused by the one-bit nature of generalized e-values, that is, all groups initially rejected by $\mathcal{G}^{(m)}$ are assigned equal importance. As a consequence, when there are significant conflicts in the discovery of different layers, FEFP is forced to consider all features as unimportant.

\subsection{SFEFP: Stabilized Flexible E-Filter Procedure}\label{section2.3}
To avoid the zero-power-dilemma, we aim to obtain non-one-bit generalized e-values that better reflect the ranking of (group) feature importance.
We consider two settings as follows.
\begin{itemize}
\item[(1)] If the base detection procedure $\mathcal{G}^{(m)}\in \mathcal{K}_{\text{finite}} \bigcup \mathcal{K}_{\text{asy}}$ has inherent randomness (e.g., the model-X knockoff filter), then different runs of $\mathcal{G}^{(m)}$ may result in different $\{e_{g}^{(m)}\}$. In this case, we suggest repeating $\mathcal{G}^{(m)}$ multiple times and using the averaged generalized e-values.
\item[(2)] If the base detection procedure $\mathcal{G}^{(m)}\in \mathcal{K}_{\text{finite}} \bigcup \mathcal{K}_{\text{asy}}$ has no randomness, then $\{e_{g}^{(m)}\}$ is determined in different runs. In this case, we take the fusion decision, i.e., running a different determined procedure $\mathcal{G}^{(m)}$ and averaging their one-bit generalized e-values.
\end{itemize}

\begin{algorithm}[t]
	\caption{SFEFP: a stabilized flexible e-filter procedure for feature detection}
	\label{alg3}
	\KwIn {Data $(\mathbf{X},\mathbf{y})$; target FDR level vector $(\alpha^{(1)}, \ldots, \alpha^{(M)})$; original FDR level vector $(\alpha_{0}^{(1)}, \ldots, \alpha_{0}^{(M)})$; partition $m$ given by $\left\{ \mathcal{A}_{1}^{(m)}, \ldots, \mathcal{A}_{G^{(m)}}^{(m)} \right\}$, $m \in [M]$. }
\For {$m=1, \ldots, M$}{ 
\For {$r=1,\dots,R^{(m)}$}{
Process detection procedure $\mathcal{G}^{(m)}_{r} \in \mathcal{K}_{\text{finite}} \bigcup \mathcal{K}_{\text{asy}}$ with level $\alpha_{0}^{(m)}$.

Compute the generalized e-value $e_{gr}^{(m)}$ by Equation \eqref{e-value_different_run}, $g\in[G^{(m)}]$.
}
Compute $\overline{e}_{g}^{(m)} =\sum_{r\in[R^{(m)}]} w^{(m)}_{r}e_{gr}^{(m)},g\in[G^{(m)}]$.
}
\textbf{Initialize:} $t^{(m)}= \frac{1}{\alpha^{(m)}}, m \in[M]$.

	\Repeat{ $t^{(1)}, \ldots, t^{(M)} $ are all unchanged }{
		\For{$m=1, \ldots, M$}{
			$t^{(m)} \leftarrow \min \left\{t \in [t^{(m)},+ \infty) : \frac{G^{(m)} \frac{1}{t}}{ 1  \bigvee 
				\left| \mathcal{S}^{(m)}(t^{(1)},\ldots,t^{(m-1)},t, t^{(m+1)},\ldots,t^{(M)})\right|}  \le \alpha^{(m)} \right\} $, 
			
			where $\mathcal{S}^{(m)}$ is computed as in Equation~(\ref{Set}) with input $\overline{e}_{g}^{(m)}$.
	}}
    $\widehat{t}^{(m)} \leftarrow t^{(m)}$, $\mathcal{S}_{\text{stab}}=\mathcal{S}(\widehat{t}^{(1)},\dots,\widehat{t}^{(M)})$.
	
	\KwOut{The selected set $\mathcal{S}_{\text{stab}}$. }
	
\end{algorithm}

For these two settings, we summarize the SFEFP method in Algorithm~\ref{alg3}. 
Specifically, for each layer $m \in [M]$ and each run $r \in [R^{(m)}]$, we choose the detection procedure $\mathcal{G}^{(m)}_{r}\in \mathcal{K}_{\text{finite}} \bigcup \mathcal{K}_{\text{asy}}$.
We compute the $r$-th selection set $\mathcal{G}^{(m)}_{r}(\alpha_{0}^{(m)})$ and the $r$-th estimated number of false discoveries $\widehat{V}_{\mathcal{G}^{(m)}_{r}(\alpha_{0}^{(m)})}$.
The corresponding generalized e-value is given by
\begin{equation}\label{e-value_different_run}
e_{gr}^{(m)}=G^{(m)} \cdot \frac{\mathbb{I}\left \{g \in \mathcal{G}^{(m)}_{r}(\alpha_{0}^{(m)}) \right \}} {\widehat{V}_{\mathcal{G}^{(m)}_{r}(\alpha_{0}^{(m)})} \bigvee \alpha^{(m)}_{0}}.
\end{equation}
The generalized e-values of the realizations of $R^{(m)}$ are subsequently aggregated by calculating a weighted average, given by
\begin{equation}\label{average}
\overline{e}^{(m)}_{g}=\sum_{r=1}^{R^{(m)}} \omega^{(m)}_{r}e_{gr}^{(m)}, \quad\quad \sum_{r=1}^{R^{(m)}} \omega^{(m)}_{r}=1,
\end{equation}
where $\{\omega^{(m)}_{r}\}$ denotes the set of normalized weights. 
We take $\omega^{(m)}_{r}=1/R$ to illustrate the proposed method.
Subsequently, by applying the generalized e-filter to the stabilized generalized e-values at the target levels $\alpha^{(1)}, \dots, \alpha^{(M)}$, we derive a set of selected features $\mathcal{S}_{\text{stab}}$.
Note that for setting (I), $\mathcal{G}_{r}^{(m)}$ is identical to $\mathcal{G}^{(m)}$. And for setting (II), $\mathcal{G}_{r}^{(m)}$ changes when $r$ varies, 
where weights $\{\omega^{(m)}_{r}\}$ reflect the importance of different procedures.

For single resolution, as demonstrated by simulations of \cite{ren2022derandomized} and clarified in \cite{lee2024boosting}, the derandomized procedure tends to suffer power loss compared to 
the original non-derandomized one,
since the average of tight e-values is no longer tight. 
Here, the results of multiple resolutions are quite different. SFEFP facilitates the acquisition of more precise generalized e-values that better rank the features or groups, which are instrumental in addressing potential irreconcilable conflicts between different layers.

\subsubsection*{Multilayer FDR control and stability guarantee}
The following theorem demonstrates that SFEFP is capable of simultaneously controlling $\text{FDR}^{(m)}$ at the level $\alpha^{(m)}$ for $m\in[M]$.

\begin{theorem}\label{theorem3}
	Given any original vector of FDR level $(\alpha^{(1)}_{0},\dots,\alpha^{(M)}_{0}) \in (0,1)^{M}$, any target vector of FDR level $(\alpha^{(1)}, \dots,\alpha^{(M)}) \in (0,1)^{M}$, and any number of replications $R \geq1$,
	for each $m \in [M]$, if $\mathcal{G}^{(m)} \in \mathcal{K}_{\text{finite}}$, then the selected set $\mathcal{S}_{\text{stab}}$ calculated by SFEFP satisfies $\text{FDR}^{(m)} \leq \alpha^{(m)}$ in finite sample settings. Furthermore, if $\mathcal{G}^{(m)} \in \mathcal{K}_{\text{asy}}$ and the group size $|\mathcal{A}_{g}^{(m)}|$ is uniformly bounded, then SFEFP satisfies $\text{FDR}^{(m)} \leq \alpha^{(m)}$ as $(n, N)\to \infty$ at a proper rate.
\end{theorem}

We further study the stability of SFEFP. We only need to consider setting (I) since the selection set is determined under setting (II).
For each layer $m\in[M]$, suppose that the feature or group detection procedure $\mathcal{G}^{(m)}$ is randomized. 
Conditional on $(\boldsymbol{X}, \boldsymbol{y})$, the $e_{gr}^{(m)}$'s are i.i.d. and bounded across each replication $r\in[R^{(m)}]$. 
By the strong law of large numbers, we have 
\[
\overline{e}_{g}^{(m)} = \frac{1}{R^{(m)}}\sum\limits_{r=1}^{R^{(m)}} e_{gr}^{(m)} \to \Bar{\Bar{e}}_{g}^{(m)}:= \mathbb{E}[e_{g1}^{(m)} | \boldsymbol{X}, \boldsymbol{y}] \quad \text{almost surely as } R^{(m)}\to \infty. 
\]
Let $\mathcal{S}_{\infty}$ denote the set of selected features obtained by SFEFP as $R^{(m)}\to \infty$, $m\in[M]$.
The set $\mathcal{S}_{\infty}$ becomes a deterministic function of the data $(\boldsymbol{X}, \boldsymbol{y})$ because $\Bar{\Bar{e}}_{g}^{(m)}$ is fixed. 
The following theorem demonstrates the stability guaranty of SFEFP.

\begin{theorem}[Stability]\label{derandomized}
	Let 
	\[
	\Delta^{(m)} = \min_{g\in[G^{(m)}]}\left|\Bar{\Bar{e}}_{g}^{(m)}- \widehat{t}^{(m)}_{\infty}\right|, \quad m\in[M],
	\]
	where $(\widehat{t}_{\infty}^{(1)},\dots,\widehat{t}_{\infty}^{(M)})$ represents the threshold vector corresponding to $\mathcal{S}_{\infty}$. Based on the data $(\boldsymbol{X}, \boldsymbol{y})$,  we have
	\[
\mathbb{P}\left(\mathcal{S}_{\text{stab}} =\mathcal{S}_{\infty} | \boldsymbol{X}, \boldsymbol{y} \right) \geq 1-2\min\left\{\frac{1}{2},\sum_{m=1}^M G^{(m)}\exp\left(\frac{-2(\Delta^{(m)})^2R^{(m)}}{(G^{(m)})^2}\right)\right\} .
	\]
\end{theorem}

\subsection{Choices of parameters}\label{Choices of parameters}
The original FDR level $\{\alpha_0^{(m)}\}$ does not affect the FDR guaranty, but influences the detection power.
In general, a small $\alpha_{0}^{(m)}$ implies running the base procedure strictly, which reduces the number of rejections and lowers the estimated number of false rejections, thereby resulting in non-zero generalized e-values of greater magnitude. However, this also leads to smaller $|\mathcal{G}^{(m)}|$.
Furthermore, unlike single resolution, a large fraction $\alpha^{(m)}/\alpha_0^{(m)}$ results in a better ability to handle conflicts between different layers, as demonstrated by Theorem~\ref{onebitproperty}.

For $M=1$, the choice of $\alpha_0^{(1)}$ is consistent with \cite{ren2022derandomized}.
Suppose that non-nulls exhibit extremely strong signal strength such that $\mathcal{G}_{r}^{(1)}$ selects the entire set of non-nulls for each run $r$. Denote the selected set by $\mathcal{C}$.
Thus, the number of false discoveries $|\mathcal{G}_{r}^{(1)}(\alpha_{0}^{(1)}) \cap \mathcal{H}_{0}|$ can be approximated by $\alpha_{0}^{(1)} |\mathcal{C}| \slash (1-\alpha_{0}^{(1)})$ as a result of $\widehat{\text{FDP}}  \approx \text{FDP} \approx \alpha_{0}^{(1)}$, which further implies that $\overline{e}_{j} \approx N \slash (\alpha_{0}^{(1)} |\mathcal{C}| \slash (1-\alpha_{0}^{(1)})) $ for $j\in \mathcal{B}$ and $\overline{e}_{j} \approx 0$ for $j \in [N] \setminus \mathcal{C}$.
To ensure that the set $\mathcal{C}$ is selected by the generalized e-BH procedure, it is necessary that $ N \slash (\alpha_{0}^{(1)} |\mathcal{C}| \slash (1-\alpha_{0}^{(1)})) \geq N \slash (\alpha^{(1)} |\mathcal{C}|)$, which subsequently implies $\alpha_{0}^{(1)} \leq \alpha^{(1)} \slash (1+\alpha^{(1)})$.
A special case occurs when $R=1$, where $\alpha_{0}^{(1)}=\alpha^{(1)}$ is evidently the optimal choice.

The choice of $\{\alpha_0^{(m)}\}$ under multiple resolutions ($M \geq 2$) differs from $M=1$. The key is that $\alpha_{0}^{(m)} = \alpha^{(m)}$ is no longer optimal when $R=1$.  
Specifically, if $R=1$, the generalized e-values are binary at each resolution.
As previously discussed, coordination between the discoveries across different layers is crucial, and such all-or-nothing generalized e-values may reduce the coordinated space, potentially resulting in zero power. 
Therefore, the selection of $\alpha_{0}^{(m)}$ should consider reducing the contradictions between different layers, rather than only pursuing to maximize the number and the magnitude of the non-zero e-values at each layer.

We provide several recommendations as follows. 
Firstly, choose $\alpha_{0}^{(m)} \leq \alpha^{(m)} \slash (1+\alpha^{(m)})$.
This is because a feature is selected by SFEFP only when its corresponding group is selected separately at each layer. 
Demonstrated by the simulation results, $\alpha_{0}^{(m)} = \alpha^{m} \slash 2$ is great choice.
Secondly, a smaller $\alpha_{0}^{(m)}$ is recommended in high-dimensional scenarios with sparse signals.
The reason is that increasing $\alpha^{(m)}_{0}$ has a limited impact on boosting the number of non-zero generalized e-values due to sparsity, but it tends to get inflated estimates of false discoveries, ultimately resulting in non-zero generalized e-values with smaller magnitude. 
Finally, as suggested by Gablenz and Sabatti~\cite{gablenz2023catch}, selecting $\alpha_{0}^{(m)}$ by tuning with a hold-out data set may be feasible, although their context differs somewhat from ours.

\section{Examples for SFEFP}\label{section4_application}
We introduce several useful filters as examples of SFEFP. These filters serve as powerful alternatives to MKF~\cite{katsevich2019multilayer} in different specific scenarios.
The first example is \textit{eDS-filter}, which is designed for the high-correlated features.
We develop this by firstly extending DS to group detection and then applying the SFEFP. The second example is \textit{eDS+gKF-filter}, which is designed for the settings where features within the same group are highly correlated, whereas correlations across groups are weak. Additional potential extensions of SFEFP are also discussed.

\subsection{eDS-filter: multilayer FDR control by data splitting}\label{section3.1}
Dai et al.~\cite{Dai2022} proposed the DS procedure and its refinement (MDS) as feature detection methods with asymptotic FDR control, specifically tailored for high-dimensional regression models. It has been demonstrated through extensive simulations~\cite{Dai2022,dai2023scale} that DS and MDS exhibit significantly higher power than knockoff when features are highly correlated in regression models. Inspired by this, in this paper, we extend the DS method to group detection by constructing group mirror statistics~\cite{dai2016knockoff}. Then we establish the \textit{eDS-filter} for multilayer FDR control.

We begin by reviewing the DS method. Let the set of coefficients $\widehat{\boldsymbol{\beta}}=(\widehat{\beta}_{1},\dots,\widehat{\beta}_{N})^\top$ represent the importance of each feature, where a larger value of $|\widehat{\beta}_{j}|$ suggests that $X_{j}$ is more important.
The DS procedure randomly splits the $n$ observations into two subsets $(\boldsymbol{y}^{(i)}, \boldsymbol{X}^{(i)})_{i=1,2}$ and uses each subset of the data to estimate $\widehat{\boldsymbol{\beta}}^{(1)}$ and $\widehat{\boldsymbol{\beta}}^{(2)}$, respectively. 
To ensure the independence of $\widehat{\boldsymbol{\beta}}^{(1)}$ and $\widehat{\boldsymbol{\beta}}^{(2)}$, the data-splitting must be independent of the response vector $\boldsymbol{y}$. 
In order to estimate the false discoveries, the following assumption is required. 
\begin{assumption}[Symmetry]\label{assumption1}
	For $j\in \mathcal{H}_{0}$, the sampling distribution of $\widehat{\beta}^{(1)}_{j}$ or $\widehat{\beta}^{(2)}_{j}$ is symmetric about zero.
\end{assumption}
A general form of the test statistic $W_{j}$ is
\[
W_{j}=\text{sign}(\widehat{\beta}_{j}^{(1)}\widehat{\beta}_{j}^{(2)})f(|\widehat{\beta}^{(1)}_{j}|,|\widehat{\beta}^{(2)}_{j}|),
\]
where the function $f(u,v)$ is nonnegative, exchangeable, and monotonically increasing defined for nonnegative $u$ and $v$; for instance, the function $f(u, v) = u+v$ satisfies these conditions.
A feature with a larger positive test statistic is more likely to be relevant. 
Meanwhile, the sampling distribution of the test statistic for any irrelevant feature is symmetric around zero.
For any $\alpha \in (0,1)$, the set $\{j:W_{j} \geq t_{\alpha}\}$ is selected, where
\[
t_{\alpha}=\min\left\{ t > 0 : \widehat{\text{FDP}}(t)=\frac{\#\{j : W_{j} < -t   \}} { \#\{j : W_{j} >t   \} \vee \text{1}} \leq \alpha \right\}.
\]
Under the assumption of weak dependence in~\cite{Dai2022} and certain technical conditions, it has been demonstrated that the FDR control is achieved as $N \to \infty$.

\subsubsection*{Data splitting for group detection}
We now expand the data splitting framework to address the problem of group detection.
To avoid introducing additional symbols, we focus on any fixed $m\in \{2,\dots,M\}$.
For hypothesis $H_{g}^{(m)}$, the test statistic $T_{g}^{(m)}$ is constructed as 
\begin{equation}\label{group statistics}
T_{g}^{(m)}=\frac{1}{|\mathcal{A}_{g}^{(m)}|}\sum\limits_{j\in \mathcal{A}_{g}^{(m)}}W_{j},\quad g\in[G^{(m)}].
\end{equation}

\begin{lemma} \label{lemma01}
	Under Assumption~\ref*{assumption1}, regardless of the data-splitting procedure, we have
	\[
	\#\left\{g\in  \mathcal{H}_{0}^{(m)}: T_{g}^{(m)}\geq t \right\} \overset{d}{=} \#\left\{g\in  \mathcal{H}_{0}^{(m)}: T_{g}^{(m)}\leq -t\right\}, \quad {\forall} t>0.
	\]
\end{lemma}

\begin{remark}\label{group_test}
	There exist several alternative ways to construct test statistics.
	For example, when an important group contains only a few relevant features, but these features exhibit strong signal strength, the test statistic $T_{g}^{(m)} = \max_{j\in \mathcal{A}_{g}^{(m)}} W_{j}$ may exhibit greater power. Similarly, we have $\#\left\{g\in  \mathcal{H}_{0}^{(m)}: T_{g}^{(m)}\geq t \right\} \overset{d}{=} \#\left\{g\in  \mathcal{H}_{0}^{(m)}: \min_{j\in\mathcal{A}_{g}^{(m)}}W_{j} \leq -t\right\}$ for any $t>0$. 
\end{remark}
The proof is straightforward. It can be easily verified that if hypothesis $H_{g}^{(m)}$ is false, the test statistic $T_{g}^{(m)}$ tends to have a large value.
This property, together with Lemma~\ref{lemma01}, immediately provides a precise upper bound on the number of false discoveries that
\[
\#\{g \in  \mathcal{H}_{0}^{(m)}: T_{g}^{(m)} \geq t   \}  \approx \#\{g \in  \mathcal{H}_{0}^{(m)}: T_{g}^{(m)}\leq -t \}  \leq  \#\{g : T_{g}^{(m)} \leq -t  \},\quad {\forall} t>0.
\]
Hence, the estimate of $\text{FDP}^{(m)}(t)$, denoted as $\widehat{\text{FDP}}^{(m)} (t) $, is given by
\[
 \widehat{\text{FDP}}^{(m)} (t) =  \frac{\#\left\{g : T_{g}^{(m)}< -t  \right\}} { \#\left\{g : T_{g}^{(m)} >t   \right\} \bigvee \text{1}}.
\]
The set of selected groups is then given by $\mathcal{G}^{(m)}(\alpha^{(m)})=\left\{g\in[G^{(m)}]:T_{g}^{(m)} \geq t_{\alpha^{(m)}}\right\}$, where $t_{\alpha^{(m)}}=\min\left\{ t > 0 : \widehat{\text{FDP}}^{(m)} (t) \leq \alpha^{(m)} \right\}$.
The following assumption ensures weak dependence among the null test statistics, which is essential for controlling the FDR.

\begin{assumption}[Group weak dependence]\label{assumption2}
	The test statistics $T^{(m)}_{g}$'s are continuous random variables, and there exist constants $c^{(m)} > 0$ and $\lambda^{(m)} \in (0,2) $ such that
	\[
	\sum\limits_{ g\neq h \in \mathcal{H}_{0}^{(m)}}\mathrm{Cov}\left(\mathbb{I}\left\{T_{g}^{(m)}>t\right\}, \mathbb{I}\left\{T_{h}^{(m)}>t\right\}\right) \leq c^{(m)} |\mathcal{H}^{(m)}_{0}|^{\lambda^{(m)}}, \quad {\forall}t \in \mathbb{R} \nonumber.
	\]
\end{assumption}

\begin{theorem}\label{theorem6}
	Given any $\alpha^{(m)} \in (0, 1)$, suppose that there exists a constant $\tau_{\alpha^{(m)}} >  0$ such that $\mathbb{P}(\text{FDP}^{(m)}(\tau_{\alpha^{(m)}}) \leq \alpha^{(m)}) \to 1 $ as $G^{(m)} \to \infty$. Additionally, assume that $\text{Var}(T_{g}^{(m)})$ is uniformly upper bounded and also lower bounded away from zero. Then, under Assumptions~\ref{assumption1} and \ref{assumption2}, we have 
	\[
	\limsup_{G^{(m)}\to \infty} \mathrm{FDR}^{(m)}(t_{\alpha^{(m)}})  \leq \alpha^{(m)}.
	\]
\end{theorem}
The proof is similar to that of Proposition 2.2 in~\cite{Dai2022}. It is important to note that, apart from Assumptions \ref{assumption1} and \ref{assumption2}, the remaining conditions are merely technical, aimed at handling the most general case without the need to specify a parametric model between the response and features. Under the data splitting framework for both feature detection and group detection, we further provide multilayer FDR control by introducing \textit{eDS-filter}, an instantiation of SFEFP by specifying $\mathcal{G}_{r}^{(m)}$ as DS method for $m\in[M]$ and $r\in[R]$. The application of SFEFP not only enhance the detection power, but also overcomes the instability of data-splitting framework, i.e., the selection result of data splitting framework can vary substantially across different sample splits.

\subsubsection*{eDS-filter}
The \textit{eDS-filter} is introduced as follows.
At each layer $m\in[M]$, we independently repeat the data splitting procedure for group (individual) detection $R$ times at the partition $\{\mathcal{A}_{g}^{(m)}\}_{g\in[G^{(m)}]}$. 
Let $\{T^{(m)}_{gr}\}_{g\in[G^{(m)}]}$ represent the set of test statistics computed in the $r$-th replication. Given the original FDR level $\alpha_{0}^{(m)}$, we compute the threshold as
\begin{equation}
t^{r}_{\alpha_{0}^{(m)}} = \inf\left\{t >0:  \frac{\#\left\{g : T_{gr}^{(m)}< -t  \right\}} { \#\left\{g : T_{gr}^{(m)} >t   \right\} \bigvee \text{1}} \leq \alpha_{0}^{(m)}\right\}.
\end{equation}
Then the DS generalized e-values is given by
\begin{equation}\label{DS-evalue}
	e_{gr}^{(m)}= G^{(m)} \cdot \frac{\mathbb{I} \left \{ T_{gr}^{(m)} \geq t^{r}_{\alpha_{0}^{(m)}} \right \} }{\#\left\{g: T_{gr}^{(m)}\leq -t^{r}_{\alpha_{0}^{(m)}}\right\} \bigvee \alpha_{0}^{(m)} },\quad g\in[G^{(m)}].
\end{equation}
Subsequently, Equation~(\ref{average}) is employed to calculate the averaged generalized e-values $\{\overline{e}_{g}^{(m)}\}_{g\in[G^{(m)}]}$.
Finally, by applying the generalized e-filter at the target level $(\alpha^{(1)},\dots,\alpha^{(M)})$ to the averaged generalized e-values $\{\overline{e}_{g}^{(m)}\}_{g\in[G^{(m)}],m\in[M]}$, we obtain the set of selected features, denoted as $\mathcal{S}_{\text{eDS-filter}}$.

\begin{theorem}\label{theorem7}
	Consider any $(\alpha^{(1)}_{0},\dots,\alpha^{(M)}_{0}) \in (0,1)^{M}$, $(\alpha^{(1)}, \dots,\alpha^{(M)}) \in (0,1)^{M}$, and any number of replications $R \geq1$. For each layer $m\in[M]$, suppose that there exists a constant $\tau_{\alpha^{(m)}_{0}} > 0$ such that $\mathbb{P}(\text{FDP}^{(m)}(\tau_{\alpha^{(m)}_{0}}) \leq \alpha^{(m)}_{0}) \to 1 $ as $G^{(m)} \to \infty$, the variance $\text{Var}(T_{g}^{(m)})$ is uniformly upper bounded and lower bounded away from zero, and $\sum_{g\in\mathcal{H}_0^{(m)}}e_{gr}^{(m)}/G^{(m)}$ is uniformly integrable.
	If the group size $|\mathcal{A}_{g}^{(m)}|$ is uniformly bounded, then under Assumption~\ref{assumption1} and \ref{assumption2}, $\{e_{gr}^{(m)}\}_{g\in[G^{(m)}]}$ given in \eqref{DS-evalue} is a set of asymptotic relaxed e-values, and the eDS-filter guarantees that $\limsup_{N\to \infty}\mathrm{FDR}^{(m)}\leq \alpha^{(m)}$.
\end{theorem}

We note that the DS procedure for group detection is a special case of eDS-filter.
As a result, Theorem~\ref{theorem7} provides an alternative explanation for why the DS procedure for group detection controls the FDR in the asymptotic sense. When $M=1$ and $G^{(1)}=N$, the eDS-filter reduces to the eDS procedure for feature detection and, therefore, can be viewed as an alternative to MDS~\cite{Dai2022}.
The following remark clarifies the differences between MDS and eDS procedure.
\begin{remark}[Comparison with MDS~\cite{Dai2022}]
	The MDS procedure computes the estimated inclusion rates $\widehat{I}_{j}$ by $\widehat{I}_{j} =\sum_{r=1}^{R}\mathbb{I}\{j \in \mathcal{G}^{(1)}_{r}(\alpha^{(1)})\}/(R|\mathcal{G}^{(1)}_{r}(\alpha^{(1)})|)$ for $j\in[N]$ and selects the set $\{j: \widehat{I}_{j} \geq \widehat{I}_{\widehat{k}}\}$, where
    $\widehat{k}=\max\{k:\sum_{j=1}^k \widehat{I}_{(j)}\leq \alpha^{(1)}\}$, $ \widehat{I}_{(1)}\leq \dots \leq \widehat{I}_{(N)}$.
	In comparison, the generalized e-value and the estimated inclusion rate can be regarded as the weighted selection frequencies of features, with each assigned a distinct weight. The way of computing the cutoff also differs. With regard to the assumptions, the eDS procedure does not rely on the ranking consistency needed by MDS to ensure FDR control.
\end{remark}

\subsection{eDS+gKF-filter}\label{section3.2}
We introduce a flexible variant, the eDS+gKF-filter, which implements the SFEFP framework by specifying the feature detection procedure $\mathcal{G}^{(1)}$ as the DS method and assigning the group detection procedure $\mathcal{G}^{(m)}$ to the group knockoff filter for $m=2,\dots,M$.
The idea behind eDS+gKF-filter is as follows.
On one hand, if the features are highly correlated, the DS procedure performs better than the knockoff procedure~\cite{Dai2022}. However, if the signals within the group are very sparse, the DS statistic for group detection~\eqref{group statistics} can be powerless, which is caused by adopting the average signal strength as a representative of the group signal strength.
In this case, leveraging group knockoff for group detection is a better choice.
This motivates the development of the eDS+gKF-filter.

We outline the steps for implementing the \textit{eDS+gKF-filter} for linear models. Assuming that $Y=\beta_{1}X_{1}+\dots+\beta_{N}X_{N}+\epsilon$, where $\epsilon  \sim N(0, \sigma^{2})$. At the first layer, we perform DS using the Lasso+OLS procedure~\cite{Dai2022}, illustrated as follows.
\begin{enumerate}
\item[Step 1.] Apply Lasso to data $(\boldsymbol{y}^{(1)}, \boldsymbol{X}^{(1)})$ to obtain $\widehat{\boldsymbol{\beta}}^{(1)}$ and $\mathcal{S}_{\text{lasso}}=\{j: \widehat{\beta}^{(1)}_{j} \neq 0\}$;
\item[Step 2.] Restricted to data $\mathcal{S}_{\text{lasso}}$, apply OLS to $(\boldsymbol{y}^{(2)}, \boldsymbol{X}^{(2)})$ to obtain the estimate $\widehat{\boldsymbol{\beta}}^{(2)}$.
\end{enumerate}
Note that under other model assumptions of $Y|(X_1,\dots,X_N)$, such as generalized linear models, constructions of $\{(\widehat{\beta}_j^{(1)},\widehat{\beta}_{j}^{(2)})\}$ are discussed in \cite{Dai2022, dai2023scale}. Next, we construct the test statistic $T_{j}^{(1)} = W_{j} = \text{sign}(\widehat{\beta}_{j}^{(1)}\widehat{\beta}_{j}^{(2)})f(|\widehat{\beta}^{(1)}_{j}|,|\widehat{\beta}^{(2)}_{j}|)$. For the random-design matrix $\boldsymbol{X}$, Dai et al.~\cite{Dai2022} demonstrated that both the sure screening property of Lasso and the weak dependence (Assumption~\ref{assumption2} when $m=1$) hold under regular conditions, and further asserted that the fixed-design setting can also be established.
Given that the sure screening property holds, Assumption~\ref{assumption1} is thereby ensured by the OLS.
Consequently, with Theorem~\ref{theorem7} for each replication $r\in[R]$, the set $\{e_{jr}^{(1)}\}_{j\in[N]}$ given in Equation~(\ref{DS-evalue}) is a set of asymptotic relaxed e-values.

For the subsequent layers, we compute group knockoff statistics following \cite{2016A} if $n\geq N$; otherwise, we can adopt the constructions developed in \cite{Candes2018, chu2023second, katsevich2019multilayer, sesia2019gene, spector2022powerful}, which require assumptions on the joint distribution of the features. Specifically, for each $m\in\{2,\dots,M\}$, $r\in[R]$, the group knockoff procedure $\mathcal{G}^{(m)}_{r}$ constructs a test statistic $T_{gr}^{(m)}$ for $H_{g}^{(m)}$, which satisfies two key properties: (a) if $H_{g}^{(m)}$ is false, then $T_{gr}^{(m)}$ tends to take a large positive value; (b) the martingale property holds that 
\begin{equation}\label{martingale}
\mathbb{E}\left[\frac{\#\left\{g\in\mathcal{H}_{0}^{(m)}: T_{gr}^{(m)} \geq t^{r}_{\alpha_{0}^{(m)}}\right\}}{1+\#\left\{g\in[G^{(m)}]: T_{gr}^{(m)} \leq -t^{r}_{\alpha_{0}^{(m)}}\right\}}\right] \leq 1,
\end{equation}
where
\begin{equation}\label{knockoff-t}
t^{r}_{\alpha_{0}^{(m)}} = \inf\left\{t >0:  \frac{1+\#\left\{g : T_{gr}^{(m)}< -t  \right\}} { \#\left\{g : T_{gr}^{(m)} >t   \right\} \bigvee \text{1}} \leq \alpha_{0}^{(m)}\right\}.
\end{equation}
The group knockoff procedure $\mathcal{G}^{(m)}_{r}$ selects the set $\left\{g\in[G^{(m)}]: T_{gr}^{(m)} \geq t^{r}_{\alpha_{0}^{(m)}} \right\}$, and false discovery control is guaranteed by Equation~(\ref{martingale}) and (\ref{knockoff-t}). 
Based on Equation (\ref{martingale}), we can construct a set of relaxed e-values by
\begin{equation}
e_{gr}^{(m)}=G^{(m)} \cdot \frac{\mathbb{I} \left \{ T_{gr}^{(m)} \geq t^{r}_{\alpha_{0}^{(m)}} \right \} }{1+ \#\left\{g: T_{gr}^{(m)}\leq -t^{r}_{\alpha_{0}^{(m)}}\right\} }.
\end{equation}
For the target FDR level vector $(\alpha^{(1)},\dots,\alpha^{(M)})$, the rejection set of the eDS+gKF-filter is then given by the output of the generalized e-filter, which takes as input the generalized e-values obtained from Lasso+OLS and the group knockoff filter. 
The multilayer FDR control is ensured by the properties of the obtained generalized e-values.

In practical applications, it is recommended to slightly relax the target FDR levels to achieve higher power. This is inspired by Katsevich and Sabatti \cite{katsevich2019multilayer}, who demonstrated that the $\text{MKF}(c)+$ satisfies $\text{FDR}^{(m)} \leq c_{\text{kn}}\alpha^{(M)}/ c$ for $m\in[M]$, where $c_{\text{kn}} = 1.93$ and $c>0$. They suggest that practitioners deploy MKF(1)+ (denoted by MKF+ for simplicity) rather than $\text{MKF}(c_{\text{kn}})+$ without concern of FDR control, as multilayer methods are generally more conservative in practice than theoretically anticipated.
For the same reason, we also recommend that practitioners appropriately relax the target FDR levels when applying the eDS-filter and eDS+gKF-filter, etc.

\subsection{Other applications and possible extensions of SFEFP}\label{section3.3}
\subsubsection*{SFEFP with other detection procedures}
Gablenz et al. \cite{gablenz2023catch} developed e-MKF by leveraging the knockoff e-values and the e-filter, which is consistent with applying FEFP to knockoff procedures.
For convenience of distinction, we denote this method by *e-MKF.
By Theorem~\ref{onebitproperty}, *e-MKF has the risk of zero power.
By using SFEFP, we can extend *e-MKF to a stable and powerful version (termed e-MKF). Our simulation results also show that the stabilized method e-MKF consistently has higher power than *e-MKF. 
 Furthermore, in the Section F of the Supplementary Material, we develop asymptotical relaxed e-values for the GM method~\citep{Xing2023} and the SAS framework~\citep{wang2025adaptive}. Therefore, other multilayer filtering procedures can be developed by utilizing the SFEFP and these base detection procedures.

\subsubsection*{Extensions to time series data}
Time series data are prevalent in many practical applications such as economics.
Chi et al.~\cite{chi2021high} proposed the method of time series knockoffs inference (TSKI), and demonstrated that FDR can be asymptotically controlled at a preset level under regular conditions.
The TSKI procedure relies on the ideas of subsampling and robust knockoff~\cite{2018Robust}, and the control of FDR depends on the property of robust e-values. 
Indeed, when $\mathcal{G}^{(1)}$ with $M = 1$ corresponds to the robust knockoffs for FEFP, the TSKI procedure without subsampling is consistent with FEFP.
The framework of TSKI with subsampling is also similar to our proposed SFEFP.
The key difference is that we use the entire dataset in each replication, while TSKI utilizes different disjointed subsamples to address the complex serial dependence.
It is straightforward to establish multilayer FDR control for time series data by modifying SFEFP to subsampling setting and constructing robust e-values for group knockoff filter.

\section{Simulation}\label{section5_simulation}
Simulation studies are conducted to demonstrate the power enhancement brought about by combining different methods and the tie-breaking via derandomization.
Performances under low and high dimensions are explored, respectively.

\subsection{Basic settings}
We generate data from the linear model $\boldsymbol{y} = \boldsymbol{X}\boldsymbol{\beta} + \boldsymbol{\epsilon}$, where $\boldsymbol{\epsilon} \sim N(\textbf{0}, \boldsymbol{I}_{n})$ and $\boldsymbol{X} \in \mathbb{R}^{n \times N}$. Consider that there are $M=2$ layers: one layer with $N$ individual features and the other one with $G$ groups, each containing $N/G$ features.
The design matrix $\boldsymbol{X}$ is constructed by independently sampling each row from the distribution $N(\textbf{0},\Sigma_{\rho})$, where $\Sigma_{\rho}$ is a block-diagonal matrix composed of $G$ Toeplitz submatrices, each with correlation $\rho$. These submatrices describe the covariance structure within each group of features. Specifically, each block along the diagonal is 
\[
\begin{bmatrix}
	1 & \frac{(G^{'}-2)\rho}{G^{'}-1} & \frac{(G^{'}-3)\rho}{G^{'}-1} & \dots & \frac{\rho}{G^{'}-1}& 0  \\
	\frac{(G^{'}-2)\rho}{G^{'}-1}  & 1 & \frac{(G^{'}-2)\rho}{G^{'}-1} & \dots & \frac{2\rho}{G^{'}-1}& \frac{\rho}{G^{'}-1}  \\
	\vdots& & \ddots & & & \vdots\\
	0 & \frac{\rho}{G^{'}-1} & \frac{2\rho}{G^{'}-1} &\dots &  \frac{(G^{'}-2)\rho}{G^{'}-1} & 1
\end{bmatrix},
\]
where $G^{'} = N \slash G$ represents the group size.
Note that this construction closely aligns with many real-world datasets where features within a group are highly correlated while features between groups exhibit near-zero correlation. 
To determine $\mathcal{H}_{1}$, we begin by randomly selecting $K$ groups, followed by randomly selecting $|\mathcal{H}_{1}| $ elements from features within these $K$ groups. On average, a relevant group consists of $|\mathcal{H}_1|/K$ relevant features.
For $j \in \mathcal{H}_{1}$, we sample $\beta_{j}$ from the distribution $N(0, \delta\sqrt{\text{log}~N/n})$, where $\delta$ represents the signal strength of relevant features.

We consider four methods, e-MKF, eDS-filter, eDS+gKF and KF+gDS, designed by applying SFEFP to different $\mathcal{G}^{(m)}$ with $R=50$. 
To show the power enhancement brought out by leveraging non-one-bit generalized e-values, we set the number of replications for each layer as $R=1$, and the corresponding methods are denoted by *e-MKF, *eDS-filter, *eDS+gKF, and *KF+gDS, respectively.
In this case, *e-MKF is exactly the method proposed by \cite{gablenz2023catch}.
We also compare proposed methods to the MKF+~\cite{katsevich2019multilayer}. Details of methods are illustrated as follows.
For low-dimensional scenarios, we use the fixed design (group) knockoffs with the signed-max function. For high-dimensional settings, we use the R package \textit{knockoffsr}
\footnote{The R package \textit{knockoffsr} is  available at \href{https://github.com/biona001/knockoffsr/tree/main}{https://github.com/biona001/knockoffsr/tree/main}.}\cite{chu2023second} to construct model-X (group) knockoffs.
For the DS procedure for individual selection, we run the Lasso+OLS procedure as the setting of \cite{Dai2022}. For the DS procedure for group selection, we compute $T_{g}^{(m)}$ as \eqref{group statistics}, where statistics $W_{j}$ are obtained through the above Lasso+OLS procedure.
The target FDR levels is set as $\alpha^{(1)}=\alpha^{(2)} = 0.2$ for all methods. The original FDR levels is simply set as $\alpha_{0}^{(m)} = \alpha^{(m)} / 2 = 0.1$ for $m\in[M]$.

\subsection{Comparison under low dimensions}
\begin{figure}[htbp]
	\centering
    \includegraphics[width=45em]{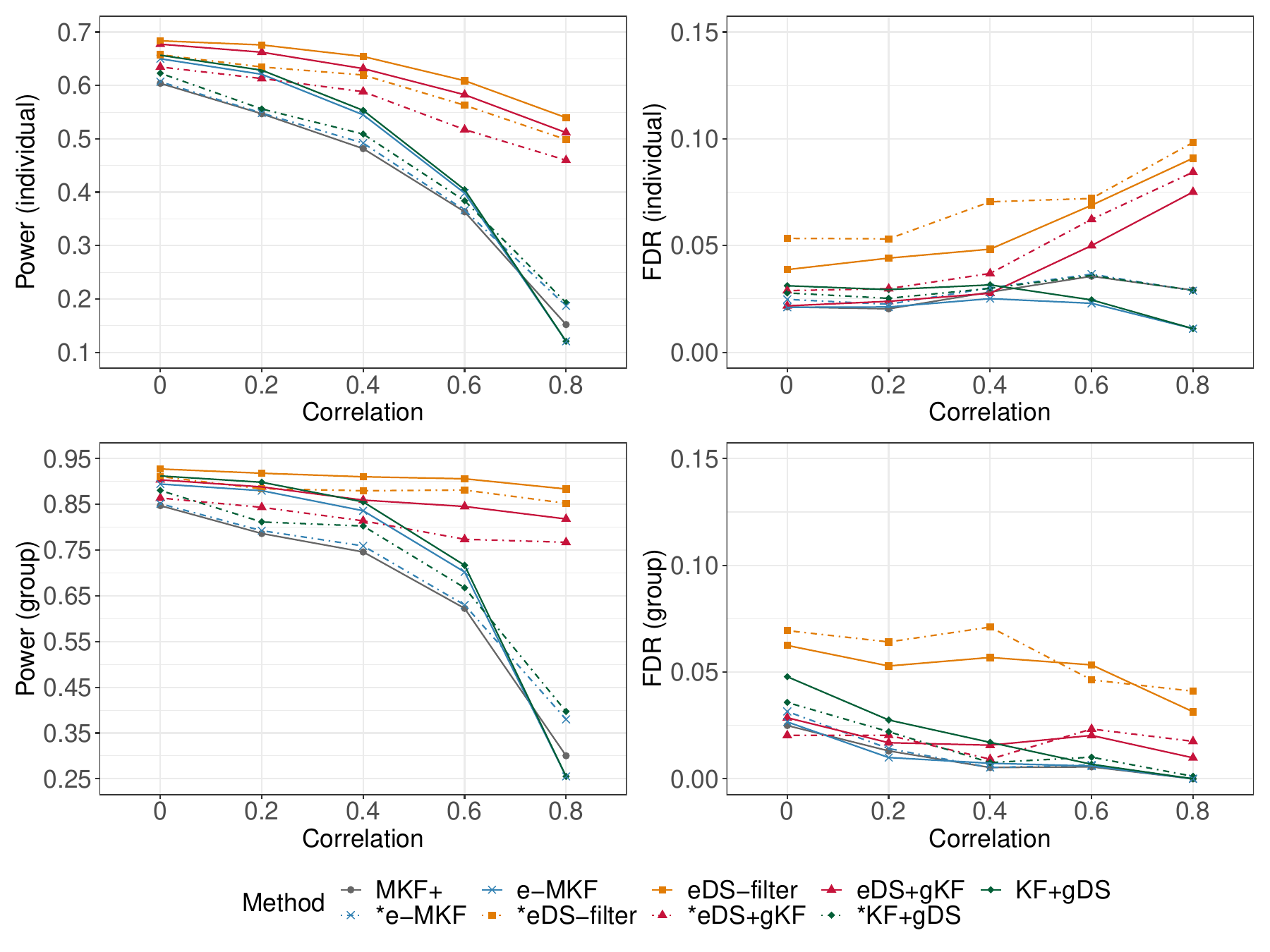}
	\caption{Simulation results for methods MKF+, e-MKF, eDS-filter, eDS+gKF, KF+gDS, *e-MKF, *eDS-filter, *eDS+gKF, and *KF+gDS under different correlations, while fixing the signal strength $\delta$ as 3.}
	\label{lowsignal}
\end{figure}

\begin{figure}[htbp]
	\centering
    \includegraphics[width=45em]{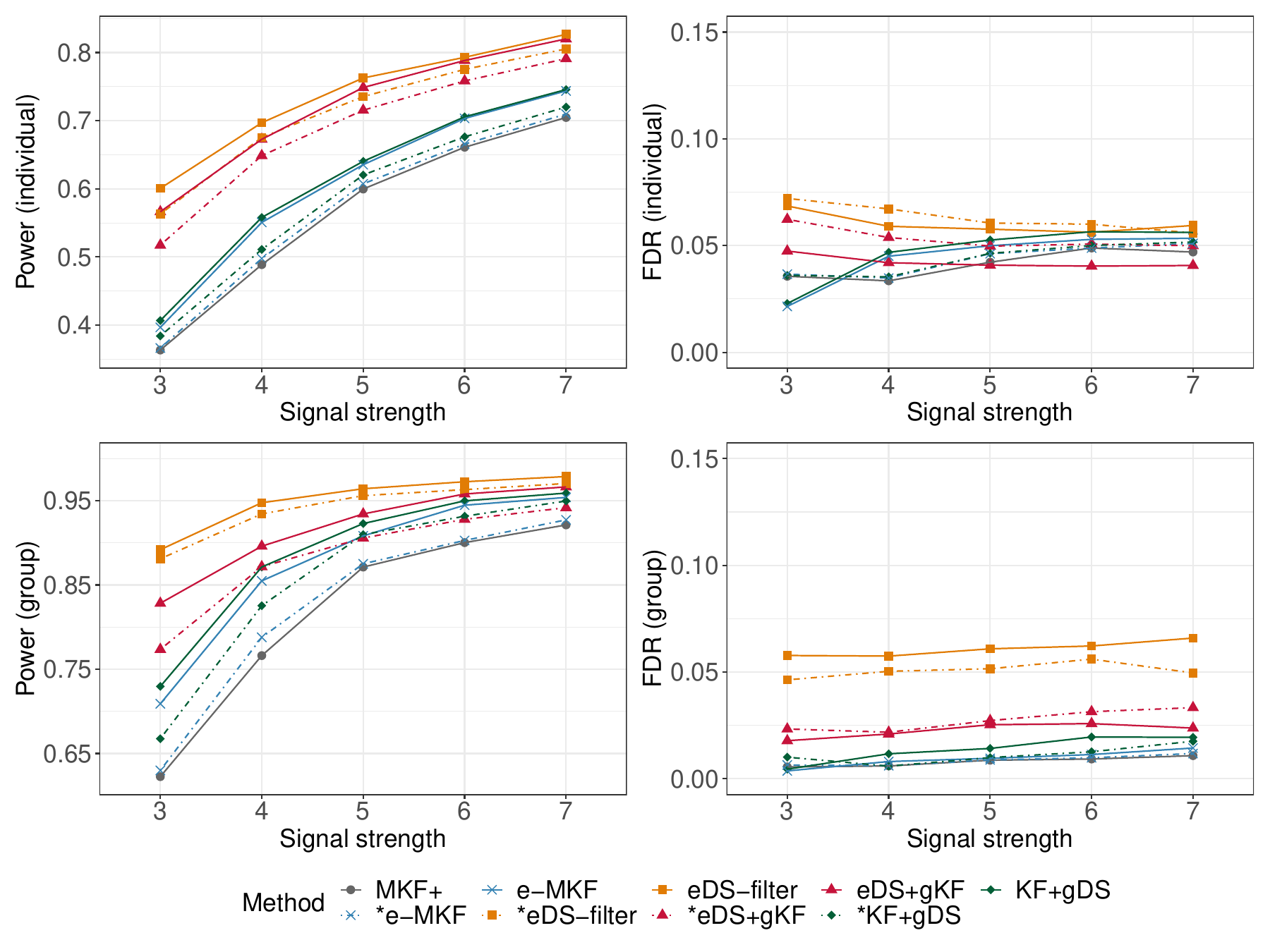}
	\caption{Simulation results for nine methods under different signal strength, while fixing the correlation $\rho$ as 0.6.}
	\label{highcorr}
\end{figure}

 Consider $n=1600$, $N=800$, $G=80$, $|\mathcal{H}_1|=60$, and $K=20$. Therefore, each group contains 10 features, and on average a relevant group consists of 3 relevant features. The theoretical FDR level is set to $\alpha^{(m)}=0.2$ for $m=1,2$.
We first fix the signal strength at $\delta = 3$ and vary the correlation $\rho$ in the set $\{0, 0.2, 0.4, 0.6, 0.8\}$. 
Figure~\ref{lowsignal} illustrates the realized FDR and power at both the individual and group layers as a function of the correlation $\rho$.
Next, we fix $\rho = 0.6$ and vary $\delta$ in the set $\{3,4,5,6,7\}$.
Figure~\ref{highcorr} presents these quantities as a function of the signal strength $\delta$, with a fixed correlation $\rho = 0.6$.
All results are averaged over 50 independent trials.

Figures \ref{lowsignal} and \ref{highcorr} highlight two attractive advantages of SFEFP. First, under high correlations, both eDS-filter and eDS+gKF-filter simultaneously achieve higher power at multiple resolutions than MKF+, e-MKF, and KF+gDS. This empirical evidence demonstrates the flexibility of SFEFP in accommodating different detection procedures at different resolutions, thus enhancing the detection power.
Second, by comparing the dashed and solid lines, we observe that the stabilization step consistently enhances the detection power across almost all settings. This demonstrates the advantage of stabilization under multiple resolutions.

\subsection{Comparison under high-dimensional settings}
\begin{figure}[htbp]
	\centering
    \includegraphics[width=45em]{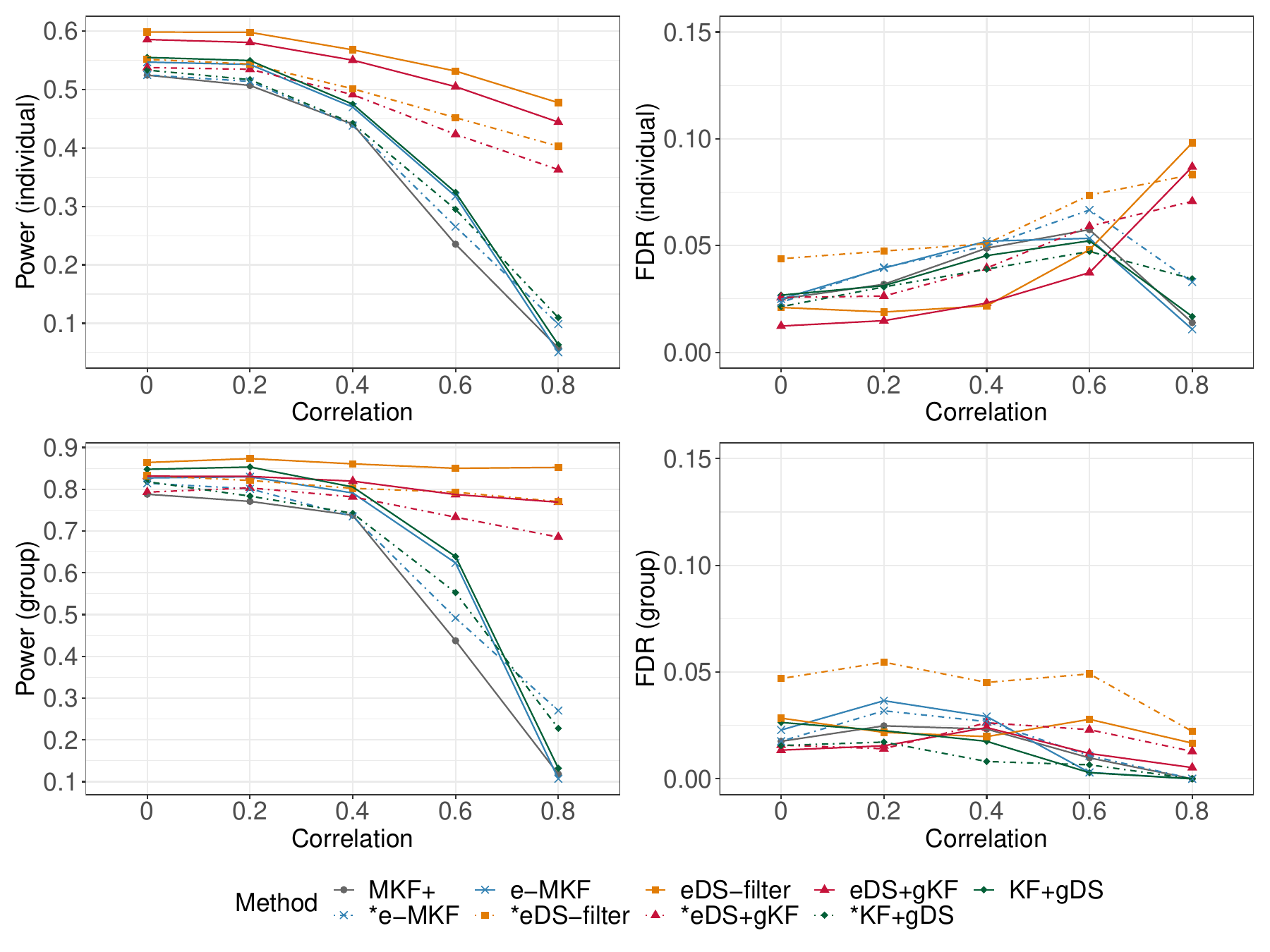}
	\caption{Simulation results for the high-dimensional setting under different correlations with $\delta=3$. }
	\label{highdim_high_correlation}
\end{figure}

\begin{figure}[htbp]
	\centering
    \includegraphics[width=45em]{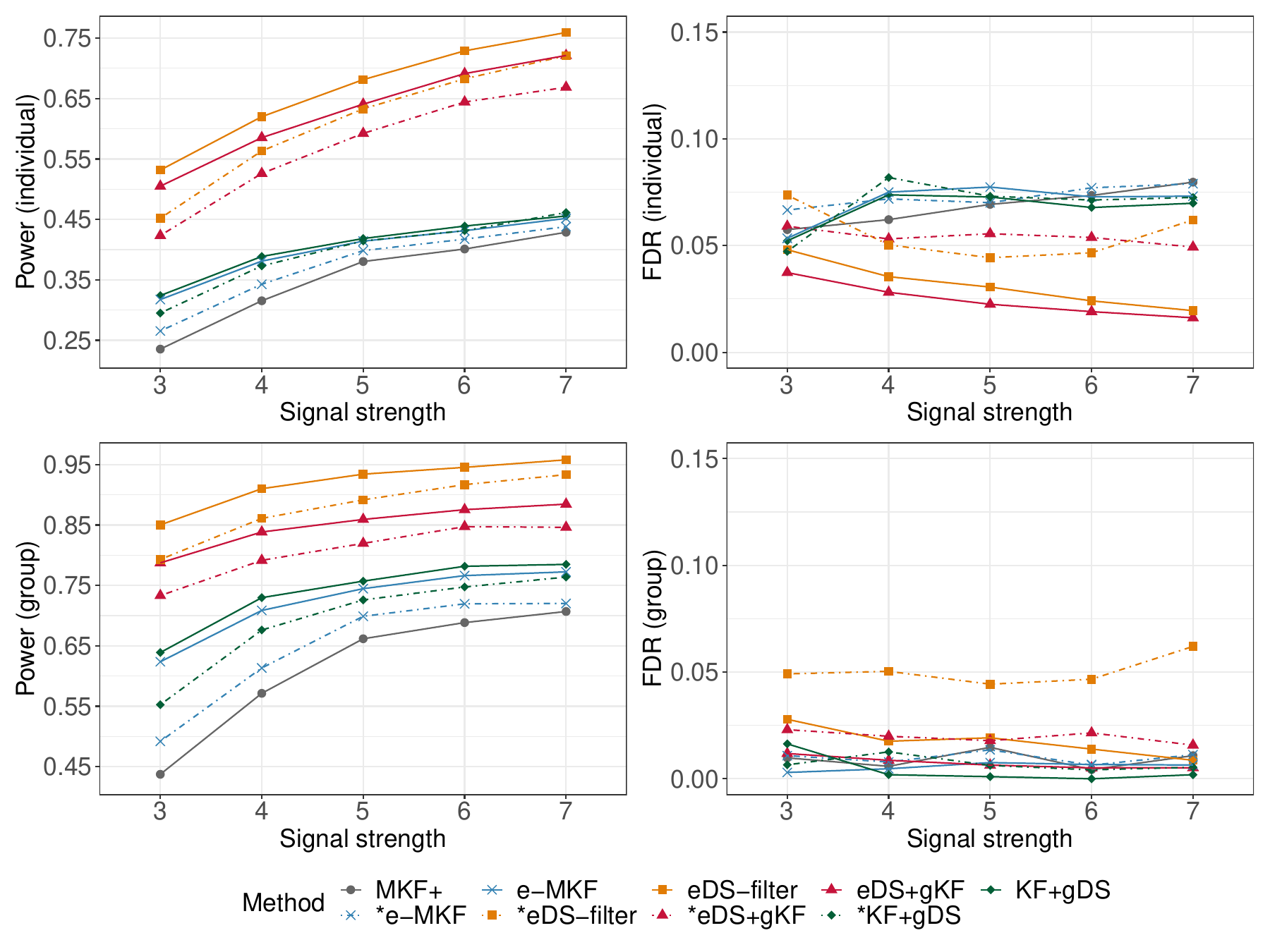}
	\caption{Simulation results for the high-dimensional setting under different signal strength with $\rho=0.6$.}
	\label{hihdim_highcorr}
\end{figure}

Consider a high-dimensional setting where $n=600$, $N=800$, $G=80$, $|\mathcal{H}_1|=60$, and $K=20$. We set $\alpha^{(m)}=0.2$ for $m=1,2$. Figures~\ref{highdim_high_correlation} and ~\ref{hihdim_highcorr} present the performance of the methods under different correlations and signal strength, respectively. All results are averaged over 50 independent trials. The analysis and conclusion in low-dimensional settings also hold in high-dimensional settings.

\section{Experiment on HIV mutation data}\label{section5}
We compare the proposed eDS-filter, e-MKF~\citep{gablenz2023catch}, MKF+~\citep{katsevich2019multilayer}, and the knockoff filter~\cite{Barber2015} in the study of identifying mutations and their clusters in the human immunodeficiency virus type 1 (HIV-1) associated with drug resistance.
The dataset\footnote{Data available online at \href{http://hivdb.stanford.edu/pages/published\_analysis/genophenoPNAS2006/}{http://hivdb.stanford.edu/pages/published\_analysis/genophenoPNAS2006/}.} previously analyzed in~\cite{Barber2015, Dai2022, lu2018deeppink, rhee2006genotypic, song2021variable}, contain drug resistance measurements in various drug classes and genotype information in HIV-1 samples.
Our analysis focuses on seven drugs (APV, ATV, IDV, LPV, NFV, RTV, SQV) for protease inhibitors (PIs) and four drugs (ABC, AZT, D4T, DDI) for nucleoside reverse transcriptase inhibitors (NRTIs). 
A notable distinction of our study from previous work is our dual focus: we aim not only to obtain reliable results for individual mutations (i.e., controlling individual-FDR), but also to ensure reliable findings at the position level (i.e., controlling group-FDR simultaneously).

\begin{table}
	\caption{Sample information for the seven PI-type drugs and the three NRTI-type drugs.}
	\centering
	\begin{tabular}{ccccc}
		\hline
		Drug type&Drug&Sample size&\# mutations&\# positions genotyped\\
		\hline
		&APV&767&201&65\\
		&ATV&328&147&60\\
		&IDV& 825&206&66\\
		PI&LPV&515&184&65\\
		&NFV&842&207&66\\
		&RTV& 793&205&65\\
		&SQV& 824&206&65\\
		\hline
		&ABC&623&283&105\\
		NRTI&AZT& 626&283&105\\
        &D4T&625&281&104\\
        &DDI& 628 &283&105\\
		\hline
	\end{tabular}
	\label{table1}
\end{table}

The information and preprocessing steps for the data set are described below.
The response variable $Y$ represents the log-fold increase of lab-tested drug resistance. The feature $X_{j}$ indicates whether the mutation $j$ is present, which is binary. Different mutations at the same location are treated as distinct features. Accordingly, mutations are naturally grouped based on their known locations.
Our goal is to identify which locations are more likely to harbor important mutations and which mutations influence drug resistance.	
For each drug, we preprocess the data by removing rows lacking drug resistance information and retaining mutations that appear more than three times in the sample for that drug. The sample size, number of mutations and number of groups (positions) are summarized in Table~\ref{table1}. Consistent with~\cite{Barber2015, Dai2022}, we assume a linear model between the response and features with no interaction terms.

For each drug, we apply the Knockoff+, MKF+~\cite{katsevich2019multilayer}, e-MKF~\cite{gablenz2023catch}, and eDS-filter to identify important mutations and their clusters (positions).
The $\text{Knockoff}+$ is included as a comparison to highlight the necessity of multilayer FDR control. The treatment-selected mutation (TSM) panels~\cite{rhee2005hiv} provide a good approximation to the ground truth of the real datasets.
In this study, we take the TSM panel as the reference standard to evaluate the performance of the methods under comparison.
The results for the PI-type drugs and the NRTI-type drugs are presented in Table~\ref{PI} and Table~\ref{table_NRTI}, respectively.

\begin{table}
	\centering
	\setlength{\tabcolsep}{2pt} 
	\caption{Results for PI drugs. ``True'' represents the number of true positives, i.e., the number of mutations or positions identified in the TSM panel for the PI class of treatments. ``False'' presents the number of false positives. The FDP is calculated as the ratio of the number of false positives to the total number of positives. The target FDR levels are $\alpha^{(1)}=\alpha^{(2)} = 0.2$. For eDS-filter, we set $R=50$ and $\alpha^{(m)}_0=\alpha^{(m)}/2$ for $m=1,2$. The best-performing method is highlighted in \textbf{bold}.}
	\begin{tabular}{cc|ccc|ccc}
		\hline
		Drug&Method&True (ind)&False (ind)&FDP (ind)&True (grp)&False (grp)&FDP (grp)\\
		\hline
		APV&KF+&27&9&0.250&18&7&0.280\\
		&MKF+&0&0&0&0&0&0\\
          &e-MKF& 0&0 &0  &0 &0 &0 \\
		&\textbf{eDS-filter}& \textbf{27} & \textbf{4} & \textbf{0.129} &\textbf{18} &\textbf{2} &\textbf{0.100}\\
		\hline
		ATV&KF+&19&6&0.240&19&1&0.050\\
		&MKF+&0&0&0&0&0&0\\
        &e-MKF&0 & 0& 0 & 0& 0& 0\\
		&\textbf{eDS-filter}& \textbf{18}&\textbf{1}&\textbf{0.053}&\textbf{18}&\textbf{0}&\textbf{0}\\
		\hline
		IDV&KF+&34&33&0.493&24&15&0.385\\
		&MKF+&26&3&0.103&17&0&0\\
         &e-MKF& 26& 4& 0.133 &18 &0 &0 \\
		&\textbf{eDS-filter}&\textbf{27}&\textbf{3}&\textbf{0.100}&\textbf{18}&\textbf{0}&\textbf{0}\\
		\hline
		LPV&KF+&27&8&0.229&20&3&0.130\\
		&MKF+&19&3&0.136&13&1&0.071\\
         &e-MKF&19 & 3& 0.136 &13 &1 &0.071 \\
		&\textbf{eDS-filter}&\textbf{23}&\textbf{3}&\textbf{0.115}&\textbf{15}&\textbf{0}&\textbf{0}\\
		\hline
		NFV&KF+&33&22&0.400&24&8&0.250\\
		&MKF+&0&0&0&0&0&0\\
          &e-MKF& 0& 0& 0 &0 &0 &0 \\
		&\textbf{eDS-filter}&\textbf{32}&\textbf{8}&\textbf{0.200}&\textbf{20}&\textbf{2}&\textbf{0.091}\\
		\hline
		RTV&KF+&19&5&0.208&12&2&0.143\\
		&MKF+&0&0&0&0&0&0\\
          &e-MKF& 0& 0& 0 & 0& 0& 0\\
		&\textbf{eDS-filter}&\textbf{25}&\textbf{7}&\textbf{0.219}&\textbf{17}&\textbf{2}&\textbf{0.105}\\
		\hline
		SQV&KF+&22&6&0.214&16&2&0.111\\
		&MKF+&0&0&0&0&0&0\\
          &e-MKF& 0& 0& 0 & 0& 0& 0\\
		&\textbf{eDS-filter}&\textbf{22}&\textbf{4}&\textbf{0.154} &\textbf{15}&\textbf{0}&\textbf{0} \\
		\hline
	\end{tabular}
	\label{PI}
\end{table}

\begin{table}
	\centering
	\setlength{\tabcolsep}{2pt} 
	\caption{Results for NRTI drugs. The target FDR levels are $\alpha^{(1)}=\alpha^{(2)} = 0.3$. For eDS-filter, we set $R=100$. The best-performing method is highlighted in \textbf{bold}.}
	\begin{tabular}{cc|ccc|ccc}
		\hline
		Drug&Method&True (ind)&False (ind)&FDP (ind)&True (grp)&False (grp)&FDP (grp)\\
		\hline        ABC&\textbf{KF+}&\textbf{14}&\textbf{3}&\textbf{0.176}&\textbf{14}&\textbf{2}&\textbf{0.125}\\
		&MKF+&0&0&0&0&0&0\\
       &e-MKF& 0& 0 & 0 & 0& 0&0 \\
		&eDS-filter&13&2&0.133&12&2&0.143\\
		\hline
        AZT&KF+&17&8&0.320&16&5&0.238\\
		&MKF+&11&0&0&10&0& 0\\
        &e-MKF&11&0 &0  &10 &0 &0 \\
		&\textbf{eDS-filter}&\textbf{15}&\textbf{1}&\textbf{0.063} &\textbf{14} & \textbf{0}& \textbf{0}\\
		\hline
		D4T&KF+&10&1&0.091&9&1&0.100\\
		&MKF+&0&0&0&0&0&0\\
         &e-MKF& 0&0 &0&0 &0 &0 \\
		&\textbf{eDS-filter}& \textbf{18}&\textbf{2}&\textbf{0.100}&\textbf{15}&\textbf{2}&\textbf{0.118}\\
		\hline
        DDI&KF+&0 & 0& 0& 0& 0& 0\\
		&MKF+&0&0&0 &0&0&0\\
        &e-MKF& 0& 0& 0 & 0& 0& 0\\
		&\textbf{eDS-filter}& \textbf{18}& \textbf{4}&\textbf{0.182} &\textbf{17} &\textbf{3} &\textbf{0.150} \\
		\hline
	\end{tabular}
	\label{table_NRTI}
\end{table}

For the seven PI-type drugs, the proposed eDS-filter performs consistently the best among the three methods.
Specifically, the method of KF+ fails to simultaneously control the FDP for the individual and group resolutions, especially for the drugs APV, ATV, IDV, LPV, and NFV.
In contrast, the eDS-filter almost always has fewer false discoveries and thus controls FDP at multiple resolution simultaneously, while maintaining comparable power as KF+.
For drug RTV, the eDS-filter even achieves a higher power.
The MKF+ and e-MKF have no discoveries for the drugs APV, ATV, NFV, RTV, and SQV. For drugs IDV and LPV, the eDS-filter still outperforms the MKF+ and e-MKF by achieving higher power and lower FDP.
Additionally, we compare the selection results of our methods with those obtained by DeepPINK~\cite{lu2018deeppink} on the individual resolution, since they aimed to control the individual FDR level to $0.2$.
We observe that the eDS-filter achieves a higher power than DeepPINK for all drugs except for ATV, and achieves a lower FDP for drugs ATV, IDV, LPV, and NFV.

For NRTI, the eDS-filter performs best for the drugs AZT, D4T, and DDI. 
Specifically, the MKF+ and e-MKF have no discoveries for the drugs ABC, D4T and DDI. The KF+ fails to control the FDP for the drug AZT and has no discoveries for the drug DDI.
In comparison, the eDS-filter consistently controls FDP at multiple resolutions and achieves satisfactory power. Especially for the drugs D4T and DDI, the eDS-filter achieves a significantly higher power.

In conclusion, the eDS-filter consistently performs significantly better than MKF+ \cite{katsevich2019multilayer} and e-MKF \cite{gablenz2023catch} by achieving higher power.
Compared with KF+, the eDS-filter has more precise discoveries and achieves multilayer FDR control with comparable or even higher power (e.g., drugs RTV, D4T, DDI).
Note that we simply set $\alpha^{(m)}_{0} = \alpha^{(m)}/2$ for $m=1,2$; therefore, higher power may be achieved by adjusting these parameters. 
The outstanding performance of the eDS-filter also demonstrated the ability of SFEFP to enable the design of tailored procedures for distinct contexts, and this flexibility in turn enhances the efficiency with which researchers can identify critical features. Indeed, it is highly feasible for practitioners to select the state-of-the-art method for each layer based on prior knowledge, thereby leading to more ideal results.

\section{Conclusion and discussion}\label{section6}
In this paper, we introduce SFEFP, a flexible, stable, and powerful detection procedure with multilayer FDR control.
Specifically, by developing generalized e-filter and generalized e-values, SFEFP combines diverse detection procedures to yield a selection set with FDR guarantees across multiple resolutions.
By stabilization, SFEFP can obtain more stable and powerful selection sets.
Several practical multilayer detection methods, such as the eDS-filter, are developed as illustrations of SFEFP. Simulations and case studies have shown that (1) through the flexible design, we can often achieve higher power, and (2) the stabilization step tends to enhance the detection power of one-bit inputs.

Our work is founded on the core idea that there is no one-size-fits-all method, and the optimal procedure varies depending on the situation. Consequently, it is imperative to provide practitioners with ample flexibility to develop methods informed by domain expertise.
The proposed SFEFP is also based on in-depth exploration of multilayer filtering, where one-bit input can cause a zero power disaster.
Guided by this objective, the primary strength of SFEFP lies in its flexibility and stability, enabling the creation of a diverse array of powerful multilayer detection methodologies with almost no cost.
In the realm of high-dimensional regression and scenarios where features are highly correlated, the establishment of the eDS-filter has demonstrated itself as a successful attempt.

Several meaningful issues are worth solving. 
First, the impact of the original FDR level vector on the power of FEFP and SFEFP requires further theoretical investigation. 
Second, it is important to explore the mild conditions under which a sharper FDR bound can be derived, since simulations indicate that SFEFP sometimes achieves significantly lower empirical FDR than the preset level.
Third, appropriate data-driven weights for different replications can improve the overall reliability of the results. Thus, how to achieve adaptive weights is of our interest. Techniques related to enhanced e-values, such as those in \cite{blier2024improved, lee2024boosting}, are also expected to enhance the power of SFEFP. We will leave these issues for future work.

\begin{acks}[Acknowledgments]
	Ruixing Ming was supported in part by the Characteristic \& Preponderant Discipline of Key Construction Universities in Zhejiang Province (Zhejiang Gongshang University--Statistics) and the Collaborative Innovation Center of Statistical Data  Engineering Technology \& Application. Zhanfeng Wang was supported in part by National Natural Science Foundation of China (No. 12371277, 12231017).
\end{acks}


\bibliographystyle{imsart-nameyear} 
\bibliography{aos}

\newpage


\begin{appendix}

\section{Proofs for Section 3}
\subsection{Proof of Proposition 1}
For any detection procedure $\mathcal{G} \in \mathcal{K}_{\text{finite}} \bigcup \mathcal{K}_{\text{asy}}$ and any target FDR level $\alpha \in(0,1)$, the rejection set is denoted by $\mathcal{G}(\alpha)$. Under the alternative to Definition 1 that
\[
t_{\alpha}= \sideset{}{}{\arg\max}_{t}^{}R(t), \quad \text{subject to}~\widehat{\text{FDP}} (t) = \frac{\widehat{V}(t)} {R(t)\vee 1} \leq \alpha,
\]
the rejection set is denoted by $\widetilde{\mathcal{G}}(\alpha)$.
We aim to prove $\mathcal{G}(\alpha) =\widetilde{\mathcal{G}}(\alpha)$.
On one hand, the fact $ \mathcal{G}(\alpha) \subseteq\widetilde{\mathcal{G}}(\alpha)$ is obvious, because for any $t$ we have $\widehat{V}(t) \leq \max\{\widehat{V}(t), \alpha\}$.
On the other hand, if $\widetilde{\mathcal{G}}(\alpha) \neq \emptyset$, then we have $ \mathcal{G}(\alpha) \supseteq\widetilde{\mathcal{G}}(\alpha)$. The reason is as follows. When $\widetilde{\mathcal{G}}(\alpha) \neq \emptyset$, if $\widehat{V}(t_{\alpha}) < \alpha$, then 
\[
 \frac{\widehat{V}(t_{\alpha}) \vee \alpha} {R(t_{\alpha})\vee 1} = \frac{
\alpha}{R(t_{\alpha})} \leq \alpha,
\]
which implies that $ \mathcal{G}(\alpha) \supseteq\widetilde{\mathcal{G}}(\alpha)$.
If $\widehat{V}(t_{\alpha}) \geq \alpha$, then we still have $ \mathcal{G}(\alpha) \supseteq\widetilde{\mathcal{G}}(\alpha)$, because $\max\{\widehat{V}(t_{\alpha}), \alpha\}= \widehat{V}(t_{\alpha}) $.
This concludes the proof of Proposition 1.
\subsection{Proof of Theorem 1}
The selection set of the generalized e-BH procedure is denoted by $\mathcal{S}_{\text{gebh}}(\alpha)$. We only need to consider the nontrivial case that $\mathcal{G}(\alpha) \neq \emptyset$. On one hand, for $j \notin \mathcal{G}(\alpha)$, we have
\[
e_{j} = N \cdot \frac{\mathbb{I}\{j \in \mathcal{G}(\alpha)\}}{\widehat{V}_{\mathcal{G}(\alpha)} \bigvee \alpha}=0.
\]
Since only those satisfying $e_{j} \geq N \slash (\alpha |\mathcal{S}_{\text{gebh}}(\alpha)|)$ can be selected by the generalized e-BH procedure when $\mathcal{S}_{\text{gebh}} \neq \emptyset$, we have $j \notin \mathcal{S}_{\text{gebh}}$.
Thus we have $\mathcal{S}_{\text{gebh}} \subseteq \mathcal{G}(\alpha)$. 
On the other hand, for $j \in \mathcal{G}(\alpha)$, we have
\[
e_{j} = N \cdot \frac{\mathbb{I}\{j \in \mathcal{G}(\alpha)\}}{\widehat{V}_{\mathcal{G}(\alpha)} \bigvee \alpha}\geq \frac{N}{\alpha R_{\mathcal{G}(\alpha)}}.
\]
According to the definition of the generalized e-BH procedure, we have $|\mathcal{S}_{\text{gebh}}(\alpha)| = R_{\mathcal{G}(\alpha)}$. Combining $\mathcal{S}_{\text{gebh}} \subseteq \mathcal{G}(\alpha)$, we have $\mathcal{S}_{\text{gebh}} = \mathcal{G}(\alpha)$.

\subsection{Proof of Proposition 2}
The proof is almost identical to the proof in~\cite{P-filter}.
The threshold calculation process can be viewed as a careful adjustment of the threshold $t^{(m)}$ for each $m \in [M]$.
This process terminates when the $\text{FDP}^{(m)}$ of all layers are under control.
It suffices to prove the following two aspects to establish the proposition.
The first aspect is the prudence of the algorithm, which ensures that $t_{k}^{(m)} \leq \widehat{t}^{(m)}$ for all $m,k$, where the vector $(t_{k}^{(1)}, \ldots, t_{k}^{(M)})$ represents the threshold vector after the $k$-th pass through the loop.
The second aspect is that $(t_{\text{final}}^{(1)}, \ldots, t_{\text{final}}^{(M)}) \in \mathcal{T}(\alpha^{(1)}, \ldots, \alpha^{(M)})$, where $(t_{\text{final}}^{(1)}, \ldots, t_{\text{final}}^{(M)})$ denotes the final output threshold vector.
The first aspect ensures that $(t_{\text{final}}^{(1)}, \ldots, t_{\text{final}}^{(M)}) \le (\widehat{t}^{(1)},\ldots,\widehat{t}^{(M)})$, with the inequality holding elementwise. 
Together with the definition of $(\widehat{t}^{(1)},\ldots,\widehat{t}^{(M)})$, the second aspect implies that $(t^{(1)}_{\text{final}}, \ldots, t^{(M)}_{\text{final}}) \ge  (\widehat{t}^{(1)},\ldots,\widehat{t}^{(M)})$. 
Therefore, we conclude that $(t^{(1)}_{\text{final}}, \ldots, t^{(M)}_{\text{final}}) =   (\widehat{t}^{(1)},\ldots,\widehat{t}^{(M)})$. 

Based on the initialization of the algorithm, we observe that  $t_{0}^{(m)} \leq \widehat{t}^{(m)}$ for all $m\in[M]$. 
We assum that $t_{k-1}^{(m)} \le \widehat{t}^{(m)}$ for all $m\in[M]$ and for any fixed $m$, we have $t_{k}^{(m')} \le \widehat{t}^{(m')}$ for $m' \in \{1, \dots, m-1\}$. 
We proceed to prove $t_{k}^{(m)} \leq \widehat{t}^{(m)}$. Based on our assumptions, we have
\[
\mathcal{S}^{(m)}(t_{k}^{(1)}, \ldots, t_{k}^{(m-1)}, t , t_{k-1}^{(m+1)},  \ldots, t_{k-1}^{(M)}) \supseteq \mathcal{S}^{(m)}(\widehat{t}^{(1)}, \ldots, \widehat{t}^{(m-1)}, t , \widehat{t}^{(m+1)}, \ldots, \widehat{t}^{(M)}) 
\]
for any $t >0$. Then we have
\[
\mathcal{S}^{(m)}(t_{k}^{(1)}, \ldots, t_{k}^{(m-1)}, \widehat{t}^{(m)} , t_{k-1}^{(m+1)},  \ldots, t_{k-1}^{(M)}) \supseteq \mathcal{S}^{(m)}(\widehat{t}^{(1)}, \dots,\widehat{t}^{(M)}),
\]
which further leads to
\[
\widehat{\text{FDP}}^{(m)}(t_{k}^{(1)},\dots,t_{k}^{(m-1)},\widehat{t}^{(m)},t_{k-1}^{(m+1)},\dots,t_{k-1}^{(M)}) \leq \widehat{\text{FDP}}^{(m)}(\widehat{t}^{(1)},\dots, \widehat{t}^{(M)}).
\]
By the definition of the algorithm that 
\[
t_{k}^{(m)} = \min \left\{t \in [t_{k-1}^{(m)},+ \infty) : \frac{G^{(m)} \frac{1}{t}}{ 1  \bigvee 
	\left| \mathcal{S}^{(m)}(t^{(1)},\ldots,t^{(m-1)},t,t^{(m+1)},\ldots,t^{(M)})\right| } \le \alpha^{(m)}\right\},
\]
we have $t_{k}^{(m)} \le \widehat{t}^{(m)}$. 
This completes the first part of the proof.
We then proceed to the second part of the proof.
Since the vector $(t_{\text{final}}^{(1)}, \ldots, t_{\text{final}}^{(M)} )$ is the output of the algorithm, we have 
\[
\frac{G^{(m)} \frac{1}{t_{\text{final}}^{(m)}}}{ 1  \bigvee 
	\left| \mathcal{S}^{(m)}(t_{\text{final}}^{(1)}, \ldots, t_{\text{final}}^{(M)} )\right| } \le \alpha^{(m)} \quad \text{for all } m \in [M].
\]
By the definition of $\mathcal{T}(\alpha^{(1)}, \ldots, \alpha^{(M)})$, we have $(t_{\text{final}}^{(1)}, \ldots, t_{\text{final}}^{(M)} ) \in \mathcal{T}(\alpha^{(1)}, \ldots, \alpha^{(M)})$. This completes the second part of the proof.


\subsection{Proof of Lemma 1}
For the generalized e-filter, the maximizer $\widehat{t}^{(m)}$ lies in the set $\left\{G^{(m)}\slash (\alpha^{(m)}k): k\in[G^{(m)}]\right\} \cup \left\{\infty\right\}$ for each $m\in[M]$. For any fixed $m\in[M]$, we first consider the case where $\widehat{t}^{(m)} \in \left\{G^{(m)}\slash (\alpha^{(m)}k): k\in[G^{(m)}]\right\} $. We have 
\begin{equation*}
	\begin{split}
		&\quad~\mathbb{E}\left[\text{FDP}^{(m)}(\widehat{t}^{(1)},\dots,\widehat{t}^{(M)}) \right]\\
		& = \mathbb{E}\left[\frac{\sum_{g\in \mathcal{H}^{(m)}_{0}} \mathbb{I} \left\{ g \in \mathcal{S}^{(m)}(\widehat{t}^{(1)},\ldots,\widehat{t}^{(M)}) \right\} }{ \left| \mathcal{S}^{(m)}(\widehat{t}^{(1)},\ldots,\widehat{t}^{(M)})\right| \vee 1}\right] \\
		& = \mathbb{E}\left[ \sum\limits_{g\in \mathcal{H}^{(m)}_{0}} \frac{ G^{(m)} \frac{1}{\widehat{t}^{(m)}} }{ \left| \mathcal{S}^{(m)}(\widehat{t}^{(1)},\ldots,\widehat{t}^{(M)})\right| \vee 1} \cdot 
		\frac{\mathbb{I} \left\{ g \in \mathcal{S}^{(m)}(\widehat{t}^{(1)},\ldots,\widehat{t}^{(M)}) \right\}}{G^{(m)} \frac{1}{\widehat{t}^{(m)}}} \right]   \\
		& \le \frac{\alpha^{(m)}}{G^{(m)}} \cdot \mathbb{E}\left[ \sum\limits_{g\in \mathcal{H}^{(m)}_{0}} \widehat{t}^{(m)} \cdot \mathbb{I} \left\{ g \in \mathcal{S}^{(m)}(\widehat{t}^{(1)},\ldots,\widehat{t}^{(M)}) \right\} \right] \\ 
		& \le  \sum\limits_{g\in \mathcal{H}^{(m)}_{0}} \frac{\alpha^{(m)}}{G^{(m)}} \cdot \mathbb{E} \left[\widehat{t}^{(m)} \cdot \mathbb{I} \left\{ e_{g}^{(m)} \ge \widehat{t}^{(m)}  \right\} \right] \\
		& \leq \sum\limits_{g\in \mathcal{H}^{(m)}_{0}}  \frac{\alpha^{(m)}}{G^{(m)}} \cdot\mathbb{E} \left[e_{g}^{(m)} \cdot \mathbb{I} \left\{ e_{g}^{(m)} \ge \widehat{t}^{(m)}  \right\} \right] \\
		& \le \alpha^{(m)}\cdot \frac{1}{G^{(m)}}\sum\limits_{g\in \mathcal{H}^{(m)}_{0}} \mathbb{E}[e_{g}^{(m)}].
	\end{split}
\end{equation*} 		
If $\{e_{g}^{(m)}\}_{g\in [G^{(m)}]}$ is a set of e-values, then we have $\mathbb{E}\left[\text{FDP}^{(m)}(\widehat{t}^{(1)},\dots,\widehat{t}^{(M)}) \right] \leq \pi_{0}^{(m)}\alpha^{(m)}$.
If $\{e_{g}^{(m)}\}_{g\in [G^{(m)}]}$ is a set of relaxed e-values, then we have $\mathbb{E}\left[\text{FDP}^{(m)}(\widehat{t}^{(1)},\dots,\widehat{t}^{(M)}) \right] \leq \alpha^{(m)}$.
If $\{e_{g}^{(m)}\}_{g\in [G^{(m)}]}$ represents a set of asymptotic e-values or asymptotic relaxed e-values, then the expectation $\mathbb{E}\left[\text{FDP}^{(m)}(\widehat{t}^{(1)},\dots,\widehat{t}^{(M)}) \right]$ is asymptotically controlled under $\overline{\pi}_{0}^{(m)}\alpha^{(m)}$ or $\alpha^{(m)}$, respectively, as $(n, G^{(m)}) \to \infty$. 

Next, we examine the case where $\widehat{t}^{(m)} = \infty$ for any fixed $m\in[M]$. 
If $\{e_{g}^{(m)}\}_{g\in [G^{(m)}]}$ is a set of e-values or relaxed e-values, then for any fixed $g\in \mathcal{H}_{0}^{(m)}$, we have $e_{g}^{(m)} < \infty$ almost surely because $e_{g}^{(m)}$ is nonnegative and $\mathbb{E}[e_{g}^{(m)} ]$ is upper bounded. 
Consequently, no features are selected almost surely because we have $e_{g}^{(m)} < t^{(m)}$ almost surely. 
This implies that we have $ \mathbb{E}\left[\text{FDP}^{(m)}(\widehat{t}^{(1)},\dots,\widehat{t}^{(M)}) \right]= 0$ almost surely. 
If $\{e_{g}^{(m)}\}_{g\in [G^{(m)}]}$ represents a set of asymptotic e-values, we also have $e_{g}^{(m)} < \infty$ almost surely. 
The remainder of the proof follows a similar manner as described above.

\subsection{Proof of Theorem 2}
For any fixed $m\in[M]$, it can be observed that once the set $\{e_{g}^{(m)}\}_{g \in G^{(m)}}$ is obtained, the subsequent steps follow exactly the same procedure as the generalized e-filter. 
Specifically, we translate the results of detection procedure $\mathcal{G}^{(m)}$ to obtain the set $\{e_{g}^{(m)}\}_{g\in [G^{(m)}]}$ and subsequently use these as inputs to the generalized e-filter.
Referring to the proof of Lemma 1, we observe that the control of $\text{FDR}^{(m)}$ depends on $\sum_{g\in\mathcal{H}^{(m)}_{0}}\mathbb{E}[e_{g}^{(m)}]$. When $\mathcal{G}^{(m)} \in \mathcal{K}_{\text{finite}}$, we have 
\begin{equation*}
	\begin{split}
		\sum\limits_{g \in \mathcal{H}_{0}^{(m)}}\mathbb{E}[e_{g}^{(m)}]
		&=\mathbb{E}\left[G^{(m)} \cdot \frac{\sum\limits_{g \in \mathcal{H}_{0}^{(m)}} \mathbb{I}\left \{g \in \mathcal{G}^{(m)}(\alpha_{0}^{(m)})\right\}}{\widehat{V}_{\mathcal{G}^{(m)}(\alpha_{0}^{(m)})} \bigvee \alpha_{0}^{(m)}}\right]\\
		&=G^{(m)}\cdot \mathbb{E}\left[\frac{V_{\mathcal{G}^{(m)}(\alpha_{0}^{(m)})}}{\widehat{V}_{\mathcal{G}^{(m)}(\alpha_{0}^{(m)})} \bigvee \alpha_{0}^{(m)}}\right]\\
		&\leq G^{(m)}.
	\end{split}
\end{equation*} 
This means that $\{e_{g}^{(m)}\}_{g\in[G^{(m)}]}$ is a set of relaxed e-values. According to Lemma 1, we have $\text{FDR}^{(m)} \leq \alpha^{(m)}$.
Since the group size $|\mathcal{A}_{g}^{(m)}|$ is uniformly bounded, we have $G^{(m)} \to \infty$ as $(n, N)\to \infty$ at a proper rate.
Thus, if $\mathcal{G}^{(m)} \in \mathcal{K}_{\text{asy}}$, we have 
\[
\frac{1}{G^{(m)}}\sum_{g \in \mathcal{H}_{0}^{(m)}}\mathbb{E}[e_{g}^{(m)}] \leq 1 
\]
as $(n, N)\to \infty$ at a proper rate.
This further leads that the FEFP satisfies $\text{FDR}^{(m)} \leq \alpha^{(m)}$ as $(n, N)\to \infty$ at a proper rate.

\subsection{Proof of Theorem 3}
We first prove the necessary and sufficient conditions.
Let $t^{(m)}=G^{(m)}/\widehat{V}^{(m)}_{\mathcal{G}^{(m)}(\alpha_0^{(m)})}$ for $m\in[M]$. 
If $g\in\mathcal{G}^{(m)}(\alpha_0^{(m)})$, then $e_{g}^{(m)}=t^{(m)}$; otherwise, $e_{g}^{(m)}=0 <t^{(m)}$.
Therefore, the selection set under this threshold vector is 
\[
\mathcal{S}_{\text{init}}^{(m)}=\left\{g\in G^{(m)}: \mathcal{A}_{g}^{(m)}\bigcap \left[\bigcap_{l=1}^{M}\left\{j:e_{h(l,j)}^{(l)}>0\right\}\right]\neq \emptyset\right\}.
\]
\begin{itemize}
\item If the event $\bigcap_{m=1}^M\left\{\widehat{V}^{(m)}_{\mathcal{G}^{(m)}(\alpha_0^{(m)})}\leq\alpha^{(m)}|\mathcal{S}_{\text{init}}^{(m)}|\right\}$ occurs, then 
\[
\widehat{\text{FDP}}^{(m)}(t^{(1)},\dots,t^{(M)})=\frac{G^{(m)}/t^{(m)}}{|\mathcal{S}^{(m)}_{\text{init}}|}=\frac{\widehat{V}^{(m)}_{\mathcal{G}^{(m)}(\alpha_0^{(m)})}}{|\mathcal{S}^{(m)}_{\text{init}}|}\leq \alpha^{(m)}~\text{for all } m\in [M].
\]
This implies that $(t^{(1)},\dots,t^{(M)})\in \mathcal{T}(\alpha^{(1)},\dotsm\alpha^{(M)})$.
By the monotonicity of $\mathcal{S}$, we have $\mathcal{S}(t^{(1)},\dots,t^{(M)})\subseteq \mathcal{S}(\widehat{t}^{(1)},\dots,\widehat{t}^{(m)})$ since $\widehat{t}^{(m)}\leq t^{(m)}$.
On the other hand, by the definition of $\mathcal{S}^{(m)}$, $\mathcal{S}^{(m)}_{\text{init}}$ is the largest selection set since all non-zero $e_{g}^{(m)}$ are not less than the threshold.
Therefore, we have $\mathcal{S}(\widehat{t}^{(1)},\dots,\widehat{t}^{(M)})=\mathcal{S}(G^{(1)}/\widehat{V}_{\mathcal{G}^{(1)}(\alpha_0^{(1)})},\dots,G^{(M)}/\widehat{V}_{\mathcal{G}^{(M)}(\alpha_0^{(M)})})$.
\item If $\mathcal{S}(\widehat{t}^{(1)},\dots,\widehat{t}^{(M)})=\mathcal{S}(G^{(1)}/\widehat{V}_{\mathcal{G}^{(1)}(\alpha_0^{(1)})},\dots,G^{(M)}/\widehat{V}_{\mathcal{G}^{(M)}(\alpha_0^{(M)})})$, then we have $\widehat{t}^{(m)}\leq e_{g}^{(m)}$ for $g\in \mathcal{G}^{(m)}(\alpha_0^{(m)})$ and $\widehat{\text{FDP}}^{(m)}(\widehat{t}^{(1)},\dots,\widehat{t}^{(M)})\leq \alpha^{(m)}$ for all $m$.
Note that for all $m\in[M]$, $\mathcal{S}_{\text{init}}^{(m)}=\mathcal{S}^{(m)}(G^{(1)}/\widehat{V}_{\mathcal{G}^{(1)}(\alpha_0^{(1)})},\dots,G^{(M)}/\widehat{V}_{\mathcal{G}^{(M)}(\alpha_0^{(M)})})$.
Therefore, for any $m\in[M]$, we have
\begin{align*}
\frac{\widehat{V}_{\mathcal{G}^{(m)}(\alpha_0^{(m)})}}{|\mathcal{S}^{(m)}_{\text{init}}|}
&=\frac{G^{(m)}/(G^{(m)}/\widehat{V}_{\mathcal{G}^{(m)}(\alpha_0^{(m)})})}{|\mathcal{S}^{(m)}(G^{(1)}/\widehat{V}_{\mathcal{G}^{(1)}(\alpha_0^{(1)})},\dots,G^{(M)}/\widehat{V}_{\mathcal{G}^{(M)}(\alpha_0^{(M)})})|}\\
&\leq \frac{G^{(m)}/\widehat{t}^{(m)}}{|\mathcal{S}^{(m)}(G^{(1)}/\widehat{V}_{\mathcal{G}^{(1)}(\alpha_0^{(1)})},\dots,G^{(M)}/\widehat{V}_{\mathcal{G}^{(M)}(\alpha_0^{(M)})})|}\\
&=\widehat{\text{FDP}}(\widehat{t}^{(1)},\dots,\widehat{t}^{(M)})\\
&\leq\alpha^{(m)}.
\end{align*}
This implies that the event $\bigcap_{m=1}^M\left\{\widehat{V}^{(m)}_{\mathcal{G}^{(m)}(\alpha_0^{(m)})}\leq\alpha^{(m)}|\mathcal{S}_{\text{init}}^{(m)}|\right\}$ occurs.
\end{itemize}
Now we prove the second sufficient condition. If $|\mathcal{G}^{(m)}(\alpha_0^{(m)})|/|\mathcal{S}^{(m)}_{\text{init}}|\leq \alpha^{(m)}/\alpha_0^{(m)}$, then we have
\begin{align*}
\frac{\widehat{V}_{\mathcal{G}^{(m)}(\alpha_0^{(m)})}}{|\mathcal{S}^{(m)}_{\text{init}}|}
&=\frac{\widehat{V}_{\mathcal{G}^{(m)}(\alpha_0^{(m)})}}{\mathcal{G}^{(m)}(\alpha_0^{(m)})}\frac{\mathcal{G}^{(m)}(\alpha_0^{(m)})}{|\mathcal{S}^{(m)}_{\text{init}}|}\\
& \overset{(i)}{\leq} \alpha_{0}^{(m)}\frac{\mathcal{G}^{(m)}(\alpha_0^{(m)})}{|\mathcal{S}^{(m)}_{\text{init}}|}\\
&\leq \alpha_0^{(m)},
\end{align*}
where (i) holds by the definition of $\mathcal{K}_{\text{finite}}\cup \mathcal{K}_{\text{asy}}$. Thus, the event $\bigcap_{m=1}^M\left\{\widehat{V}^{(m)}_{\mathcal{G}^{(m)}(\alpha_0^{(m)})}\leq\alpha^{(m)}|\mathcal{S}_{\text{init}}^{(m)}|\right\}$ occurs.

\subsection{Proof of Theorem 4}
Consider any fixed $m\in[M]$. We only prove the case that $\mathcal{G}^{(m)} \in \mathcal{K}_{\text{finite}}$.
The proof for the case where $\mathcal{G}^{(m)} \in \mathcal{K}_{\text{asy}}$ can be finished similarly.
For $\mathcal{G}^{(m)} \in \mathcal{K}_{\text{finite}}$, we have $\mathbb{E}[V_{\mathcal{G}^{(m)}_{r}(\alpha_{0}^{(m)})} \slash\widehat{V}_{\mathcal{G}^{(m)}_{r}(\alpha^{(m)}_{0})}] \leq 1$. This further implies that
\begin{equation*}
	\begin{split}
		\sum\limits_{g \in \mathcal{H}_{0}^{(m)}}\mathbb{E}[e_{gr}^{(m)}]
		&=\mathbb{E}\left[G^{(m)} \cdot \frac{\sum\limits_{g \in \mathcal{H}_{0}^{(m)}} \mathbb{I}\left \{g \in \mathcal{G}^{(m)}_{r}(\alpha_{0}^{(m)})\right\}}{\widehat{V}_{\mathcal{G}^{(m)}_{r}(\alpha_{0}^{(m)})} \bigvee \alpha_{0}^{(m)}}\right]\\
		&=G^{(m)}\cdot \mathbb{E}\left[\frac{V_{\mathcal{G}^{(m)}_{r}(\alpha_{0}^{(m)})}}{\widehat{V}_{\mathcal{G}^{(m)}_{r}(\alpha_{0}^{(m)})}  \bigvee \alpha_{0}^{(m)}}\right]\\
		&\leq G^{(m)}.
	\end{split}
\end{equation*} 
By averaging over $r=1,\dots,R$, we have
\begin{equation*}
	\begin{split}
		\sum\limits_{g\in\mathcal{H}_{0}^{(m)}}\mathbb{E}[\overline{e}^{(m)}_{g}]
		&=\sum\limits_{g\in\mathcal{H}_{0}^{(m)}}\mathbb{E}\left[\sum\limits_{r\in[R]} \omega^{(m)}_{r}e_{gr}^{(m)}\right]\\
		&=\sum\limits_{r\in[R]}\omega^{(m)}_{r}\sum\limits_{g\in\mathcal{H}_{0}^{(m)}}\mathbb{E}[e_{gr}^{(m)}] \\
		&\leq \sum\limits_{r\in[R]}\omega^{(m)}_{r}G^{(m)}\\
		&= G^{(m)}.
	\end{split}
\end{equation*} 
Based on the proof of Lemma 1, it follows that $\text{FDR}^{(m)} \leq \alpha^{(m)}$.


\subsection{Proof of Theorem 5}
For each layer $m\in[M]$ and group $g\in[G^{(m)}]$, conditional on the data $(\boldsymbol{X}, \boldsymbol{y})$, the $e_{gr}^{(m)}$'s are are independent and identically distributed for different replications $r\in[R^{(m)}]$. By Hoeffding's inequality, we have 
\begin{align*}
	\begin{split}
\mathbb{P}\left(\max\limits_{g\in[G^{(m)}]}\left|\overline{e}_{g}^{(m)}-\Bar{\Bar{e}}_{g}^{(m)}\right| < \Delta^{(m)}\big{|} \boldsymbol{X}, \boldsymbol{y}\right)
		&\geq 1- \sum\limits_{g=1}^{G^{(m)}}\mathbb{P}\left(\left|\overline{e}_{g}^{(m)}-\Bar{\Bar{e}}_{g}^{(m)}\right| > \Delta^{(m)} \big{| }\boldsymbol{X}, \boldsymbol{y}\right)\\
		&\geq 1 - 2G^{(m)}\exp\left(\frac{-2(\Delta^{(m)})^{2}R^{(m)}}{(G^{(m)})^{2}}\right). 
	\end{split}
\end{align*} 
Thus, we have 
\begin{align*}
	\begin{split}
		\mathbb{P}\left(\bigcap\limits_{m=1}^{M}\left\{\max_{k\in[G^{(m)}]}\left|\overline{e}_{g}^{(m)} - \Bar{\Bar{e}}_{g}^{(m)}\right| < \Delta^{(m)}\right\} \Big{|} \boldsymbol{X}, \boldsymbol{y} \right)
    &=1-\mathbb{P}\left(\bigcup_{m=1}^M\left\{\max_{k\in[G^{(m)}]}\left|\overline{e}_{g}^{(m)} - \Bar{\Bar{e}}_{g}^{(m)}\right| \geq \Delta^{(m)}\right\} \big{|} \boldsymbol{X}, \boldsymbol{y} \right)\\
		&\geq 1-\sum_{m=1}^M \mathbb{P}\left(\max_{k\in[G^{(m)}]}\left|\overline{e}_{g}^{(m)} - \Bar{\Bar{e}}_{g}^{(m)}\right| \geq \Delta^{(m)} \big{|} \boldsymbol{X}, \boldsymbol{y} \right)\\
		&\geq 1-2\sum_{m=1}^M G^{(m)}\exp\left(\frac{-2(\Delta^{(m)})^2R^{(m)}}{(G^{(m)})^2}\right) .
	\end{split}
\end{align*}

Recall that for $j\in \mathcal{S}_{\infty}$, we have $\Bar{\Bar{e}}_{h(m,j)}^{(m)} \geq \widehat{t}^{(m)}_{\infty}$ for all $m\in[M]$. Thus, on the event that $\bigcap\limits_{m=1}^{M}\left\{\max_{g\in[G^{(m)}]}\left|\overline{e}_{g}^{(m)} - \Bar{\Bar{e}}_{g}^{(m)}\right| < \Delta^{(m)}\right\}$, for $j\in \mathcal{S}_{\infty}$, we have
\begin{align*}
	\begin{split}
		\overline{e}_{h(m,j)}^{(m)} 
		&=  \Bar{\Bar{e}}_{h(m,j)}^{(m)} - (\Bar{\Bar{e}}_{h(m,j)}^{(m)}   - \overline{e}_{h(m,j)}^{(m)})\\
		&>  \Bar{\Bar{e}}_{h(m,j)}^{(m)}-\Delta^{(m)}\\
		&\geq\Bar{\Bar{e}}_{h(m,j)}^{(m)}-\left|\Bar{\Bar{e}}_{h(m,j)}^{(m)}-\widehat{t}^{(m)}_{\infty}\right|\\
		&= \widehat{t}_{\infty}^{(m)},
	\end{split}
\end{align*}
and for $j\notin \mathcal{S}_{\infty}$, we have
\begin{align*}
	\begin{split}
		\overline{e}_{h(m,j)}^{(m)} 
		&=  \Bar{\Bar{e}}_{h(m,j)}^{(m)} - (\Bar{\Bar{e}}_{h(m,j)}^{(m)}   - \overline{e}_{h(m,j)}^{(m)})\\
		&<  \Bar{\Bar{e}}_{h(m,j)}^{(m)}+\Delta^{(m)}\\
		&\leq\Bar{\Bar{e}}_{h(m,j)}^{(m)}+\left|\Bar{\Bar{e}}_{h(m,j)}^{(m)}-\widehat{t}^{(m)}_{\infty}\right|\\
		&= \widehat{t}_{\infty}^{(m)}.
	\end{split}
\end{align*}
Thus, the threshold vector $(\widehat{t}^{(1)},\dots,\widehat{t}^{(M)}) = (\widehat{t}^{(1)}_{\infty},\dots,\widehat{t}^{(M)}_{\infty})$ belongs to the set $\mathcal{T}(\alpha^{(1)},\dots,\alpha^{(M)})$ and we have  $\mathcal{S}_{\infty} \subseteq \mathcal{S}_{\text{stab}}$.
Adopting the same argument as above, we also have $\mathcal{S}_{\infty} \supseteq \mathcal{S}_{\text{stab}}$. Hence, the event $\bigcap\limits_{m=1}^{M}\left\{\max_{g\in[G^{(m)}]}\left|\overline{e}_{g}^{(m)} - \Bar{\Bar{e}}_{g}^{(m)}\right| < \Delta^{(m)}\right\}$ implies that $\mathcal{S}_{\text{stab}} =\mathcal{S}_{\infty}$. Therefore, we complete the proof that 
\[
\mathbb{P}\left(\mathcal{S}_{\text{stab}} =\mathcal{S}_{\infty} | \boldsymbol{X}, \boldsymbol{y} \right) \geq 1-2\sum_{m=1}^M G^{(m)}\exp\left(\frac{-2(\Delta^{(m)})^2R^{(m)}}{(G^{(m)})^2}\right) .
\]

\section{Proofs for Section 4}
\subsection{Proof of Theorem 6}
We introduce the following notations to present several key ratios. For $t \in \mathbb{R} $, denote
\[
\widehat{L}^{0}_{G^{(m)}}(t)=\frac{1}{G_{0}^{(m)}}\sum\limits_{g \in \mathcal{H}_{0}^{(m)}}\mathbb{I}\left\{T^{(m)}_{g}>t \right\}, \quad	L^{0}_{G^{(m)}}(t)=\frac{1}{G_{0}^{(m)}}\sum\limits_{g \in \mathcal{H}_{0}^{(m)}}\mathbb{P}\left\{T^{(m)}_{g}>t\right\},	
\]
\[
\widehat{L}^{1}_{G^{(m)}}(t)=\frac{1}{G_{1}^{(m)}}\sum\limits_{g \in \mathcal{H}_{1}^{(m)}}\mathbb{I}\left\{T_{g}^{(m)}>t \right\}, \quad 	\widehat{V}^{0}_{G^{(m)}}(t)=\frac{1}{G^{(m)}_{0}}\sum\limits_{g \in \mathcal{H}_{0}^{(m)}}\mathbb{I}\left\{T_{g}^{(m)}<-t\right\},
\]
where $G_{0}^{(m)} = |\mathcal{H}_{0}^{(m)}|$ and $G_{1}^{(m)}=G^{(m)}-G_{0}^{(m)}$.
Let $r_{G^{(m)}}$ represent $G^{(m)}_{1} \slash G^{(m)}_{0}$. We define several forms of the false discovery proportion as
\[
\text{FDP}_{G^{(m)}}(t)=\frac{\widehat{L}_{G^{(m)}}^{0}(t)}{\widehat{L}_{G^{(m)}}^{0}(t)+r_{G^{(m)}}\widehat{L}_{G^{(m)}}^{1}(t)},
\]
\[	\text{FDP}_{G^{(m)}}^{+}(t)=\frac{\widehat{V}_{G^{(m)}}^{0}(t)}{\widehat{L}_{G^{(m)}}^{0}(t)+r_{G^{(m)}}\widehat{L}_{G^{(m)}}^{1}(t)},
\]
\[ \overline{\text{FDP}}_{G^{(m)}}(t)=\frac{L_{G^{(m)}}^{0}(t)}{L_{G^{(m)}}^{0}(t)+r_{G^{(m)}}\widehat{L}_{G^{(m)}}^{1}(t)}.
\]
\begin{lemma} \label{lemma1}
	Under Assumption~3.3, if $G_{0}^{(m)} \to \infty$ as $G^{(m)} \to \infty$, we have in probability,
	\[
	\sup\limits_{t\in \mathbb{R}}| \widehat{L}_{G^{(m)}}^{0}(t) - L_{G^{(m)}}^{0}| \to 0, \quad 	\sup\limits_{t\in \mathbb{R}}| \widehat{V}_{G^{(m)}}^{0}(t) - L_{G^{(m)}}^{0}(t)| \to 0.
	\]
\end{lemma}
\begin{proof}[Proof of Lemma~\ref{lemma1}]
	For any $\epsilon \in (0, 1)$, denote $-\infty = \alpha_{0}^{G^{(m)}} < \alpha_{1}^{G^{(m)}} < \dots < \alpha_{N_{\epsilon}}^{G^{(m)}} = + \infty$ with $N_{\epsilon} = [2 \slash\epsilon]$, such that $L_{G^{(m)}}^{0}(\alpha^{G^{(m)}}_{k-1}) - L^{0}_{G^{(m)}}(\alpha^{G^{(m)}}_{k}) \leq \epsilon/2$ for $k \in [N_{\epsilon}]$. By Assumption 2, such a sequence $\{\alpha_{k}^{G^{(m)}}\}$ exists since $L_{G^{(m)}}^{0}(t)$ is a continuous function for $t \in \mathbb{R}$. We have 
	\begin{align*}
		\begin{split}
			\mathbb{P}\left(\sup\limits_{t\in \mathbb{R}}(\widehat{L}^{0}_{G^{(m)}}(t) - L_{G^{(m)}}^{0}(t)) > \epsilon \right)
			&\leq \mathbb{P}\left(\bigcup\limits_{k=1}^{N_{\epsilon}} \sup\limits_{t\in [\alpha_{k-1}^{G^{(m)}}, \alpha_{k}^{G^{(m)}})} (\widehat{L}^{0}_{G^{(m)}}(t) - L^{0}_{G^{(m)}}(t)) \geq \epsilon \right) \\
			&\leq \sum\limits_{k=1}^{N_{\epsilon}}\mathbb{P}\left(\sup\limits_{t\in [\alpha_{k-1}^{G^{(m)}}, \alpha_{k}^{G^{(m)}})} (\widehat{L}^{0}_{G^{(m)}}(t) - L^{0}_{G^{(m)}}(t)) \geq \epsilon \right).
		\end{split}
	\end{align*}
	Thus, for any fixed $k\in [N_{\epsilon}]$, we have 
	\begin{equation*}
		\begin{split}
			&\quad~\sup\limits_{t\in[\alpha_{k-1}^{G^{(m)}},\alpha_{k}^{G^{(m)}})}(\widehat{L}_{G^{(m)}}^{0}(t) - L^{0}_{G^{(m)}}(t)) \leq \widehat{L}_{G^{(m)}}^{0}(\alpha_{k-1}^{G^{(m)}}) - L^{0}_{G^{(m)}}(\alpha_{k}^{G^{(m)}}) \\
			&\quad\quad\quad\quad \quad\quad \quad\quad \quad\quad \quad\quad \quad\quad \quad\quad~ \leq \widehat{L}_{G^{(m)}}^{0}(\alpha_{k-1}^{G^{(m)}}) - L^{0}_{G^{(m)}}(\alpha_{k-1}^{G^{(m)}})+ \epsilon/2.
		\end{split}
	\end{equation*}
	It further leads that 
	\begin{equation*}
		\begin{split}
			&\quad ~\mathbb{P}\left(\sup\limits_{t\in \mathbb{R}}(\widehat{L}^{0}_{G^{(m)}}(t) - L_{G^{(m)}}^{0}(t)) > \epsilon \right) \\
			&\leq \sum\limits_{k=1}^{N_{\epsilon}}\mathbb{P}\left(\sup\limits_{t\in [\alpha_{k-1}^{G^{(m)}}, \alpha_{k}^{G^{(m)}})} (\widehat{L}^{0}_{G^{(m)}}(t) - L^{0}_{G^{(m)}}(t)) \geq \epsilon \right)\\
			&\leq \sum\limits_{k=1}^{N_{\epsilon}}\mathbb{P}\left(\sup\limits_{t\in [\alpha_{k-1}^{G^{(m)}}, \alpha_{k}^{G^{(m)}})} (\widehat{L}^{0}_{G^{(m)}}(\alpha_{k-1}^{G^{(m)}}) - L^{0}_{G^{(m)}}(\alpha_{k-1}^{G^{(m)}})) \geq \epsilon/2 \right)\\
			&\leq \sum\limits_{k=1}^{N_{\epsilon}}\frac{4}{\epsilon^{2}(G^{(m)}_{0})^{2}}\text{Var}\left(\sum\limits_{g\in \mathcal{H}_{0}^{(m)}} \mathbb{I}\{T_{g}^{(m)} > \alpha_{k-1}^{G^{(m)}}\}\right)\\
			&\leq \frac{4c^{(m)}N_{\epsilon}}{\epsilon^{2}(G^{(m)}_{0})^{2-\lambda^{(m)}}} \to 0, \quad \text{as}~G^{(m)}\to \infty.
		\end{split}
	\end{equation*}
	The last inequality is guaranteed by Assumption 2.
	Similarly, we have
		\begin{equation*}
		\begin{split}
			&\quad~\mathbb{P}\left(\inf\limits_{t\in \mathbb{R}}(\widehat{L}^{0}_{G^{(m)}}(t) - L_{G^{(m)}}^{0}(t)) < -\epsilon \right) \leq  \sum\limits_{k=1}^{N_{\epsilon}}\mathbb{P}\left(\widehat{L}^{0}_{G^{(m)}}(\alpha_{k}^{G^{(m)}}) - L^{0}_{G^{(m)}}(\alpha_{k}^{G^{(m)}}) \leq -\epsilon/2 \right)\\
			&\quad \quad \quad \quad \quad \quad  \quad \quad \quad \quad \quad \quad \quad \quad \quad \quad~\leq \frac{4c^{(m)}N_{\epsilon}}{\epsilon^{2}(G^{(m)}_{0})^{2-\lambda^{(m)}}} \to 0,\quad \text{as } G^{(m)} \to \infty.
		\end{split}
	\end{equation*}
	Then the following inequality holds that
	\[
	\mathbb{P}\left( \sup\limits_{t\in \mathbb{R}}\left| \widehat{L}_{G^{(m)}}^{0}(t) - L_{G^{(m)}}^{0}(t)\right| > \epsilon \right) \to 0, \quad \text{as}~G^{(m)} \to \infty.
	\]
	Under Assumption 1, we have $\widehat{L}_{G^{(m)}}^{0}(t) \overset{d}{=} \widehat{V}_{G^{(m)}}^{0}(t)$, which means that $\sup\limits_{t\in \mathbb{R}}|\widehat{V}^{0}_{G^{(m)}}(t) - L^{0}_{G^{(m)}}(t)| \to 0$ in probability.
\end{proof}
\begin{proof}[Proof of Theorem 6]
	With the assumption that $T_g^{(m)}$'s are continuous variables and $\text{Var}(T_{g}^{(m)})$ is uniformly upper bounded and lower bounded away from 0, by Lemma~\ref{lemma1} and continuous mapping theorem, we have
	\[
	\sup\limits_{0<t \leq c}| \text{FDP}_{G^{(m)}}^{+}(t) - \text{FDP}_{G^{(m)}}(t)| \overset{p}{\to} 0
	\]
	for any constant $c > 0$. For any fixed $\epsilon \in (0, \alpha^{(m)})$ and enough large $G^{(m)}$, we have
	\begin{equation*}
		\begin{split}
			\mathbb{P}(t_{\alpha^{(m)}} \leq \tau_{\alpha^{(m)}-\epsilon})
			&\geq \mathbb{P}(\text{FDP}_{G^{(m)}}^{+}(\tau_{\alpha^{(m)}-\epsilon}) \leq \alpha^{(m)})\\
			&\geq \mathbb{P}(| \text{FDP}_{G^{(m)}}^{+}(\tau_{\alpha^{(m)}-\epsilon}) - \text{FDP}_{G^{(m)}}(\tau_{\alpha^{(m)}-\epsilon})| \leq \epsilon, \text{FDP}_{G^{(m)}}(\tau_{\alpha^{(m)}-\epsilon}) \leq \alpha^{(m)}-\epsilon)\\
			&\geq 1 - \epsilon.
		\end{split}
	\end{equation*}
	Conditioning on the event $t_{\alpha^{(m)}} \leq \tau_{\alpha^{(m)}-\epsilon}$, we have
	\begin{equation*}
		\begin{split}
			&\quad~\limsup\limits_{G^{(m)} \to \infty}\mathbb{E}[\text{FDP}_{G^{(m)}}(t_{\alpha^{(m)}})]\\
			&\leq \limsup\limits_{G^{(m)} \to \infty}\mathbb{E}[\text{FDP}_{G^{(m)}}(t_{\alpha^{(m)}}) | t_{\alpha^{(m)}} \leq \tau_{\alpha^{(m)}-\epsilon}] \mathbb{P}(t_{\alpha^{(m)}}\leq \tau_{\alpha^{(m)}-\epsilon})+\epsilon\\
			&\leq \limsup\limits_{G^{(m)} \to \infty}\mathbb{E}[|\text{FDP}_{G^{(m)}}(t_{\alpha^{(m)}}) - \overline{\text{FDP}}_{G^{(m)}}(t_{\alpha^{(m)}}) || t_{\alpha^{(m)}} \leq \tau_{\alpha^{(m)}-\epsilon}] \mathbb{P}(t_{\alpha^{(m)}}\leq \tau_{\alpha^{(m)}-\epsilon})\\
			&~~~~+ \limsup\limits_{G^{(m)} \to \infty}\mathbb{E}[|\text{FDP}^{+}_{G^{(m)}}(t_{\alpha^{(m)}}) - \overline{\text{FDP}}_{G^{(m)}}(t_{\alpha^{(m)}}) || t_{\alpha^{(m)}} \leq \tau_{\alpha^{(m)}-\epsilon}] \mathbb{P}(t_{\alpha^{(m)}}\leq \tau_{\alpha^{(m)}-\epsilon})\\
			&~~~~+\limsup\limits_{G^{(m)} \to \infty}\mathbb{E}[\text{FDP}^{+}_{G^{(m)}}(t_{\alpha^{(m)}}) | t_{\alpha^{(m)}} \leq \tau_{\alpha^{(m)}-\epsilon}] \mathbb{P}(t_{\alpha^{(m)}}\leq \tau_{\alpha^{(m)}-\epsilon})\\
			&\leq \limsup\limits_{G^{(m)} \to \infty} \mathbb{E} \left[ \sup\limits_{0<t\leq \tau_{\alpha^{(m)}-\epsilon}}|\text{FDP}_{G^{(m)}}(t) - \overline{\text{FDP}}_{G^{(m)}}(t)|\right]\\
			&~~~~+\limsup\limits_{G^{(m)} \to \infty} \mathbb{E} \left[ \sup\limits_{0<t\leq \tau_{\alpha^{(m)}-\epsilon}}|\text{FDP}^{+}_{G^{(m)}}(t) - \overline{\text{FDP}}_{G^{(m)}}(t)|\right]\\
			&~~~~+\limsup\limits_{G^{(m)} \to \infty} \mathbb{E}[\text{FDP}_{G^{(m)}}^{+}(t_{\alpha^{(m)}})] + \epsilon \\
			&\leq \alpha^{(m)}+\epsilon.
		\end{split}
	\end{equation*}
	 The last inequality is based on the definition of $t_{\alpha^{(m)}}$ ($\text{FDP}_{G^{(m)}}^{+}(t_{\alpha^{(m)}}) \leq \alpha^{(m)}$) and the result of dominated convergence theorem and Lemma~\ref{lemma1} that 
	\begin{align*}
	\limsup\limits_{G^{(m)} \to \infty} \mathbb{E} \left[ \sup\limits_{0<t\leq \tau_{\alpha^{(m)}-\epsilon}}|\text{FDP}_{G}(t) - \overline{\text{FDP}}_{G}(t)|\right] &= \limsup\limits_{G^{(m)} \to \infty} \mathbb{E} \left[ \sup\limits_{0<t\leq \tau_{\alpha^{(m)}-\epsilon}}|\text{FDP}^{+}_{G}(t) - \overline{\text{FDP}}_{G}(t)|\right] \\&= 0.
	\end{align*}
	According to the arbitrariness of $\epsilon$, we have 
	\[
	\limsup\limits_{G^{(m)}\to \infty}\text{FDR}(t_{\alpha^{(m)}})\leq \alpha^{(m)}.
	\]
\end{proof}

\subsection{Proof of Theorem 7}
We first prove that for any $m\in[M]$, the set $\{e_{gr}^{(m)}\}_{g\in[G^{(m)}]}$ is a set of asymptotic relaxed e-values. 
For any $t \geq 0$, let $\widehat{V}_{G}^{1}(t) =  \sum_{g \in \mathcal{H}_{1}^{(m)}} \mathbb{I}\left\{T_{g}^{(m)} \leq -t \right\}/G^{(m)}_{1}$. We have
\begin{align*}
&~~~~~\frac{1}{G^{(m)}}\sum_{g\in \mathcal{H}_{0}^{^{(m)}}}\mathbb{E}\left[e^{(m)}_{gr}\right] \\
&=\mathbb{E}\left[\frac{\sum_{g\in \mathcal{H}_{0}^{^{(m)}}}\mathbb{I}\left\{g \in \mathcal{G}^{(m)}_{r}(\alpha_{0}^{(m)})\right\}}{\widehat{V}_{\mathcal{G}^{(m)}_{r}(\alpha_{0}^{(m)})} \bigvee \alpha_{0}^{(m)}}\right]    \\
&\leq \mathbb{E}\left[ \frac{\widehat{L}_{G^{(m)}}^{0}(t_{\alpha_{0}^{(m)}}^{r})}{\widehat{V}_{G^{(m)}}^{0}(t_{\alpha_{0}^{(m)}}^{r})+r_{G^{(m)}}\widehat{V}_{G^{(m)}}^{1}(t_{\alpha_{0}^{(m)}}^{r}) } \right]\\
&\leq  \mathbb{E}\left[\left| \frac{\widehat{L}_{G^{(m)}}^{0}(t_{\alpha_{0}^{(m)}}^{r})}{\widehat{V}_{G^{(m)}}^{0}(t_{\alpha_{0}^{(m)}}^{r})+r_{G^{(m)}}\widehat{V}_{G^{(m)}}^{1}(t_{\alpha_{0}^{(m)}}^{r})} - \frac{L_{G^{(m)}}^{0}(t_{\alpha_{0}^{(m)}}^{r})}{L^{0}_{G^{(m)}}(t_{\alpha_{0}^{(m)}}^{r})+r_{G^{(m)}}\widehat{V}_{G^{(m)}}^{1}(t_{\alpha^{(m)}_{0}}^{r}) } \right|  \right]\\
&~~~~+\mathbb{E}\left[\left| \frac{L_{G^{(m)}}^{0}(t_{\alpha^{(m)}_{0}}^{r})}{L_{G^{(m)}}^{0}(t_{\alpha^{(m)}_{0}}^{r})+r_{G^{(m)}}\widehat{V}_{G^{(m)}}^{1}(t_{\alpha^{(m)}_{0}}^{r})} - \frac{\widehat{V}_{G^{(m)}}^{0}(t_{\alpha^{(m)}_{0}}^{r})}{\widehat{V}^{0}_{G^{(m)}}(t_{\alpha^{(m)}_{0}}^{r})+r_{G^{(m)}}\widehat{V}_{G^{(m)}}^{1}(t_{\alpha^{(m)}_{0}}^{r})} \right|   \right]  \\
&~~~~+\mathbb{E}\left[ \frac{\widehat{V}_{G^{(m)}}^{0}(t_{\alpha^{(m)}_{0}}^{r})}{\left(\widehat{V}^{0}_{G^{(m)}}(t_{\alpha^{(m)}_{0}}^{r})+r_{G^{(m)}}\widehat{V}_{G^{(m)}}^{1}(t_{\alpha^{(m)}_{0}}^{r})\right)} \right]\\
&\leq 1+\mathbb{E}\left[\left| \frac{\widehat{L}_{G^{(m)}}^{0}(t_{\alpha_{0}^{(m)}}^{r})}{\widehat{V}_{G^{(m)}}^{0}(t_{\alpha_{0}^{(m)}}^{r})+r_{G^{(m)}}\widehat{V}_{G^{(m)}}^{1}(t_{\alpha_{0}^{(m)}}^{r})} - \frac{L_{G^{(m)}}^{0}(t_{\alpha_{0}^{(m)}}^{r})}{L^{0}_{G^{(m)}}(t_{\alpha_{0}^{(m)}}^{r})+r_{G^{(m)}}\widehat{V}_{G^{(m)}}^{1}(t_{\alpha^{(m)}_{0}}^{r}) } \right|  \right]\\
&~~~~+\mathbb{E}\left[\left| \frac{L_{G^{(m)}}^{0}(t_{\alpha^{(m)}_{0}}^{r})}{L_{G^{(m)}}^{0}(t_{\alpha^{(m)}_{0}}^{r})+r_{G^{(m)}}\widehat{V}_{G^{(m)}}^{1}(t_{\alpha^{(m)}_{0}}^{r})} - \frac{\widehat{V}_{G^{(m)}}^{0}(t_{\alpha^{(m)}_{0}}^{r})}{\widehat{V}^{0}_{G^{(m)}}(t_{\alpha^{(m)}_{0}}^{r})+r_{G^{(m)}}\widehat{V}_{G^{(m)}}^{1}(t_{\alpha^{(m)}_{0}}^{r})} \right|   \right].
\end{align*}
For any fixed $\epsilon \in (0, \alpha_{0}^{(m)})$, denote the event $\{ t^{r}_{\alpha^{(m)}_{0}} \leq \tau^{r}_{\alpha^{(m)}_{0}-\epsilon}\}$ by $\mathcal{C}$.
Since $T_g^{(m)}$'s are continuous variables and $\text{Var}(T_{g}^{(m)})$ is uniformly upper bounded and lower bounded away from 0, according to the proof of Theorem 6, for enough large $G^{(m)}$, we have $\mathbb{P}(\mathcal{C}) \geq 1- \epsilon$. Thus, we have $\mathbb{P}(\mathcal{C}^{c}) \leq \epsilon$.
Therefore, 
\begin{align*}
&~~~~~\mathbb{E}\left[\left| \frac{L_{G^{(m)}}^{0}(t_{\alpha^{(m)}_{0}}^{r})}{L_{G^{(m)}}^{0}(t_{\alpha^{(m)}_{0}}^{r})+r_{G^{(m)}}\widehat{V}_{G^{(m)}}^{1}(t_{\alpha^{(m)}_{0}}^{r})} - \frac{\widehat{V}_{G^{(m)}}^{0}(t_{\alpha^{(m)}_{0}}^{r})}{\widehat{V}^{0}_{G^{(m)}}(t_{\alpha^{(m)}_{0}}^{r})+r_{G^{(m)}}\widehat{V}_{G^{(m)}}^{1}(t_{\alpha^{(m)}_{0}}^{r})} \right|   \right] \\
&\leq \mathbb{E}\left[\left| \frac{L_{G^{(m)}}^{0}(t_{\alpha^{(m)}_{0}}^{r})}{L_{G^{(m)}}^{0}(t_{\alpha^{(m)}_{0}}^{r})+r_{G^{(m)}}\widehat{V}_{G^{(m)}}^{1}(t_{\alpha^{(m)}_{0}}^{r})} - \frac{\widehat{V}_{G^{(m)}}^{0}(t_{\alpha^{(m)}_{0}}^{r})}{\widehat{V}^{0}_{G^{(m)}}(t_{\alpha^{(m)}_{0}}^{r})+r_{G^{(m)}}\widehat{V}_{G^{(m)}}^{1}(t_{\alpha^{(m)}_{0}}^{r})} \right| \Bigg{|}\mathcal{C}\right]\mathbb{P}(\mathcal{C})+2\mathbb{P}(\mathcal{C}^{c})\\
&\leq \mathbb{E}\left[\sup\limits_{0<t<\tau_{\alpha_0^{(m)}-\epsilon}}\left|\frac{L_{G^{(m)}}^{0}(t_{\alpha^{(m)}_{0}}^{r})}{L_{G^{(m)}}^{0}(t_{\alpha^{(m)}_{0}}^{r})+r_{G^{(m)}}\widehat{V}_{G^{(m)}}^{1}(t_{\alpha^{(m)}_{0}}^{r})} - \frac{\widehat{V}_{G^{(m)}}^{0}(t_{\alpha^{(m)}_{0}}^{r})}{\widehat{V}^{0}_{G^{(m)}}(t_{\alpha^{(m)}_{0}}^{r})+r_{G^{(m)}}\widehat{V}_{G^{(m)}}^{1}(t_{\alpha^{(m)}_{0}}^{r})} \right|\right]+2\epsilon.
\end{align*}
According to Lemma~\ref{lemma1} and the dominated convergence theorem, we have
\[
\limsup_{G^{(m)} \to \infty} \mathbb{E}\left[\sup_{0 < t \leq \tau_{\alpha_{0}^{(m)}-\epsilon}} \left| \frac{L_{G^{(m)}}^{0}(t)}{L_{G^{(m)}}^{0}(t)+r_{G^{(m)}}\widehat{V}_{G^{(m)}}^{1}(t)} - \frac{\widehat{V}_{G^{(m)}}^{0}(t)}{\left(\widehat{V}^{0}_{G^{(m)}}(t)+r_{G^{(m)}}\widehat{V}_{G^{(m)}}^{1}(t)\right)} \right|   \right] =0.
\]
By the assumption that $\sum_{g=1}^{G^{(m)}}\mathbb{E}[e_{gr}^{(m)}|t^{r}_{\alpha_0^{(m)}} >\tau_{\alpha_0^{(m)}}]/G^{(m)}$ has a uniform upper bound, we have
\begin{align*}
\mathbb{E}\left[\left| \frac{\widehat{L}_{G^{(m)}}^{0}(t_{\alpha_{0}^{(m)}}^{r})}{\widehat{V}_{G^{(m)}}^{0}(t_{\alpha_{0}^{(m)}}^{r})+r_{G^{(m)}}\widehat{V}_{G^{(m)}}^{1}(t_{\alpha_{0}^{(m)}}^{r})} - \frac{L_{G^{(m)}}^{0}(t_{\alpha_{0}^{(m)}}^{r})}{L^{0}_{G^{(m)}}(t_{\alpha_{0}^{(m)}}^{r})+r_{G^{(m)}}\widehat{V}_{G^{(m)}}^{1}(t_{\alpha^{(m)}_{0}}^{r}) }\right| \cdot 1\{\mathcal{C}\} \right]\leq c\epsilon,
\end{align*}
where $c>0$ is an absolute constant.
Similarly, we have
\begin{align*}
&~~~~~\mathbb{E}\left[\left| \frac{\widehat{L}_{G^{(m)}}^{0}(t_{\alpha_{0}^{(m)}}^{r})}{\widehat{V}_{G^{(m)}}^{0}(t_{\alpha_{0}^{(m)}}^{r})+r_{G^{(m)}}\widehat{V}_{G^{(m)}}^{1}(t_{\alpha_{0}^{(m)}}^{r})} - \frac{L_{G^{(m)}}^{0}(t_{\alpha_{0}^{(m)}}^{r})}{L^{0}_{G^{(m)}}(t_{\alpha_{0}^{(m)}}^{r})+r_{G^{(m)}}\widehat{V}_{G^{(m)}}^{1}(t_{\alpha^{(m)}_{0}}^{r}) }\right| \right] \\
&\leq \mathbb{E}\left[\sup\limits_{0<t\leq \tau_{\alpha_0^{(m)}-\epsilon}}\left| \frac{L_{G^{(m)}}^{0}(t_{\alpha^{(m)}_{0}}^{r})}{L_{G^{(m)}}^{0}(t_{\alpha^{(m)}_{0}}^{r})+r_{G^{(m)}}\widehat{V}_{G^{(m)}}^{1}(t_{\alpha^{(m)}_{0}}^{r})} - \frac{\widehat{V}_{G^{(m)}}^{0}(t_{\alpha^{(m)}_{0}}^{r})}{\widehat{V}^{0}_{G^{(m)}}(t_{\alpha^{(m)}_{0}}^{r})+r_{G^{(m)}}\widehat{V}_{G^{(m)}}^{1}(t_{\alpha^{(m)}_{0}}^{r})} \right| \right]+c\epsilon.
\end{align*}
Since $\sum_{g\in\mathcal{H}_0^{(m)}}e_g^{(m)}/G^{(m)}$ is uniformly integrable, with Lemma~\ref{lemma1}, we have
\[
\limsup_{G^{(m)} \to \infty} \mathbb{E}\left[ \sup_{0 < t \leq \tau_{\alpha_{0}^{(m)}-\epsilon}}\left| \frac{\widehat{L}_{G^{(m)}}^{0}(t)}{\widehat{V}_{G^{(m)}}^{0}(t)+r_{G^{(m)}}\widehat{V}_{G^{(m)}}^{1}(t)} - \frac{L_{G^{(m)}}^{0}(t)}{L^{0}_{G^{(m)}}(t)+r_{G^{(m)}}\widehat{V}_{G^{(m)}}^{1}(t)} \right|  \right]  = 0.
\]
Putting these pieces together, we have
\begin{align*}
		\limsup_{G^{(m)} \to \infty} \frac{1}{G^{(m)}}\sum\limits_{g\in \mathcal{H}_{0}^{^{(m)}}}\mathbb{E}\left[e^{(m)}_{gr}\right] \leq(2+c)\epsilon.
\end{align*}
According to the arbitrariness of $\epsilon$, we have
\[
\limsup\limits_{G^{(m)} \to \infty}\frac{1}{G^{(m)}}\sum\limits_{g\in \mathcal{H}_{0}^{^{(m)}}}\mathbb{E}\left[e^{(m)}_{gr}\right] \leq 1.
\]
Thus, the set $\{e_{gr}^{(m)}\}_{g\in[G^{(m)}]}$ is a set of asymptotic relaxed e-values. We then have
\begin{align}
	\begin{split}
		\limsup\limits_{G^{(m)} \to \infty}\frac{1}{G^{(m)}}\sum\limits_{g\in \mathcal{H}_{0}^{(m)}}\mathbb{E}\left[\overline{e}^{(m)}_{g}\right] 
		&=\limsup\limits_{G^{(m)} \to \infty}\frac{1}{G^{(m)}}\sum\limits_{g\in \mathcal{H}_{0}^{(m)}}\mathbb{E}\left[\frac{1}{R}\sum\limits_{r=1}^{R}e_{gr}^{(m)}\right]\\
		&=\frac{1}{R}\sum\limits_{r=1}^{R}\left(\limsup_{G^{(m)} \to \infty}\frac{1}{G^{(m)}}\sum\limits_{g \in \mathcal{H} _{0}^{(m)}}\mathbb{E}[e^{(m)}_{gr}]\right)\\
		& \leq 1.		
	\end{split}
\end{align}
Since the group size $|\mathcal{A}_{g}^{(m)}|$ is uniformly bounded, we have $G^{(m)} \to \infty$ as $N\to \infty$.
According to Lemma 1, we have $\limsup_{N\to \infty} \text{FDR}^{(m)} \leq \alpha^{(m)}$.

\section{Discussion on stabilization for multiple resolutions} 
\label{other_derandomized}

\subsection{A detailed comparison with derandomization for single resolution}
E-values averaging has been applied to derandomized knockoff filter~\cite{ren2022derandomized} to achieve FDR control under a single resolution setting. However, this stabilization typically entails a loss of power, as demonstrated through simulations and discussions in~\cite{lee2024boosting,ren2022derandomized}. Intuitively, this reduction arises because certain important features with “fuzzy” or borderline signals are inconsistently selected across realizations. When a feature is not selected, its corresponding e-value is assigned zero, and the subsequent averaging procedure yields substantially a attenuated e-value. Consequently, these features fail to pass the e-BH selection threshold, resulting in a reduced detection power.

For multiple resolutions, the impact of stabilization is more complex. Consider the proposed SFEFP. On one hand, looking at each resolution separately, the generalized e-values of ``fuzzy'' important (groups) features can be small. This may result in a reduced power. On the other hand, the final averaged generalized e-values are no longer one-bit, yielding a better ranking information. This naturally avoid the one-bit result of FEFP (Theorem 3) and can bring power enhancement. Therefore, the impact of stabilization on the detection power can be quite complex. Anyway, the empirical studies in our main text demonstrate that the stabilization always brings a enhanced power, making stabilization a cost free choice.

Stabilization for multiple resolutions also differs from the single resolution setting in terms of parameter selection and stability guaranty.
Specifically, for single resolution, there is no need to do derandomization if the original detection procedure is stable, and \cite{ren2022derandomized} proved that $\alpha_0^{(1)}=\alpha^{(1)}$ is the optimal choice to recover detection power.
In contrast, for multiple resolutions, we still need to choose $\alpha^{(m)}_0<\alpha^{(m)}$ if $\mathcal{G}^{(m)}$ is stable, as implied by Theorem 3 and discussed in Subsection 3.3. Furthermore, to obtain generalized e-values with better ranking information, SFEFP further takes the fusion decision, i.e., running different determined procedure $\mathcal{G}^{(m)}$ and averaging their one-bit generalized e-values.
Finally, in terms of stability, Theorem 5 differs with the previous result since the stability of each resolution needs to be handled.

\subsection{Alternative method and comparison}
Ren and Barber~\cite{ren2022derandomized} achieve derandomization by running knockoff procedure multiple times and averaging e-values obtained from the output of each running. 
A nature question arises:
\begin{center}
\textit{Is it possible to establish a connection between the output of multi-layer filters (e.g., MKF~\cite{katsevich2019multilayer} or the proposed output) and generalized e-values, and follows \cite{ren2022derandomized} to do derandomization?}
\end{center}
In this subsection, we will show that such a treatment is feasible but not competitive.

Denote the detection procedure with multilayer FDR control at the $r$-th replication by $\mathcal{G}_{r}$. Here, $\mathcal{G}_{r}$ can be the MKF or the FEFP, and in the latter case we write $\mathcal{G}_{r}=(\mathcal{G}_{r}^{(1)}, \dots, \mathcal{G}_{r}^{(M)})$. 
Let $\mathcal{G}_{r}(\gamma^{(1)},\dots,\gamma^{(M)})$ denote the rejection set of $\mathcal{G}_{r}$ with the original target FDR levels $(\gamma^{(1)},\dots,\gamma^{(M)})$. 
The set of rejected groups at layer $m$ of $\mathcal{G}_{r}(\gamma^{(1)},\dots,\gamma^{(M)})$ is denoted as $\mathcal{G}_{r}^{[m]}(\gamma^{(m)})$, which needs to be distinguished from the rejection of running $\mathcal{G}^{(m)}_{r}$ separately.
Let $\widehat{V}_{\mathcal{G}_{r}^{[m]}(\gamma^{(m)})}$ represent the estimated number of false discoveries at layer $m$ of the multilayer filtering procedure. Then we can construct generalized e-values by
\begin{equation}\label{mkf-evalue}
e^{(m)}_{gr}=G^{(m)} \cdot \frac{\mathbb{I}\left \{g \in \mathcal{G}^{[m]}_{r}(\gamma^{(m)}) \right \}} {\widehat{V}_{\mathcal{G}^{[m]}_{r}(\gamma^{(m)})}\bigvee \gamma^{(m)}}.
\end{equation}
Then for each we take a average that $\overline{e}^{(m)}_{g}=\sum_{r=1}^{R}e_{gr}^{(m)}/R$ for $g\in[G^{(m)}],~m\in[M]$.
Finally, the set of selected features $\widehat{\mathcal{S}}_{\text{derand}}$ is obtained by applying the generalized e-filter at the target levels $(\alpha^{(1)}, \dots, \alpha^{(M)})$.

The following theorem establish a connection between the output of multilayer filtering procures and generalized e-values. That is, the $r$-th selected set of the MKF or the FEFP is equivalent to applying the generalized e-filter to $\{e_{gr}^{(m)}\}_{m\in[M]}$ with identical FDR levels. 

\begin{theorem}\label{theorem4}
Consider $\mathcal{G}_{r}$ as the MKF or the FEFP.
Let $\mathcal{S}_{r}(\gamma^{(1)},\dots,\gamma^{(M)})$ be the set of selected features for the generalized e-filter applied to $\{e_{gr}^{(m)}\}_{m\in[M]}$~\eqref{mkf-evalue} with the same preset level $(\gamma^{(1)},\dots,\gamma^{(M)})$. Then $\mathcal{G}_{r}(\gamma^{(1)},\dots,\gamma^{(M)}) = \mathcal{S}_{r}(\gamma^{(1)},\dots,\gamma^{(M)})$.
\end{theorem}

\begin{proof}
Denote the final threshold of the generalized e-filter $\mathcal{S}_{r}(\gamma^{(1)},\dots,\gamma^{(M)})$ by $(\widehat{T}^{(1)}, \dots, \widehat{T}^{(M)})$. On one hand, for $j \notin \mathcal{G}_{r}(\gamma^{(1)},\dots,\gamma^{(M)})$, its generalized e-value equals to zero by definition. This leads to $\mathcal{G}_{r} \supseteq  \mathcal{S}_{r}$. 
	On the other hand, we have
	\[
	\mathcal{S}_{\text{gefilter}}(T^{(1)},\dots,T^{(M)}) \subseteq \mathcal{S}_{r}(\gamma^{(1)},\dots,\gamma^{(M)}),
	\]
	for any $(T^{(1)},\dots,T^{(M)}) \in \mathcal{T}(\gamma^{(1)},\dots,\gamma^{(M)})$, where 
	\[
	\mathcal{S}_{\text{gefilter}}(T^{(1)},\dots,T^{(M)}) = \{j: \text{for all}~m \in [M], e_{h(m,j)r}^{(m)} \geq T^{(m)} \}
	\]
	with $e_{gr}^{(m)}$ given by (\ref{mkf-evalue}).
	Thus, it suffices to find thresholds $(T^{(1)}, \dots, T^{(M)}) \in \mathcal{T}(\gamma^{(1)},\dots,\gamma^{(M)})$ satisfying $\mathcal{S}_{\text{gefilter}}(T^{(1)},\dots,T^{(M)}) \supseteq \mathcal{G}_{r}(\gamma^{(1)},\dots,\gamma^{(M)})$.
    We prove the existence of such $(T^{(1)},\dots,T^{(M)})$ as follows.
	Let 
	\[
	T^{(m)} = \frac{G^{(m)}}{\widehat{V}_{\mathcal{G}_{r}^{[m]}(\gamma^{(m)})} \bigvee \gamma^{(m)}}.
	\]
	For $j \in \mathcal{G}_{r}$, we have $e_{g(m, j)}^{(m)}=T^{(m)}$, $m\in[M]$.
	If $(T^{(1)},\dots,T^{(M)}) \in \mathcal{T}(\gamma^{(1)},\dots,\gamma^{(M)})$, then $X_{j}$ is selected by the generalized e-filter, which completes the proof. Therefore, it suffices to prove $\widehat{\text{FDP}}^{(m)} (T^{(1)},\ldots,T^{(M)}) \leq \gamma^{(m)}$ for all $m\in [M]$. This fact holds since
	\begin{align*}
		\begin{split}
			\widehat{\text{FDP}}^{(m)} (T^{(1)},\ldots,T^{(M)}) 
			& = \frac{G^{(m)}} { T^{(m)} | {\mathcal{S}^{(m)}_{\text{gefilter}}}(T^{(1)},\ldots,T^{(M)}) | }  \\
			& = \frac{ \widehat{V}_{\mathcal{G}_{r}^{[m]}(\gamma^{(m)})} \bigvee \gamma^{(m)}}{|\mathcal{S}^{(m)}_{\text{gefilter}}(T^{(1)},\dots,T^{(M)})|} \\
			&= \frac{ \widehat{V}_{\mathcal{G}_{r}^{[m]}(\gamma^{(m)})}\bigvee \gamma^{(m)}}{{\mathcal{G}_{r}^{[m]}(\gamma^{(m)})}} \cdot \frac{{\mathcal{G}_{r}^{[m]}(\gamma^{(m)})}}{|\mathcal{S}^{(m)}_{\text{gefilter}}(T^{(1)},\dots,T^{(M)})|} \\
			& \le \gamma^{(m)} \frac{{\mathcal{G}_{r}^{[m]}(\gamma^{(m)})}}{|\mathcal{S}^{(m)}_{\text{gefilter}}(T^{(1)},\dots,T^{(M)})|}\\
			&\leq \gamma^{(m)},
		\end{split}
	\end{align*}
	where the last inequality is guaranteed by $\mathcal{S}_{\text{gefilter}}(T^{(1)},\dots,T^{(M)})  \supseteq \mathcal{G}_{r}(\gamma^{(1)},\dots,\gamma^{(M)})$. 
	This concludes the proof.
\end{proof}

The following theorem demonstrates that this derandomized procedure also controls FDR at multiple resolutions.
\begin{theorem}\label{theorem5}
Consider $\mathcal{G}_{r}$ as the MKF($c$)+ or the FEFP. 
Consider any original preset FDR levels $(\gamma^{(1)},\dots,\gamma^{(M)}) \in (0,1)^{M}$, target FDR levels $(\alpha^{(1)}, \dots,\alpha^{(M)}) \in (0,1)^{M}$, and number of realizations $R \geq1$. For MKF($c$)+, the $\widehat{\mathcal{S}}_{\text{derand}}$ satisfies $\text{FDR}^{(m)} \leq c_{\text{kn}}\alpha^{(m)}/c$ for $m \in [M]$, where $c_{\text{kn}}= 1.96$.
When $\mathcal{G}_{r}$ refers to FEFP, if $\mathcal{G}^{(m)}_{r} \in \mathcal{K}_{\text{finite}}$, then $\widehat{\mathcal{S}}_{\text{derand}}$ satisfies $\text{FDR}^{(m)} \leq \alpha^{(m)}$ for finite samples. If $\mathcal{G}^{(m)}_{r} \in \mathcal{K}_{\text{asy}}$ and the group size $|\mathcal{A}_{g}^{(m)}|$ is uniformly bounded, then $\text{FDR}^{(m)} \leq \alpha^{(m)}$ when $(n, N)\to \infty$ at a proper rate.
\end{theorem}

\begin{proof}
	Denote the final threshold vector of $\mathcal{G}_{r}$ as $(\widetilde{t}^{(1)},\dots,\widetilde{t}^{(M)})$.
	We first prove for the case where $\mathcal{G}_{r}$ refers to $\text{MKF}(c_{\text{kn}})+$.
	For each layer $m$, denote the group knockoff statistics as $\{T_{gr}^{(m)}\}_{g\in [G^{(m)}]}$. 
	According to~\cite{katsevich2019multilayer}, we have 
	\[
	\widehat{V}_{\mathcal{G}_{r}^{[m]}} = c \cdot \left(1 + \left| \left\{g: T_{gr}^{(m)} \leq -\widetilde{t}^{(m)} \right\} \right| \right),
	\]
	where $c >0$ is a constant. For each replication $r\in[R]$, we have
	\begin{align*}
		\begin{split}
			\sum\limits_{g \in \mathcal{H}_{0}^{(m)}}\mathbb{E}[e_{gr}^{(m)}] 
			&=G^{(m)}\cdot \mathbb{E}\left[\frac{V_{\mathcal{G}^{[m]}_{r}(\gamma^{(m)})}}{\widehat{V}_{\mathcal{G}^{[m]}_{r}(\gamma^{(m)})} \bigvee \gamma^{(m)}}\right]\\
			&\leq G^{(m)}\cdot \mathbb{E}\left[\frac{ \left| \left\{g: T_{gr}^{(m)} \geq \widetilde{t}^{(m)} \right\} \bigcap \mathcal{H}_{0}^{(m)} \right|}{c \cdot \left(1 + \left| \left\{g: T_{gr}^{(m)} \leq -\widetilde{t}^{(m)} \right\} \right| \right)}\right] \\
			&\leq G^{(m)}\cdot \frac{1}{c} \cdot \mathbb{E}\left[\sup\limits_{t^{(m)}} \frac{ \left| \left\{g: T_{gr}^{(m)} \geq t^{(m)} \right\} \bigcap \mathcal{H}_{0}^{(m)} \right|}{ 1 + \left| \left\{g: T_{gr}^{(m)} \leq -t^{(m)} \right\} \bigcap \mathcal{H}^{(m)}_{0} \right| }\right]\\
			&\leq G^{(m)} \cdot \frac{c_{\text{kn}}}{c},
		\end{split}
	\end{align*} 
	where $c_{\text{kn}} = 1.96$. The proof of the last inequality see \cite{katsevich2019multilayer}. Then, we have 
	\[
	\sum_{g\in\mathcal{H}_{0}^{(m)}}\mathbb{E}[\overline{e}_{g}^{(m)}] \leq c_{\text{kn}}G^{(m)}/c.
	\]
	According to the proof of Lemma 1, we have $\text{FDR}^{(m)} \leq c_{\text{kn}}\alpha^{(m)}/c$.

	Now we prove for the case for FEFP, i.e, $\mathcal{G}_r=(\mathcal{G}_r^{(1)},\dots,\mathcal{G}_r^{(M)})$. We only prove for the case that $\mathcal{G}^{(m)}_{r} \in \mathcal{K}_{\text{finite}}$. 
	When $\mathcal{G}_{r}(\gamma^{(1)},\dots,\gamma^{(M)}) = \emptyset$, we have $e_{gr}^{(m)} = 0$.
	When $\mathcal{G}_{r}(\gamma^{(1)},\dots,\gamma^{(M)}) \neq \emptyset$, we have $\widetilde{t}^{(m)} \in \left\{G^{(m)}\slash (\alpha^{(m)}k): k\in[G^{(m)}]\right\} $.
\begin{align*}
\sum\limits_{g\in\mathcal{H}_{0}^{(m)}}\mathbb{E}[e_{gr}^{(m)}]
&= \sum\limits_{g\in\mathcal{H}_{0}^{(m)}}\mathbb{E}\left[\frac{G^{(m)}1\{g\in \mathcal{G}_r^{[m]}(\gamma^{(1)},\dots,\gamma^{(M)})\}}{G^{(m)}/\widetilde{t}^{(m)}}\right]\\
&= \sum\limits_{g\in\mathcal{H}_{0}^{(m)}}\mathbb{E}[\widetilde{t}^{(m)}1\{g\in \mathcal{G}_r^{[m]}(\gamma^{(1)},\dots,\gamma^{(M)})\}]\\
&\leq \sum\limits_{g\in\mathcal{H}_{0}^{(m)}}\mathbb{E}[G^{(m)}1\{g\in \mathcal{G}_r^{[m]}(\gamma^{(1)},\dots,\gamma^{(M)})\}  /\widehat{V}_{\mathcal{G}^{(m)}_r}]\\
&\leq \sum\limits_{g\in\mathcal{H}_{0}^{(m)}}\mathbb{E}[G^{(m)}1\{g\in \mathcal{G}_r^{(m)}\}  /\widehat{V}_{\mathcal{G}^{(m)}_r}]\\
&= G^{(m)} \mathbb{E}[V_{\mathcal{G}^{(m)}}/\widehat{V}_{\mathcal{G}^{(m)}}]\\
&\leq G^{(m)}.
\end{align*}
The remaining proof follows identically to that of Theorem~4.
\end{proof}

Compared with SFEFP, this alternative method is less competitive in several aspects. When $\mathcal{G}_r$ refers to MKF($c$)+, this method still suffers a reduced power due to the constant $c_{\text{kn}}$. When $\mathcal{G}_r$ refers to FEFP, this alternative requires adjusting too many parameters, including $(\gamma^{(1)},\dots,\gamma^{(M)})$ and the parameter within FEFP $(\alpha_0^{(1)},\dots,\alpha_0^{(M)})$. Furthermore, a major advantage of SFEFP is avoiding the zero-power-dilemma of FEFP by obtaining non one-bit generalized e-values. In contrast, since the alternative method is building on the result of FEFP, $e_{gr}^{(m)}$ can be zero for each $g\in[G^{(m)}]$ for certain $r\in[R]$. This leads to a small $\overline{e}_g^{(m)}$ and further results in a reduced power.

\section{Additional Simulation Results}
In this section, wo provide additional simulation results to further demonstrate the claims proposed in the manuscript.
We compare the performance of the compound e-values with ours.

\begin{figure}[htbp]
	\centering
    \includegraphics[width=45em]{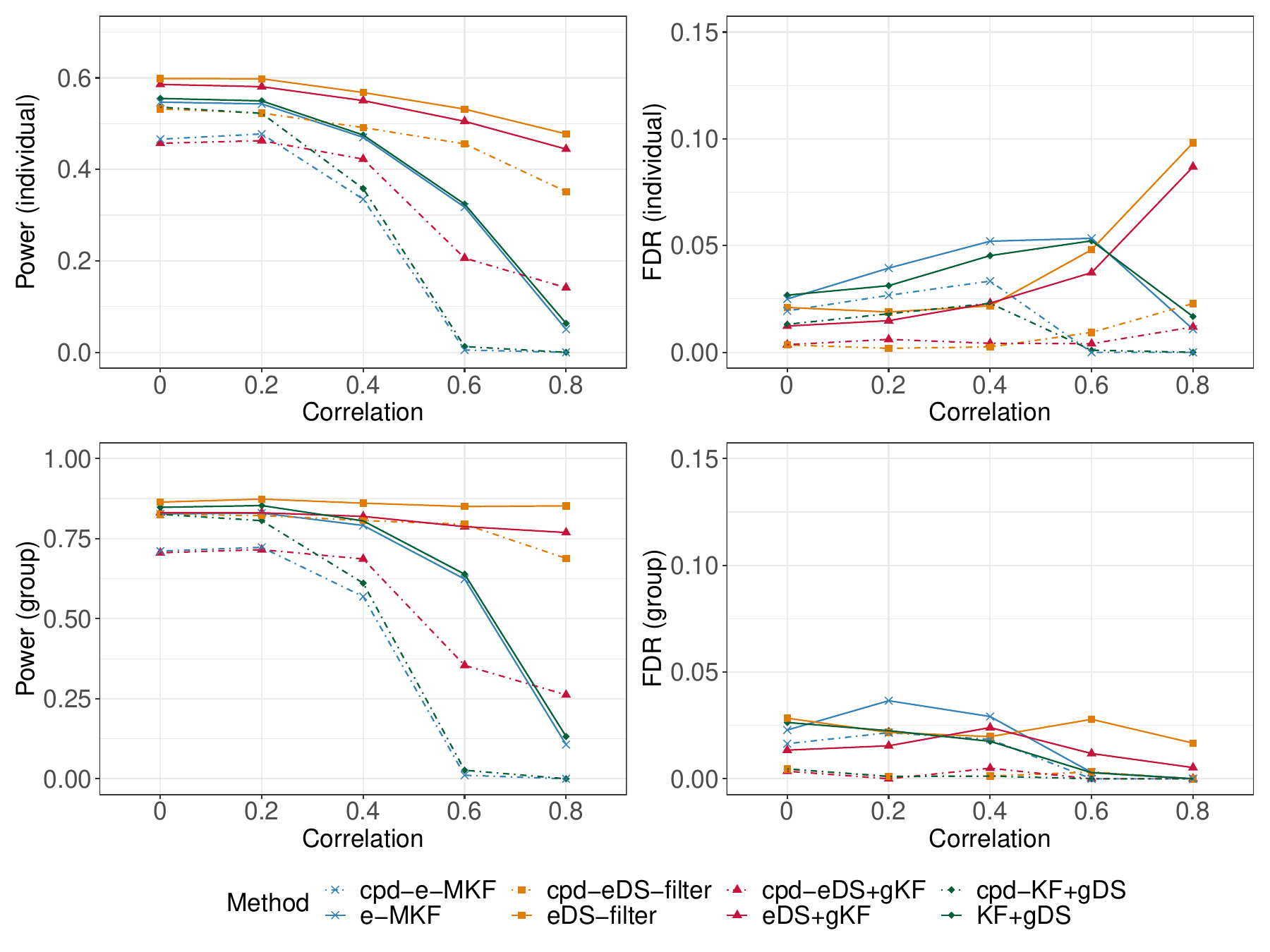}
	\caption{Comparison with compound e-values under different correlations with $\delta=3$.}
	\label{compound_signal3}
\end{figure}

\begin{figure}[htbp]
	\centering
    \includegraphics[width=45em]{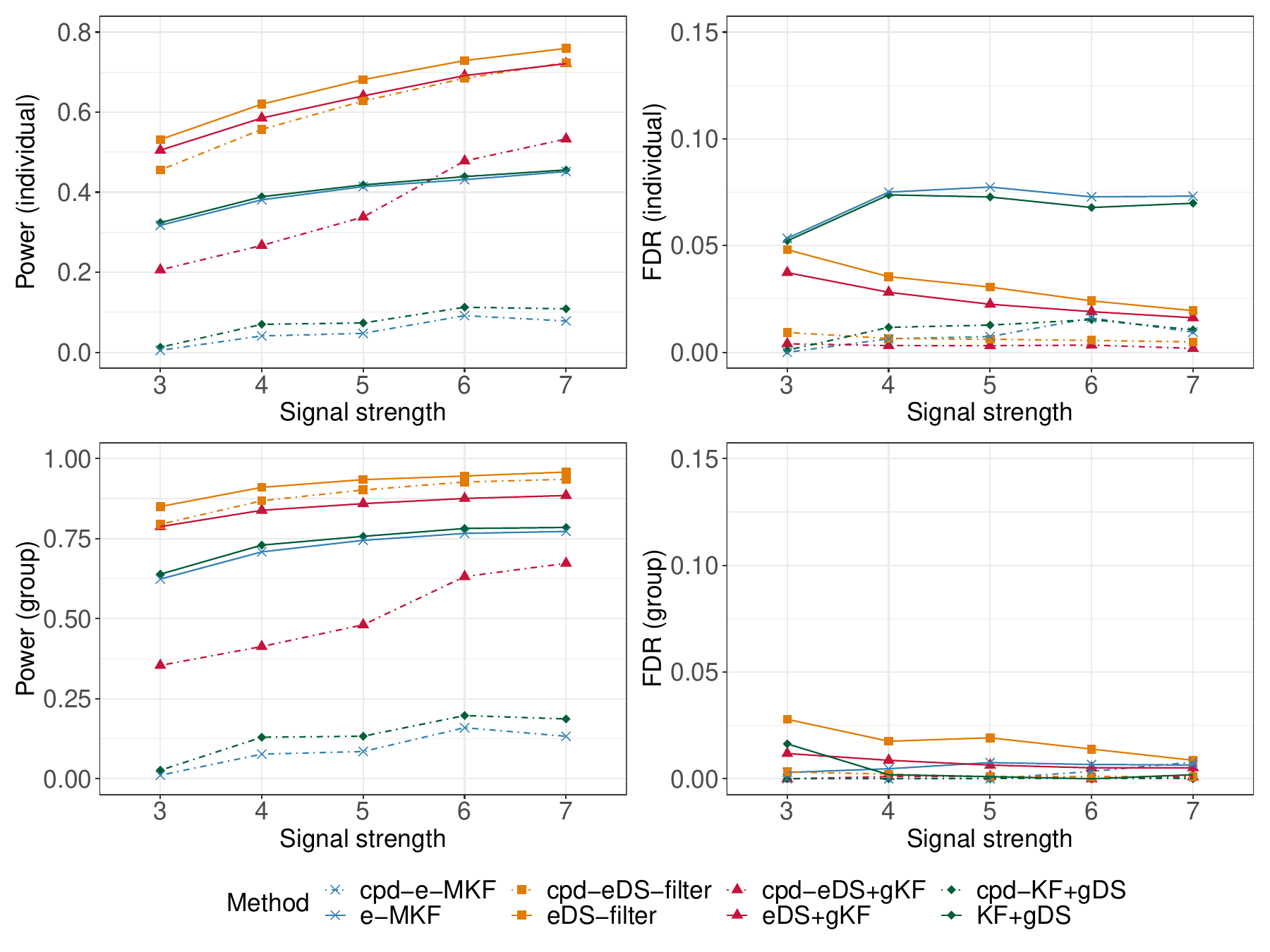}
	\caption{Comparison with compound e-values under different correlations with $\rho=0.6$.}
	\label{compound_rho6}
\end{figure}
Consider replacing the proposed generalized e-values to compound e-values, where
\begin{equation*}\label{compound_evalue2}
	e_{g}^{(m)'}=  G^{(m)} \cdot \frac{\mathbb{I}\left \{g \in \mathcal{G}^{(m)}(\alpha_{0}^{(m)})\right \} } {\alpha_{0}^{(m)}\left(\left|\mathcal{G}^{(m)}(\alpha_{0}^{(m)})\right|\bigvee 1\right)}.
\end{equation*}
With compound e-values, we can develop corresponding versions for e-MKF, eDS-filter, eDS+gKF, and KF+gDS. We denote these methods by cpd-e-MKF, cpd-eDS-filter, cpd-eDS+gKF, and cpd-KF+gDS, respectively.
The simulation settings are consistent with the setting in the main text for high-dimensional settings.
That is, consider $n=600$, $N=800$, $G=80$, $|\mathcal{H}_1|=60$, and $K=20$. We set $\alpha^{(m)}=0.2$ and $R^{(m)}=50$ for $m=1,2$. All results are averaged over 50 independent trials.

Figure~\ref{compound_signal3} and Figure~\ref{compound_rho6} present the performance of methods under different correlations and signal strength, respectively. 
We observe that our proposed constructions consistently achieve significantly higher power than compound e-values at both resolutions. The power gaps increase when features become high correlated. 

\section{Unified e-filter}
We develop an analog of \cite{Ramdas2019A} for generalized e-values, called unified e-filter, for incorporating prior knowledge into e-filter to broaden the application scope and improve detection power.
Before presenting the method, we introduce necessary notations as follows, to clarify the additional knowledge considered by unified e-filter.
\begin{itemize}
\item \textit{Overlapping groups}. The unified e-filter allows overlapping partitions, that is, it allows $\mathcal{A}_{g}^{(m)}\cap \mathcal{A}_{h}^{(m)}\ne \emptyset$ for $g\ne h$, $m\in[M]$. Let $g^{(m)}(i)=\{g\in[G^{(m)}]:i\in A_{g}^{(m)}\}$ denote the index set of groups in the $m$-th resolution to which $H_i$ belongs.
\item \textit{Incomplete partitions}. 
The unified e-filter allows incomplete partitions, say, $\cup_g\mathcal{A}_g^{(m)}\subseteq[N]$. 
Let $L^{(m)}=[N] \setminus \bigcup_{g} A_{g}^{(m)}$ denote the leftover of $m$th partition, i.e., the index set of elementary hypotheses that do not belong to any group in the $m$th partition.
\item \textit{Penalties and priors.} 
Knowledge about the scientific importance of hypothesis $H_g^{(m)}$ and the non-null beliefs are considered.
Large penalties $u_{g}^{(m)}$ reflects that hypothesis $H_g^{(m)}$ is more scientifically important, and a large prior weight $v_g^{(m)}$ indicates that the hypothesis $H_g^{(m)}$ is more likely to be non-null. Penalties $\{u_g^{(m)}\}$ and priors $\{v_g^{(m)}\}$ are normalized such that $\sum_{g=1}^{G^{(m)}}u_g^{(m)}v_{g}^{(m)}=G^{(m)}$.
\item \textit{Null-proportion adaptivity.} Fix a user-defined constant $ \lambda^{(m)}\in(0,1)$, the weighted null proportion estimator is 
\begin{equation}\label{adaptive}
\widehat{\pi}^{(m)} :=
\frac{ \, \lvert u^{(m)} \cdot v^{(m)} \rvert_{\infty}
       + \sum_{g} u^{(m)}_{g} v^{(m)}_{g} \mathbf{1}\!\left\{e^{(m)}_{g} < 1/\lambda^{(m)} \right\} }
     { G^{(m)} \bigl( 1 - \lambda^{(m)} \bigr) } .
\end{equation}
If e-values are known to be independent at $m$th layer, we can leverage $\widehat{\pi}^{(m)}$ to enhance power.
\end{itemize}

Without causing ambiguity, let
\[
\frac{a}{ b} :=
\begin{cases}
0, & \text{if } a=0, \\
\dfrac{a}{b}, & \text{if } a\ne 0,\, b\ne 0, \\
\text{undefined}, & \text{if } a\ne 0,\, b=0.
\end{cases}
\]
The unified e-filter aims to (1) incorporate the above knowledge to provide a rejection set with internal consistency, and (2) satisfy the multilayer penalty-weighted FDR control:
\[
\text{FDR}_{u}^{(m)}=\mathbb{E}[\text{FDP}_u^{(m)}] \leq \alpha^{(m)}~\text{simultaneously for } m = 1,\dots, M,
\]
where
\[
\text{FDP}_{u}^{(m)}
=
  \frac{
    \sum_{g \in \mathcal{H}_{0}^{(m)}} u^{(m)}_{g} \,\mathbf{1}\!\left\{ g \in \mathcal{S}^{(m)} \right\}
  }{
    \sum_{g \in [G^{(m)}]} u^{(m)}_{g} \,\mathbf{1}\!\left\{ g \in \mathcal{S}^{(m)} \right\}.
  }
\]

\subsection{Unified e-filter framework}
To ensure the internal consistency of discoveries, $H_i$ is rejected if and only if for each $m\in[M]$, either there is at least one rejected group containing $i$, or $i\in L^{(m)}$.
Therefore, we define the candidate selection set of individuals as
\begin{align*}
\mathcal{S}(\vec{k})=\Big\{ i : \forall m,\; &\text{either } i \in L^{(m)}, 
 \text{or } \exists g \in g^{(m)}(i),\;\\&
e_g^{(m)} \ge \max \Big\{
\frac{1\{\widehat{k}^{(m)}\ne 0\}\widehat{\pi}^{(m)}G^{(m)}}{v_g^{(m)}\alpha^{(m)} \widehat{k}^{(m)}},\;
1\{\widehat{k}^{(m)}=0\}\cdot\infty,\;
\frac{1}{\lambda^{(m)}}
\Big\}
\Big\}
\end{align*}
for any $\vec{k}=(k^{(1)},\dots,k^{(M)})$ with $k^{(m)}\in [0,G^{(m)}]$ and $0/0=0$.
The index set of selected groups in layer $m$ with $\vec{k}$ is defined by
\begin{equation}\label{candidate set}
\begin{aligned}
\mathcal{S}^{(m)}(\vec{k})
= \Big\{ g \in [G^{(m)}] :\;
& A_g^{(m)} \cap \mathcal{S}(\vec{k}) \neq \emptyset~\text{and } \\
&  e_g^{(m)} \ge \max\Big\{
\frac{1\{\widehat{k}^{(m)}\ne0\}\widehat{\pi}^{(m)}G^{(m)}}{v_g^{(m)}\alpha^{(m)}\widehat{k}^{(m)}},1\{\widehat{k}^{(m)}=0\}\cdot \infty,
\frac{1}{\lambda^{(m)}} \Big\}
\Big\}.
\end{aligned}
\end{equation}
The set of admissible rejection number, $\mathcal{K}(\alpha^{(1)},\dots,\alpha^{(M)})$, is defined by 
\begin{equation}\label{kappa}
\mathcal{K}(\alpha^{(1)},\dots,\alpha^{(M)}) = \left\{ \vec{k} \in [0, G^{(1)}] \times \cdots \times [0,G^{(m)}] : \sum\limits_{g\in\mathcal{S}^{(m)}(\vec{k})}u_g^{(m)} \geq k^{(m)} \text{ for all } m \right\}.
\end{equation}
For each layer $m\in[M]$, the final rejection number $\widehat{k}^{(m)}$ is determined by 
\begin{equation}\label{threshold_unified}
\widehat{k}^{(m)} = \max \left\{ k^{(m)} : \left( k^{(1)}, \ldots, k^{(M)} \right) 
\in \mathcal{K}\big(\alpha^{(1)}, \ldots, \alpha^{(M)}\big) \right\}.
\end{equation}
The unified e-filter is summarized in Algorithm~\ref{unified e-filter}.

\begin{algorithm}[t]\label{unified e-filter}
	\caption{The unified e-filter}
	\KwIn {Partition $\{\mathcal{A}_{g}^{(m)}\}_{g\in[G^{(m)}]}$, $m \in [M]$; e-values $ \{e_{g}^{(m)}\}_{g\in[G^{(m)}]}$, $m\in [M]$; target FDR level vector $\{\alpha^{(m)}\}_{m\in[M]}$;
    penalty weights and prior weights $\{u_g^{(m)},v_g^{(m)}\}_{g\in [G^{(m)}]}$, $m\in[M]$; thresholds for adaptive null proportion estimation $\{\lambda^{(m)}\}_{m\in[M]}$.
    }
	\textbf{Initialize:} $k^{(m)}= G^{(m)}$ and $\widehat{\pi}^{(m)}$ by Equation \eqref{adaptive}, $m\in[M]$.
	
	\Repeat{\textrm{until} $k^{(1)}, \ldots, k^{(M)} $ \textrm{are all unchanged}}{
		\For{$m=1, \ldots, M$}{
			Compute candidate rejection set $\mathcal{S}^{(m)}(k^{(1)},\dots,k^{(M)})$ by Equation 
            \eqref{candidate set}.
			
			Update the threshold
			\[
			k^{(m)} \leftarrow \max \left\{k \in [0,G^{(m)}] : \sum_{g\in\mathcal{S}^{(m)}(k^{(1)},\dots,k^{(m-1)},k,k^{(m+1)},\dots,k^{(M)})} u_{g}^{(m)}\geq k \right\}.
			\]
	}}
	
	\KwOut{the final threshold $\widehat{k}^{(m)} = k^{(m)}$ for $m\in [M]$ and the rejected set $\mathcal{S}(\widehat{k}^{(1)}, \ldots, \widehat{k}^{(M)} )$. }
\end{algorithm}

\begin{proposition}\label{feasible threshold}
Let the vector $(\widehat{k}^{(1)},\dots,\widehat{k}^{(M)})$ defined by \eqref{threshold_unified} and the $\mathcal{K}(\alpha^{(1)},\dots,\alpha^{(M)})$ defined by \eqref{kappa}. We have $(\widehat{k}^{(1)},\dots,\widehat{k}^{(M)})\in \mathcal{K}(\alpha^{(1)},\dots,\alpha^{(M)})$.
\end{proposition}

The following theorem illustrates the multilayer FDR control of the unified e-filter.
\begin{theorem}\label{unifiede-fdr}
Let vector $(\widehat{k}^{(1)},\dots,\widehat{k}^{(M)})$ be defined by Equation~\eqref{threshold_unified}. For any partition $m \in[M]$, we have:
\begin{itemize}
\item[(a)] Suppose $\{e_g^{(m)}\}$ are independent from each other.
If $\{e_g^{(m)}\}$ are valid e-values, then employing adaptivity by defining $\widehat{\pi}^{(m)}$ as in Equation \eqref{adaptive} guarantees that $\mathrm{FDR}_u^{(m)}\leq \alpha^{(m)}$.
If $\{e_g^{(m)}\}$ are asymptotic e-values, then $\lim\sup_{n,G^{(m)}\to\infty}\mathrm{FDR}_u^{(m)}\leq \alpha^{(m)}$.
\item[(b)] If $\{e_g^{(m)}\}$ satisfy $\sum_{g\in\mathcal{H}_0^{(m)}}u_g^{(m)}v_g^{(m)}\mathbb{E}[e_g^{(m)}]\leq G^{(m)}$ (relaxed e-values with priors and penalties) and are arbitrarily dependent, without adaptivity we have that
\[
\mathrm{FDR}_u^{(m)}\leq \alpha^{(m)} \frac{\sum_{g \in \mathcal{H}_0^{(m)}} u_g^{(m)}v_g^{(m)}\mathbb{E}[e_g^{(m)}]}{G^{(m)}} \leq \alpha^{(m)}.
\]
If $\{e_g^{(m)}\}$ satisfy $\lim\sup_{n,N\to \infty}\sum_{g\in\mathcal{H}_0^{(m)}}u_g^{(m)}v_g^{(m)}\mathbb{E}[e_g^{(m)}]/G^{(m)}\leq 1$, then  
\[
\lim\sup_{n,G^{(m)}\to\infty}\mathrm{FDR}_u^{(m)}\leq \alpha^{(m)}. 
\]
\end{itemize}
\end{theorem}

\subsection{Proofs for theoretical properties}
\begin{proof}[Proof of Proposition~\ref{feasible threshold}]
For any fixed $m\in[M]$, according to the definition of $\widehat{k}^{(m)}$, there exists $k_{s}^{(1)},\ldots,k_{s}^{(m-1)},k_{s}^{(m+1)},\ldots,k_{s}^{(M)} \in  \left[0,G^{(m)} \right]$ such that 
\[
(k_{s}^{(1)},\ldots,k_{s}^{(m-1)},\widehat{k}^{(m)},k_{s}^{(m+1)},\ldots,k_{s}^{(M)}) \in \mathcal{K} (\alpha^{(1)}, \ldots, \alpha^{(M)}).
\]
We have 
\[
\mathcal{S}(k_{s}^{(1)},\ldots,k_{s}^{(m-1)},\widehat{k}^{(m)},k_{s}^{(m+1)},\ldots,k_{s}^{(M)}) \subseteq \mathcal{S}(\widehat{k}^{(1)},\ldots,\widehat{k}^{(m-1)},\widehat{k}^{(m)},\widehat{k}^{(m+1)},\ldots,\widehat{k}^{(M)}),
\]
because $\widehat{k}^{(m^{'})} \geq k_{s}^{(m^{'})} (m^{'} \ne m)$ and $ \mathcal{S}(k^{(1)},\ldots,,k^{(M)}) $ is a increasing with $(k^{(1)},\ldots,,k^{(M)})$. 
Since the set $\mathcal{S}^{(m)}$ is given by $\mathcal{S}^{(m)} = \left\{  g \in [G^{(m)}] : \mathcal{S} \cap A_{g}^{(m)} \ne \emptyset \right\}$,
we have
\[
\mathcal{S}^{(m)}(k_{s}^{(1)},\ldots,k_{s}^{(m)},\widehat{k}^{(m)},k_{s}^{(m+1)},\ldots,k_{s}^{(M)}) \subseteq \mathcal{S}^{(m)}(\widehat{k}^{(1)},\ldots,\widehat{k}^{(m-1)},\widehat{k}^{(m)},\widehat{k}^{(m+1)},\ldots,\widehat{k}^{(M)}).
\]
Therefore,		
\[
\sum\limits_{g\in\mathcal{S}(\widehat{k}^{(1)},\dots,\widehat{k}^{(m)})}u_{g}^{(m)}\geq\sum_{g\in\mathcal{S}(k^{(1)},\dots,k_s^{(m-1)},\widehat{k}^{(m)},k_s^{(m+1)},\dots,k^{(M)})}u_g^{(m)}\geq \widehat{k}^{(m)}.
\]
By the definition of $\mathcal{K} (\alpha^{(1)}, \ldots, \alpha^{(M)})$, we have 
\[
(\widehat{k}^{(1)},\ldots,\widehat{k}^{(M)}) \in \mathcal{K} (\alpha^{(1)}, \ldots, \alpha^{(M)}).
\]
\end{proof}

\begin{proof}[proof of Theorem \ref{unifiede-fdr}]
For both$(a)$ and $(b)$, we have
\begin{equation}\label{main_prof}
	\begin{split}
		&\quad~\mathbb{E}\left[\text{FDP}_u^{(m)}(\widehat{k}^{(1)},\dots,\widehat{k}^{(M)}) \right]\\
		& = \mathbb{E}\left[\frac{ \sum_{g \in \mathcal{H}_{0}^{(m)}} u^{(m)}_{g} \,\mathbf{1}\!\left\{ g \in \mathcal{S}^{(m)}(\widehat{t}^{(1)},\ldots,\widehat{t}^{(M)}) \right\} }{ \sum_{g \in [G^{(m)}]} u^{(m)}_{g} \,\mathbf{1}\!\left\{ g \in \mathcal{S}^{(m)}(\widehat{t}^{(1)},\ldots,\widehat{t}^{(M)}) \right\}}\right] \\
        &\leq \mathbb{E}\left[\frac{ \sum_{g \in \mathcal{H}_{0}^{(m)}} u^{(m)}_{g} \,\mathbf{1}\!\left\{ g \in \mathcal{S}^{(m)}(\widehat{t}^{(1)},\ldots,\widehat{t}^{(M)}) \right\}}{\widehat{k}^{(m)}}\right]\\
        &\leq \mathbb{E}\left[\frac{ \sum_{g \in \mathcal{H}_{0}^{(m)}} u^{(m)}_{g} \,\mathbf{1}\!\left\{ e_g^{(m)} \geq \max\left\{\frac{\widehat{\pi}^{(m)}G^{(m)}}{v_g^{(m)}\alpha^{(m)}\widehat{k}^{(m)}},\frac{1}{\lambda^{(m)}}\right\} \right\}}{\widehat{k}^{(m)}}\right]\\
        &=\mathbb{E}\left[\frac{ \sum_{g \in \mathcal{H}_{0}^{(m)}} u^{(m)}_{g} \,\mathbf{1}\!\left\{ e_g^{(m)} \geq \max\left\{\frac{\widehat{\pi}^{(m)}G^{(m)}}{v_g^{(m)}\alpha^{(m)}\widehat{k}^{(m)}},\frac{1}{\lambda^{(m)}}\right\} \right\}\alpha^{(m)}v_g^{(m)}/G^{(m)}}{\widehat{k}^{(m)}v_g^{(m)}\alpha^{(m)}/G^{(m)}}\right]\\
        &=\frac{\alpha^{(m)}}{G^{(m)}}\sum\limits_{g\in\mathcal{H}_0^{(m)}}u_g^{(m)}v_g^{(m)}\mathbb{E}\left[\frac{\,\mathbf{1}\!\left\{ e_g^{(m)} \geq \frac{1}{\lambda^{(m)}} \right\}\,\mathbf{1}\!\left\{ e_g^{(m)} \geq \frac{\widehat{\pi}^{(m)}G^{(m)}}{v_g^{(m)}\alpha^{(m)}\widehat{k}^{(m)}} \right\}}{\widehat{\pi}^{(m)}\frac{v_g^{(m)}\alpha^{(m)}\widehat{k}^{(m)}}{\widehat{\pi}^{(m)}G^{(m)}}}\right].
        \end{split}
\end{equation} 

\textit{Theorem statement $(a)$.} Before presenting the proof, we introduce the following lemma.
\begin{lemma}[Inverse binomial lemma~\cite{Ramdas2019A}]\label{inverse_lemma}
Given a vector $a\in[0,1]^d$, constant $b\in[0,1]$, and Bernoulli variables $Z_i \stackrel{\text{i.i.d.}}{\sim} \mathrm{Bernoulli}(b)$, the weighted sum $Z:=1+\sum_{i=1}^d a_iZ_i$ satisfies 
\begin{equation}\label{bound of inverse}
\frac{1}{1 + b \sum_{i=1}^d a_i}\leq\mathbb{E}\!\left[\frac{1}{Z}\right] \leq \frac{1}{\,b\!\left(1 + \sum_{i=1}^d a_i\right)}.
\end{equation}
\end{lemma}
Let $b:=\mathbb{P}(e_g^{(m)}\leq 1/\lambda^{(m)})$ and $d:=|\mathcal{H}_0^{(m)}|-1$. We have
\begin{equation*}
b=1-\mathbb{P}\left(e_g^{(m)}\geq \frac{1}{\lambda^{(m)}}\right)\geq1-\lambda^{(m)}\mathbb{E}[e_g^{(m)}]\geq1-\lambda^{(m)}.
\end{equation*}
Define
\begin{equation*}
Z := 1 + \sum_{\substack{h\in \mathcal{H}_0^{(m)},~h\ne g}}
a_h\,\mathbf{1}\!\left\{ e_h^{(m)} < \frac{1}{\lambda^{(m)}}\right\},
\qquad
\text{where }\;
a_h=\frac{u_h^{(m)}\, w_h^{(m)}}{\bigl\|\,u^{(m)}\!\cdot\! w^{(m)}\,\bigr\|_\infty}.
\end{equation*}
We have
\begin{equation*}
\widehat{\pi}_{-g}^{(m)}=\frac{1+\sum_{h\ne g}a_h I\left\{e_h^{(m)}<\frac{1}{\lambda^{(m)}}\right\}}{G^{(m)}(1-\lambda^{(m)})/|u^{(m)}v^{(m)}|_{\infty}}\geq\frac{Z|u^{(m)}v^{(m)}|_{\infty}}{G^{(m)}(1-\lambda^{(m)})},
\end{equation*}
which leads that
\begin{equation}\label{geqinverse}
\mathbb{E}\left[\frac{1}{Z}\right]\geq \mathbb{E}\left[\frac{|u^{(m)}v^{(m)}|_{\infty}}{G^{(m)}(1-\lambda^{(m)})\widehat{\pi}_{g}^{(m)}}\right].
\end{equation}
Furthermore, by applying Lemma~\ref{inverse_lemma}, we have
\begin{equation}\label{leqinverse}
\mathbb{E}\left[\frac{1}{Z}\right]\leq\frac{|u^{(m)}v^{(m)}|_{\infty}}{b\left(|u^{(m)}v^{(m)}|_{\infty}+\sum_{h\in\mathcal{H}_{0}^{(m)},h\ne g}u_h^{(m)}v_h^{(m)}\right)}.
\end{equation}
By putting Equation~\eqref{geqinverse} and \eqref{leqinverse} together, we have
\begin{equation}\label{adaptivity_bound}
\mathbb{E}\left[\frac{1}{\widehat{\pi}_{-g}^{(m)}}\right]\leq \frac{G^{(m)}(1-\lambda^{(m)})}{b\left(|u^{(m)}v^{(m)}|_{\infty}+\sum_{h\in\mathcal{H}_{0}^{(m)},~h\ne g}u_h^{(m)}v_h^{(m)}\right)}
\leq \frac{G^{(m)}}{\sum_{h\in\mathcal{H}_0^{(m)}}u_h^{(m)}v_h^{(m)}}.
\end{equation}
Therefore, we have
\begin{equation}\label{mFDR}
\begin{split}
&\quad~\mathbb{E}\left[\text{FDP}_u^{(m)}(\widehat{k}^{(1)},\dots,\widehat{k}^{(M)}) \right]\\
&\leq \frac{\alpha^{(m)}}{G^{(m)}}\sum\limits_{g\in\mathcal{H}_0^{(m)}}u_g^{(m)}v_g^{(m)}\mathbb{E}\left[\frac{\,\mathbf{1}\!\left\{ e_g^{(m)} \geq \frac{\widehat{\pi}^{(m)}G^{(m)}}{v_g^{(m)}\alpha^{(m)}\widehat{k}^{(m)}} \right\}}{\widehat{\pi}_{-g}^{(m)}\frac{v_g^{(m)}\alpha^{(m)}\widehat{k}^{(m)}}{\widehat{\pi}^{(m)}G^{(m)}}}\right]\\
        &=\frac{\alpha^{(m)}}{G^{(m)}}\sum\limits_{g\in\mathcal{H}_0^{(m)}}u_g^{(m)}v_g^{(m)}\mathbb{E}\left[\frac{1}{\widehat{\pi}^{(m)}_{-g}}\mathbb{E}\left[\frac{\,\mathbf{1}\!\left\{ e_g^{(m)} \geq \frac{\widehat{\pi}^{(m)}G^{(m)}}{v_g^{(m)}\alpha^{(m)}\widehat{k}^{(m)}} \right\}}{\frac{v_g^{(m)}\alpha^{(m)}\widehat{k}^{(m)}}{\widehat{\pi}^{(m)}G^{(m)}}}\;\Bigg|\;e^{(m)}_{-g}\right]\right]\\
        &\leq \frac{\alpha^{(m)}}{G^{(m)}}\sum\limits_{g\in\mathcal{H}_0^{(m)}}u_g^{(m)}v_g^{(m)}\mathbb{E}\left[\frac{1}{\widehat{\pi}^{(m)}_{-g}}\mathbb{E}\left[e_g^{(m)}|e^{(m)}_{-g}\right]\right]\\
        &=\frac{\alpha^{(m)}}{G^{(m)}}\sum\limits_{g\in\mathcal{H}_0^{(m)}}u_g^{(m)}v_g^{(m)}\mathbb{E}\left[\frac{1}{\widehat{\pi}^{(m)}_{-g}}\mathbb{E}\left[e_g^{(m)}\right]\right]\\
        &\leq \frac{\alpha^{(m)}}{G^{(m)}}\sum\limits_{g\in\mathcal{H}_0^{(m)}}u_g^{(m)}v_g^{(m)}\mathbb{E}\left[e_g^{(m)}\right]\mathbb{E}\left[\frac{1}{\widehat{\pi}^{(m)}_{-g}}\right]\\
        &\leq \alpha^{(m)}\frac{\sum_{g\in\mathcal{H}_0^{(m)}}u_g^{(m)}v_g^{(m)}\mathbb{E}[e_g^{(m)}]}{\sum_{g\in\mathcal{H}_0^{(m)}}u_g^{(m)}v_g^{(m)}}.
\end{split}
\end{equation}
If $\{e_g^{(m)}\}$ are valid e-values, then based on \eqref{mFDR} we have $\text{FDR}_u^{(m)}\leq \alpha^{(m)}$.
If $\{e_g^{(m)}\}$ are asymptotic e-values, then  we have $\lim\sup_{n,N\to\infty}\text{FDR}_u^{(m)}\leq \alpha^{(m)}$.

\textit{Theorem statement $(b)$.} Based on Equation~\eqref{main_prof}, we have
\begin{align}\label{mfdr2}
\begin{split}
\mathbb{E}\left[\mathrm{FDP}_u^{(m)}(\widehat{k}^{(1)},\dots,\widehat{k}^{(M)})\right]
&\leq \frac{\alpha^{(m)}}{G^{(m)}}\sum\limits_{g\in\mathcal{H}_0^{(m)}}u_g^{(m)}v_g^{(m)}\mathbb{E}\left[\frac{\mathbf{1}\left\{e_g^{(m)}\geq \frac{G^{(m)}}{v_g^{(m)}\alpha^{(m)}\widehat{k}^{(m)}}\right\}}{v_g^{(m)}\alpha^{(m)}\widehat{k}^{(m)}/G^{(m)}}\right]\\
&\leq \frac{\alpha^{(m)}}{G^{(m)}}\sum\limits_{g\in \mathcal{H}_0^{(m)}}u_g^{(m)}v_g^{(m)}\mathbb{E}\left[e_g^{(m)}\right].
\end{split}
\end{align}
Based on \eqref{mfdr2}, if $\{e_g^{(m)}\}$ are relaxed e-values with priors and penalties, then $\text{FDR}_u^{(m)}\leq \alpha^{(m)}$.
If $\{e_g^{(m)}\}$ are asymptotic relaxed e-values with priors and penalties, then we have $\lim\sup_{n,G^{(m)}\to\infty}\text{FDR}_u^{(m)}\leq \alpha^{(m)}$.
This completes the proof.
\end{proof}

\section{Other specific generalized e-values}
In this section, we develop generalized e-values based on other base detection procedures including \cite{Xing2023} and \cite{wang2025adaptive}. We have slightly abused some notations, but this does not cause confusion.
\subsection{Gaussian mirror e-values}
We first develop asymptotic relaxed e-values based on the Gaussian Mirror (GM) method \cite{Xing2023}.

Suppose the linear regression model
\begin{equation*}
\mathbf{y}=\mathbf{X}\beta+\mathbf{\epsilon},
\end{equation*}
where $\mathbf{y}\in \mathbb{R}^{n}$ is the response vector, $\mathbf{X}=(\mathbf{X}_1,\dots,\mathbf{X}_N)\in\mathbb{R}^{n\times N}$ is the design matrix, 
$\mathbf{X}_j=(X_{1j},\dots,X_{nj})^{\top}$, and $\mathbf{\epsilon}=(\epsilon_1,\dots,\epsilon_n)$ is a vector of i.i.d. draws from $N(0,\sigma^2)$.
Let $\mathcal{H}_0=\{j\in[N]:\beta_j=0\}$ and $\mathcal{H}_1=\{j\in[N]:\beta_j\ne0\}$.
Assume that the design matrix is normalized such that $\sum_{i=1}^n X_{ij}=0$ and $||\mathbf{X}_j||_2=n$, $j\in[N]$.
We focus on low-dimensional settings with the OLS estimator.
For each $j\in[N]$, GM method constructs Gaussian Mirror by
\begin{equation*}
GM_j:=\mathbf{X}^{j}=(\mathbf{X}_j^{+},\mathbf{X}_j^{-},\mathbf{X}_{-j})=(\mathbf{X}_j+c_j\mathbf{z}_j,\mathbf{X}_j-c_j\mathbf{z}_j,\mathbf{X}_{-j}),
\end{equation*}
where $\mathbf{X}_{-j}=(\mathbf{X}_1,\dots,\mathbf{X}_{j-1},\mathbf{X}_{j+1},\dots,\mathbf{X}_N)$, $\mathbf{z}_j$ is an independently simulated from $N(\mathbf{0},\mathbf{I}_n)$, and
\[
c_j=\sqrt{\frac{\mathbf{x}_j^{\top}(\mathbf{I}_n-\mathbf{X}_{-j}(\mathbf{X}_{-j}^{\top}\mathbf{X}_{-j})^{-1}\mathbf{X}_{-j}^{\top})\mathbf{x}_{j}}{\mathbf{z}_{j}^{\top}(\mathbf{I}_n-\mathbf{X}_{-j}(\mathbf{X}_{-j}^{\top}\mathbf{X}_{-j})^{-1}\mathbf{X}_{-j}^{\top})\mathbf{z}_j}}.
\]
Then GM method computes the coefficients by 
\begin{equation*}
\widehat{\mathbf{\beta}}^{j}=\arg\min\limits_{\mathbf{\beta}_j=(\beta_j^{+},\beta_{j}^{-},\mathbf{\beta}_{-j})} ||\mathbf{y}-\mathbf{X}^{j}\mathbf{\beta}^{j}||_2^{2},
\end{equation*}
and constructs the mirror statistics for the $j$-th variable by 
\begin{equation}\label{GM statistics}
M_j=|\widehat{\beta}^{+}_j+\widehat{\beta}^{-}_j|-|\widehat{\beta}^{+}_j-\widehat{\beta}^{-}_j|.
\end{equation}
For any target FDR level $\alpha$, the selection set is given by $\mathcal{G}(\alpha)=\{j\in[N]:M_j\geq t_{\alpha}\}$, where
\begin{equation}\label{GM threshold}
t_{\alpha}=\min_{t}\left\{t>0: \widehat{\text{FDP}}=\frac{\#\{j:M_j<-t\}}{\#\{M_j> t\}\vee1 }\leq \alpha\right\}.
\end{equation}

The GM e-values are constructed by
\begin{equation}\label{GMe-value}
e_{j}=N\frac{\mathbb{I}\{M_j\geq t_{\alpha}\}}{\#\{j:M_{j}\leq-t_{\alpha}\} \vee \alpha},~j\in[N].
\end{equation}
The following proposition demonstrates that the GM e-values for low-dimensional settings are valid asymptotic relaxed e-values.
Analysis for high-dimensional settings are similar and we omit it here.

\begin{proposition}
For any given $\alpha\in(0,1)$, let $e_{j}$'s be the GM e-values calculated by Equation \eqref{GMe-value}. Suppose that there exists a constant $c\in(0,1)$ such that $N=[cn]$ and
\[
\sum_{j,k\in \mathcal{H}_0}\frac{|\Omega_{jk}^{0}|}{\sqrt{\Omega_{jj}^{0}}\sqrt{\Omega_{kk}^{0}}}<C_1 |\mathcal{H}_0|^{\lambda}
\]
hods for absolute constant $C_1>0$ and $\lambda\in(0,2)$, where $\Omega^{0}=(X^{\top}_{[\mathcal{H}_0]}X_{[\mathcal{H}_0]})^{-1}$.
Further suppose that there exists $\tau_{\alpha}>0$ such that $\mathbb{P}(\text{FDP}(\tau_{\alpha}) \leq \alpha)\to 1$ as $p\to\infty$, and $\sum_{j\in\mathcal{H}_0}e_{j}/N$ is uniformly integrable. Then we have
\[
\lim\sup_{N\to\infty}\frac{1}{N}\sum_{j\in\mathcal{H}_0}\mathbb{E}[e_j]\leq 1.
\]
\end{proposition}

\begin{proof}
We first introduce several lemma from \cite{Xing2023}.
The remaining proof is consistent with Theorem 7.
\begin{lemma}
Let $M_j$'s be the mirror statistics defined by \eqref{GM statistics}, then
\[
\mathbb{P}(M_j \leq -t|\mathbf{z}_j)=\mathbb{P}(M_j \geq t|\mathbf{z}_j), \forall t>0,
\]
for $j\in\mathcal{H}_0$.
\end{lemma}

\begin{lemma}
Let $M_j$'s be the mirror statistics defined by \eqref{GM statistics}. Suppose there exists a constant $c\in(0,1)$ such that $N=[cn]$ and
\[
\sum_{j,k\in \mathcal{H}_0}\frac{|\Omega_{jk}^{0}|}{\sqrt{\Omega_{jj}^{0}}\sqrt{\Omega_{kk}^{0}}}<C_1 |\mathcal{H}_0|^{\lambda}
\]
hods for absolute constant $C_1>0$ and $\lambda\in(0,2)$.
Then 
\[
\sum_{j,k\in\mathcal{H}_0}\text{Cov}(\mathbb{I}\{M_j\geq t\},\mathbb{I}\{M_k\geq t\})\leq C_1^{'}|\mathcal{H}_0|^{\lambda_1}, \forall t,
\]
where $\lambda_1\in(0,2)$, $C_1^{'}>0$ is a constant.
\end{lemma}
Note that $\widehat{\beta}^{j}=((\mathbf{X}^j)^{\top}\mathbf{X}^{j})^{-1}(\mathbf{X}^j)^{\top}\mathbf{y}$. Therefore, $M_j$'s are continuous random variables and $\text{Var}(M_j)$ is uniformly upper bounded and also lower bounded away from zero.
Lemma F.2 and Lemma F.3 implies that the Assumption 1 and Assumption 2 for the Theorem 7 holds. 
Therefore, by the same argument of Theorem 7, it holds that 
\[
\lim\sup_{N\to\infty}\frac{1}{N}\sum_{j\in\mathcal{H}_0}\mathbb{E}[e_j]\leq 1.
\]
\end{proof}

\subsection{Asymptotic relaxed e-values for two-dimensional statistics}
In this subsection, we develop asymptotic relaxed e-values based on the \textit{symmetry-based adaptive selection} (SAS) procedure~\cite{wang2025adaptive}, which takes the symmetry property of the two-dimensional statistics associated with the null features to determine the rejections, with the aim of improving power over approaches based on one-dimensional mirror statistic. 

We first recall the SAS framework~\citep{wang2025adaptive}. For the $j$-th variable, consider that we have a two dimensional statistic $T_j=(T_j^{(1)},T_j^{(2)})$.
Consider the the random mixture model~\citep{efron2004large}:
\begin{equation}\label{mix_model}
T_j \sim f_n(t) = (1-\pi_1) f_{0,n}(t) + \pi_1 f_{1,n}(t),\qquad 
t = (t^{(1)}, t^{(2)}) \in \mathbb{R}^2,
\end{equation}
where $\pi_1$ is the prior probability that $j\in \mathcal{H}_1$, and $f_{0,n}(t)>0$ and $f_{1,n}$ are continuous densities for null and non-null, respectively.
The SAS framework first gives an estimate of the local false discovery rate~\citep{efron2004large} $\text{Lfdr}(t)$, where
\[
\text{Lfdr}(t)=\mathbb{P}(j\in\mathcal{H}_0|T_{j}=t).
\]
Specifically, under the model \eqref{mix_model}, the local false discovery rate is 
\begin{equation*}
\text{Lfdr}(t)=\frac{1-\pi_1f_{0,n}(t)}{f_n(t)}:=\frac{f_{0,n}^{+}(t)}{f_n(t)}.
\end{equation*}
An kernel density estimator~\citep{wand1994kernel} of the denominator $f_{n}(t)$ is $\widehat{f}_{n}(t)=N^{-1}\sum_{j=1}^N K_{\text{H}}(t-T_j)$, where $K$ is a symmetric kernel function, $\text{H}$ is a bandwidth matrix, and $K_{\text{H}}(t)=|\text{H}|^{-1/2}K(\text{H}^{-1/2}t)$.
The $\widehat{f}_n(t)$ is then refined by
\begin{equation}\label{denominator}
\tilde{f}_n(t)=\frac{\widehat{f}_n(t^{(1)},t^{(2)})+\widehat{f}_n(t^{(2)},t^{(1)})}{2}
\end{equation}
when $T_j^{(1)}$ and $T_j^{(2)}$ are exchangeable.
Before introducing an estimate of $f_{0,n}^{+}(t)$, we introduce several definition and assumption used in \cite{zhao2025false}.
\begin{definition}[Asymptotic symmetry property of \cite{zhao2025false}]
For $j\in\mathcal{H}_0$, $T_j$ is asymptotically symmetric with respect to a hyperplane $H\subset \mathbb{R}^{2}$ if the density $f_{0,n}(t)$ of $T_j$ satisfies $\lim_{n\to\infty} (f_{0,n}(t)-f_{0,n}(t^{'}))=0$ where $t\in \mathbb{R}^2$ and $t^{'}$ is the symmetry point of $t$ to $H$.
\end{definition}
\begin{assumption}\label{ass1}
For $j\in\mathcal{H}_0$, $T_{j}$ fulfills the asymptotic symmetry property with respect to $H_1=\{(t^{(1)},t^{(2)}):t^{(1)}=0\}$ and $H_2=\{(t^{(1)},t^{(2)}):t^{(2)}=0\}$.
\end{assumption}
\begin{assumption}\label{ass2}
For any rejection region $G\subseteq \mathbb{R}^2$, denote $G_m^{'}$ as the symmetry region of $G$ with respect to $H_m$, $m=1,2$. The rejection region satisfies $G\cap(H_m\cup G_m^{'})=\emptyset$ and $\lim_{n\to \infty}f_{1,n}(t^{'})=0$ for $t^{'}\in G_m^{'}$.
\end{assumption}

Denote the $i$-th quadrant by $Q_i$, $i=1,2,3,4$.
Define $A_0=\{T_j:T_j\in Q_2\cup Q_4\}$ and $B_0=\{(-T_j^{(1)},T_j^{(2)}):T_j\in Q_2\cup Q_4\}$.
By leveraging the asymptotic symmetry property of $T_j$ with respect to $H_1$, the first estimate of $f_{0,n}(t)$ is given by 
\[
\widehat{f}_{0,n}^{(1)}(t)=\frac{\sum_{T_A\in A_0}(K_{\text{H}}(t-T_A)+K_{\text{H}}(t_1^{'}-T_A))}{2|A_0|},
\]
where $t_1^{'}$ is the symmetry point of $t$ about the hyperplane $H_1$.
Similarly, the second estimator of $f_{0,n}(t)$ is given by 
\[
\widehat{f}_{0,n}^{(2)}(t)=\frac{\sum_{T_A\in A_0}(K_{\text{H}}(t-T_A)+K_{\text{H}}(t_2^{'}-T_A))}{2|A_0|}.
\]
By Assumption~\ref{ass1}, $f_{0,n}(t)$ is estimated by
\[
\widehat{f}_{0,n}(t)=\frac{1}{4}(\widehat{f}^{(1)}_{0,n}(t)+\widehat{f}^{(2)}_{0,n}(t)+\widehat{f}^{(1)}_{0,n}(-t)+\widehat{f}^{(2)}_{0,n}(-t)).
\]
The proportion $\pi_1$ is estimated by $1-2|A_0|/N$. Thus, similar to \eqref{denominator}, a refined estimator of $f_{0,n}^{+}(t)$ with $t=(t^{(1)},t^{(2)})$ is given by
\begin{equation}\label{molecule}
\tilde{f}^{+}_{0,n}(t)=\frac{\widehat{f}^{+}_{0,n}(t^{(1)},t^{(2)})+\widehat{f}^{+}_{0,n}(t^{(2)},t^{(1)})}{2},
\end{equation}
where
\[
\widehat{f}^{+}_{0,n}(t)=\frac{\sum_{T_A\in A_0}K_{\text{H}}(t_1^{'}-T_A)+K_{\text{H}}(t-T_A)+K_{\text{H}}(t_2^{'}-T_A)+K_{\text{H}}(-t-T_A)}{2N},
\]
and $t_i^{'}$ are the symmetry points of $t$ about the hyperplanes $H_i$, $i=1,2$.
Denote $G_0=Q_1\cup Q_3$.
The form of $G(t)$ is then given by $G(b)=\{t\in G_0:\tilde{f}_{0,n}^{+}(t)/\tilde{f}_n(t)\leq b\}$.
The set of selected features with target FDR level $\alpha\in(0,1)$ is given by $\mathcal{G}=\{j\in[N]:T_j\in G(t_{\alpha})\}$, where
\begin{equation}\label{threshold_SAS}
t_{\alpha}=\sup\left\{b>0:\frac{1+2^{-1}\sum_{m=1}^2\#\{j:T_j\in G_m^{'}(b)\}}{\#\{j:T_j\in G(b)\} \vee 1}\leq \alpha\right\}.
\end{equation}

The asymptotic relaxed e-values for the SAS framework (called SAS e-values) are constructed by 
\begin{equation}\label{SAS e-value}
e_j=N\frac{\mathbb{I}\{T_j\in G(t_{\alpha})\}}{1+2^{-1}\sum_{m=1}^2\#\{j:T_j\in G_m^{'}(t_{\alpha})\}}, ~j\in[N].
\end{equation}
Let $\psi(t):\mathbb{R}^2\to[0,1]$ be the function that $\text{Lfdr}(t)$ converges to. Denote $\tilde{G}(b)=\{t:\psi(t)\leq b, t\in G_0\}$, $\mathcal{K}(b)=\{\tilde{G}(b),\tilde{G}_1^{'}(b),\tilde{G}_2^{'}(b)\}$, and $\tilde{G}^{-}(b)=G_0-\tilde{G}(b)$.
Before showing that the SAS e-values are asymptotic relaxed e-values, we introduce necessary assumptions needed by \cite{wang2025adaptive}. 
\begin{assumption}\label{ass3}
For any constant $0<c<1$, there exists a corresponding constant $0<c_1<c$ such that $\tilde{G}^{-}(c_1)$ is a subset of a compact set, $\sup_{t\in\tilde{G}^{-}(c_1)}|\text{Lfdr}(t)-\psi(t)|\overset{p}{\to} 0 $, and $\sup_{t\in\tilde{G}(c_1)}\text{Lfdr}(t)\leq c$.
\end{assumption}

\begin{assumption}\label{ass4}
There exists constants $c>0$ and $\tau\in(0,2)$ such that for any $b\in[0,1]$,
\[
\max_{A\in \mathcal{K}(b)}\text{Var}\left(\sum_{j\in\mathcal{H}_0}\mathbb{I}\{T_j\in A\}\right)\leq cN^{\tau}.
\]
\end{assumption}

\begin{assumption}\label{ass5}
Assume that (1) for all $n$, $f_{0,n}(t)\leq q(t)$ and $q(t)$ is integrable; (2) the $\psi(T_j)$ are continuous random variables and $\text{Var}(\psi(T_j))>0$; (3) for any target level $\alpha\in (0,1)$, there exists a constant $\tau_\alpha\geq 0$ such that $\mathbb{P}(\text{FDP}(G(\tau_{\alpha}))\leq \alpha)\to 1$ as $N\to \infty$.
\end{assumption}

\begin{proposition}\label{valid_SAS_e-values}
For any given $\alpha\in(0,1)$, let $e_{j}$'s be the GM e-values calculated by \eqref{SAS e-value}.
Suppose that $\sum_{j\in\mathcal{H}_0}e_{j}/N$ is uniformly integrable. Then under Assumption \ref{ass1}-\ref{ass5}, we have
\[
\lim\sup_{N\to\infty}\frac{1}{N}\sum_{j\in\mathcal{H}_0}\mathbb{E}[e_j]\leq 1.
\]
\end{proposition}
\begin{proof}
The main proof structure is similar to the proof of Theorem 7.
Under Assumption~\ref{ass1}-\ref{ass5}, by the proof of Theorem 1 in \cite{wang2025adaptive}, it holds that for any fixed $\epsilon\in(0,\alpha)$, $\mathbb{P}(t_{\alpha}\geq\tau_{\alpha-\epsilon})\geq 1-\epsilon$ for enough large $N$.
Let
\begin{align*}
&\widehat{R}^{0}_{N}(b) = N^{-1} \sum_{j=1}^{N} \mathbb{I}\!\left\{T_j \in G(b),\, j \in \mathcal{H}_0 \right\},\quad
\widehat{V}^{0}_{N}(b) = (2N)^{-1} 
\sum_{m=1}^{2} \sum_{j=1}^{N} 
\mathbb{I}\!\left\{T_j \in G'_m(b),\, j \in \mathcal{H}_0 \right\},\\
&\widehat{V}^{1}_{N}(b) = (2N)^{-1} 
\sum_{m=1}^{2} \sum_{j=1}^{N} 
\mathbb{I}\!\left\{T_j \in G'_m(b),\, j \in \mathcal{H}_1 \right\}.
\end{align*}
Denote
\begin{align*}
&\
\mathcal{R}^{0}_{N}(b) 
= \mathbb{P}\!\left(T_j \in \widetilde{G}(b),\, j \in S_0 \right),\quad \widehat{\mathcal{R}}^{0}_{N}(b) 
= N^{-1} \sum_{j=1}^{N} \mathbb{I}\!\left\{T_j \in \widetilde{G}(b),\, j \in S_0 \right\}\\
&\mathcal{V}^{0}_{N}(b) = 2^{-1} \sum_{m=1}^{2} 
\mathbb{P}\!\left(T_j \in \widetilde{G}'_{m}(b),\, j \in S_0 \right),\quad\widehat{\mathcal{V}}^{0}_{N}(b) 
= (2N)^{-1} \sum_{m=1}^{2} \sum_{j=1}^{N} 
\mathbb{I}\!\left\{T_j \in \widetilde{G}'_{m}(b),\, j \in S_0 \right\}.
\end{align*}
We have
\begin{align*}
\frac{1}{N}\sum_{j=1}^N\mathbb{E}\left[e_{j}\mathbb{I}\{j\in\mathcal{H}_0\}\right] &=\mathbb{E}\left[\frac{\widehat{R}_N^0(t_{\alpha})}{\widehat{V}_N^{0}(t_{\alpha})+\widehat{V}_N^{1}(t_{\alpha})+1/N}\right]    \\
&\leq  \mathbb{E}\left[\left| \frac{\widehat{R}_N^0(t_{\alpha})}{\widehat{V}_N^{0}(t_{\alpha})+\widehat{V}_N^{1}(t_{\alpha})+1/N}
-
\frac{\mathcal{R}_N^{0}(t_{\alpha})}{\mathcal{R}_N^{0}(t_{\alpha})+\widehat{V}_N^{1}(t_{\alpha})+1/N} \right|  \right]
\\
&~~~~+\mathbb{E}\left[\left|
\frac{\mathcal{R}_N^{0}(t_{\alpha})}{\mathcal{R}_N^{0}(t_{\alpha})+\widehat{V}_N^{1}(t_{\alpha})+1/N}
-
\frac{\widehat{V}_N^{0}(t_{\alpha})}{\widehat{V}_N^{0}(t_{\alpha})+\widehat{V}_N^{1}(t_{\alpha})+1/N}
\right|   \right]  \\
&~~~~+\mathbb{E}\left[ \frac{\widehat{V}_N^{0}(t_{\alpha})}{\widehat{V}_N^{0}(t_{\alpha})+\widehat{V}_N^{1}(t_{\alpha})+1/N} \right]\\
&\leq 1+\mathbb{E}\left[\left| \frac{\widehat{R}_N^0(t_{\alpha})}{\widehat{V}_N^{0}(t_{\alpha})+\widehat{V}_N^{1}(t_{\alpha})+1/N}
-
\frac{\mathcal{R}_p^{0}(t_{\alpha})}{\mathcal{R}_p^{0}(t_{\alpha})+\widehat{V}_N^{1}(t_{\alpha})+1/N} \right|  \right]\\
&~~~~+\mathbb{E}\left[\left|
\frac{\mathcal{R}_N^{0}(t_{\alpha})}{\mathcal{R}_N^{0}(t_{\alpha})+\widehat{V}_N^{1}(t_{\alpha})+1/N}
-
\frac{\widehat{V}_N^{0}(t_{\alpha})}{\widehat{V}_N^{0}(t_{\alpha})+\widehat{V}_N^{1}(t_{\alpha})+1/N}
\right|   \right] .
\end{align*}
Furthermore, we have
\begin{align*}
&~~~~~\mathbb{E}\left[\left|
\frac{\mathcal{R}_N^{0}(t_{\alpha})}{\mathcal{R}_N^{0}(t_{\alpha})+\widehat{V}_N^{1}(t_{\alpha})+1/N}
-
\frac{\widehat{V}_N^{0}(t_{\alpha})}{\widehat{V}_N^{0}(t_{\alpha})+\widehat{V}_N^{1}(t_{\alpha})+1/N}
\right|   \right]\\
&\leq
\mathbb{E}\left[\sup_{\tau_{\alpha-\epsilon} \le t < 1}\left|
\frac{\mathcal{R}_N^{0}(t)}{\mathcal{R}_N^{0}(t)+\widehat{V}_N^{1}(t)+1/N}
-
\frac{\widehat{V}_N^{0}(t)}{\widehat{V}_N^{0}(t)+\widehat{V}_N^{1}(t)+1/N}
\right|    \right]+2\epsilon.
\end{align*}
By Lemma 2 and Lemma 3 of \cite{wang2025adaptive}, for any $0<c<1$, we have 
\begin{align*}
&
\sup_{b\in(0,1)} \left| \widehat{\mathcal{R}}^{0}_{N}(b) - \mathcal{R}^{0}_{N}(b) \right|
\ \xrightarrow{p}\ 0, ~
\sup_{b\in(0,1)} \left| \widehat{\mathcal{V}}^{0}_{N}(b) - \mathcal{R}^{0}_{N}(b) \right|
\ \xrightarrow{p}\ 0,\\
&\sup_{b \in (c,1)} \left| \widehat{R}^{0}_{N}(b) - \widehat{\mathcal{R}}^{0}_{N}(b) \right|
\ \xrightarrow{p}\ 0, ~
\sup_{b \in (c,1)} \left| \widehat{V}^{0}_{N}(b) - \widehat{\mathcal{V}}^{0}_{N}(b) \right|
\ \xrightarrow{p}\ 0,
\end{align*}
which leads to 
\begin{equation}\label{mirror_prop}
\sup_{b\in(c,1)}\left|\widehat{R}_N^{0}(b)-\mathcal{R}_N^{0}(b)\right|\xrightarrow{p}\ 0,\quad
\sup_{b\in(c,1)}\left|\widehat{V}_N^{0}(b)-\mathcal{R}_N^{0}(b)\right|\xrightarrow{p}\ 0.
\end{equation}
By the dominated convergence theorem, we have
\[
\lim\sup_{N\to\infty}\mathbb{E}\left[\sup_{\tau_{\alpha-\epsilon} \le t < 1}\left|
\frac{\mathcal{R}_N^{0}(t)}{\mathcal{R}_N^{0}(t)+\widehat{V}_N^{1}(t)+1/N}
-
\frac{\widehat{V}_N^{0}(t)}{\widehat{V}_N^{0}(t)+\widehat{V}_N^{1}(t)+1/N}
\right|    \right]=0.
\]
Similarly, we have
\begin{align*}
&~~~~~\mathbb{E}\left[\left| \frac{\widehat{R}_N^0(t_{\alpha})}{\widehat{V}_N^{0}(t_{\alpha})+\widehat{V}_N^{1}(t_{\alpha})+1/N}
-
\frac{\mathcal{R}_N^{0}(t_{\alpha})}{\mathcal{R}_N^{0}(t_{\alpha})+\widehat{V}_N^{1}(t_{\alpha})+1/N} \right|  \right]\\
&\le
\mathbb{E}\left[\sup_{\tau_{\alpha-\epsilon} \le t <1}\left| \frac{\widehat{R}_N^0(t)}{\widehat{V}_N^{0}(t)+\widehat{V}_N^{1}(t)+1/N}
-
\frac{\mathcal{R}_N^{0}(t)}{\mathcal{R}_N^{0}(t)+\widehat{V}_N^{1}(t)+1/N} \right|  \right]+C\epsilon,
\end{align*}
where $C>0$ is an absolute constant.
Since $\sum_{j=1}^N e_{j}\mathbb{I}\{j\in\mathcal{H}_0\}$ is uniformly integrable, by \eqref{mirror_prop}, we have
\[
\lim\sup_{N\to\infty}\mathbb{E}\left[\sup_{\tau_{\alpha-\epsilon} \le t <1}\left| \frac{\widehat{R}_N^0(t)}{\widehat{V}_N^{0}(t)+\widehat{V}_N^{1}(t)+1/N}
-
\frac{\mathcal{R}_N^{0}(t)}{\mathcal{R}_N^{0}(t)+\widehat{V}_N^{1}(t)+1/N} \right|  \right]=0.
\]
Putting these pieces together, we have
\[
\lim\sup_{N\to\infty}\frac{1}{N}\sum_{j=1}^N\mathbb{E}[e_j\mathbb{I}\{j\in\mathcal{H}_0\}]\le 1+(2+C)\epsilon.
\]
According to the arbitrariness of $\epsilon$, we have
\[
\limsup\limits_{N \to \infty}\frac{1}{N}\sum\limits_{j=1}^N\mathbb{E}\left[e_{j}\mathbb{I}\{j\in\mathcal{H}_0\}\right] \le 1.
\]
This completes the proof.
\end{proof}
\end{appendix}

\end{document}